\documentclass[letterpaper,11pt]{article}
\usepackage[utf8]{inputenc}
\usepackage[margin=1in]{geometry}
\usepackage{authblk}

\usepackage{amsmath}
\usepackage{amsthm}
\usepackage{amssymb}
\usepackage{thmtools}
\usepackage{thm-restate}
\usepackage{enumitem}
\usepackage{dsfont}
\usepackage{xspace}
\usepackage{comment}
\usepackage{xparse}
\usepackage{tikz}
\usepackage{fixme}
\usepackage[colorlinks]{hyperref}
\usepackage[capitalize]{cleveref}
\usepackage{subcaption}

\newtheorem{theorem}{Theorem}[]
\newtheorem{question}{Question}[]
\newtheorem{fact}{Fact}[]
\newtheorem{note}{Note}[]
\newtheorem{definition}{Definition}[]

\newtheorem{corollary}{Corollary}[]
\newtheorem{lemma}{Lemma}[]

\usepackage{algorithm}
\usepackage{algpseudocode}
\usepackage{algorithmicx}  
\usepackage{float}
\usepackage{microtype}
\usepackage{algorithm}
\usepackage{algpseudocode}
\usepackage{algorithmicx}   
\usepackage{floatflt}
\usepackage{graphics}
\usepackage{amsthm}

\title{Minimum+1 Steiner Cuts and Dual Edge Sensitivity Oracle: Bridging the Gap between Global cut and (s,t)-cut}

\author[1]{Koustav Bhanja}

\affil[1]{Indian Institute of Technology Kanpur, India}
\affil[ ]{\texttt{kbhanja@cse.iitk.ac.in}}
\date{}

\begin{document}

\begin{titlepage}

\maketitle

\begin{abstract}
Let $G=(V,E)$ be an undirected multi-graph on $n=|V|$ vertices and $S\subseteq V$ be a Steiner set in $G$. 
Steiner cut is a fundamental concept; moreover, global cut ($|S|=n$), as well as $(s,t)$-cut ($|S|=2$), is just a special case of Steiner cut. 
We study Steiner cuts of capacity minimum+1, and as an important application, we provide a dual edge Sensitivity Oracle for Steiner mincut -- a compact data structure for efficiently reporting a Steiner mincut after failure/insertion of any pair of edges.

A compact data structure for cuts of capacity minimum+1 has been designed for both global cuts [STOC 1995] and $(s,t)$-cuts [TALG 2023]. Moreover, both data structures are also used crucially to design a dual edge Sensitivity Oracle for their respective mincuts. 
Unfortunately, except for these two extreme scenarios of Steiner cuts, no generalization of these results is known. Therefore, to address this gap, we present the following first results on Steiner cuts for any $S$ with $2\le |S|\le n$.

\begin{enumerate}
    \item \textbf{Data Structure for Minimum+1 Steiner Cut:} \textit{There is an ${\mathcal O}(n(n-|S|+1))$ space data structure that can determine in ${\mathcal O}(1)$ time whether a given pair of vertices are separated by a Steiner cut of capacity at least minimum+1. 
    It can report such a cut, if it exists, in ${\mathcal O}(n)$ time, which is worst case optimal.}
    
    \item \textbf{Dual Edge Sensitivity Oracle:} We design the following pair of data structures.\\  
     \textbf{(a)} \textit{There is an ${\mathcal O}(n(n-|S|+1))$ space data structure that, after the failure or insertion of any pair of edges in $G$, can report the capacity of Steiner mincut in ${\mathcal O}(1)$ time and a Steiner mincut in ${\mathcal O}(n)$ time, which is worst case optimal.}\\
    \textbf{(b)} \textit{ If we are interested in reporting only the capacity of $S$-mincut, there is a more compact data structure that occupies ${\mathcal O}((n-|S|)^2+n)$ space and can report the capacity of Steiner mincut in ${\mathcal O}(1)$ time after the failure or insertion of any pair of edges.}

    \item \textbf{Lower Bound for Sensitivity Oracle:} \textit{For undirected multi-graphs, for every Steiner set $S\subseteq V$, any data structure that, after the failure or insertion of any pair of edges, can report the capacity of Steiner mincut must occupy $\Omega((n-|S|)^2)$ bits of space, irrespective of the query time.} 
\end{enumerate}

\noindent
To arrive at our results, we provide several techniques, especially a generalization of the \textsc{3-Star Lemma} given by Dinitz and Vainshtein [SICOMP 2000], which is of independent interest.

Our results achieve the same space and time bounds of the existing results for the two extreme scenarios of Steiner cuts -- global and $(s,t)$-cut. In addition, the space occupied by our data structures in (1) and (2) reduces as $|S|$ tends to $n$. Also, they occupy \textit{subquadratic} space if $|S|$ is \textit{close} to $n$. 
\end{abstract}
\end{titlepage}

\tableofcontents{}
\pagebreak
\pagenumbering{arabic}
\newpage
\section{Introduction}
\label{sec : introduction}
 In the real world, networks (or graphs) experience changes as a result of the failure/insertion of edges/vertices. They may lead to a change in the solution of various graph problems. Although these failures/insertions may occur at any point in the network at any time, they are mostly transient in nature. The aim is to quickly report the solution of the given problem once any failure/insertion has occurred in the network. This motivates us to design Sensitivity Oracles for various graph problems. A \textit{Sensitivity Oracle} for a graph problem is a compact data structure that can efficiently report the solution of the given problem after the failure/insertion of a set of edges/vertices. Sensitivity Oracles have been designed for many fundamental graph problems -- shortest paths  \cite{DBLP:journals/talg/BaswanaCHR20, DBLP:conf/stoc/BiloCCC0KS23}, reachability  \cite{DBLP:conf/soda/ItalianoKP21, DBLP:conf/stoc/BaswanaCR16}, traversals \cite{DBLP:conf/podc/Parter15}, etc. 
The (minimum) cut of a graph is also a fundamental concept consisting 
of numerous applications \cite{DBLP:books/daglib/0069809}. In this manuscript, for undirected multi-graphs, we present (1) a compact data structure for Steiner cuts of capacity minimum+1 and, as an important application, (2) a dual edge Sensitivity Oracle for Steiner mincut -- a data structure for handling insertion/failure of any pair of edges. The results in $(1)$ and $(2)$ provide the first generalization to the existing results for global mincut \cite{DBLP:conf/stoc/DinitzN95} and $(s,t)$-mincut  \cite{DBLP:journals/talg/BaswanaBP23} while matching their bounds on both space and time. In addition, our Sensitivity Oracle in (2) is the first result on Steiner mincut for handling multiple insertions/failures of edges after the single edge Sensitivity Oracles in undirected multi-graphs \cite{DBLP:conf/soda/BaswanaP22, DBLP:conf/soda/DinitzV95, DBLP:conf/stoc/DinitzV94, DBLP:journals/siamcomp/DinitzV00}. 

Let $G=(V,E)$ be an undirected connected multi-graph on $n=|V|$ vertices and $m=|E|$ edges ($E$ is a multi-set). There are mainly two types of cuts in a graph -- global cut and $(s,t)$-cut. 
A \textit{(global) cut} is a set of vertices $C\subset V$. Let $C$ be any cut in $G$. Cut $C$ is said to \textit{separate} a pair of vertices $u,v$ if $u\in C$ and $v\in \overline{C}$ or vice versa. For any pair of vertices $s,t$, cut $C$ is said to be an \textit{$(s,t)$-cut} if $C$ separates $s$ and $t$. An edge $e=(u,v)$ is a \textit{contributing edge} of $C$ if $C$ separates endpoints $u$ and $v$ of $e$.
Let $S\subseteq V$ be a \textit{Steiner set}. A generalization of global cut and $(s,t)$-cut is captured using Steiner cut and is formally defined as follows.
\begin{definition} [Steiner cut] \label{def : steiner cut}
    A cut $C$ is said to be a Steiner cut if there is a pair of vertices $s_1,s_2\in S$ such that $s_1\in C$ and $s_2\notin C$.
\end{definition}
The \textit{edge-set} of a cut $C$ is the set of contributing edges of $C$. The \textit{capacity} of a cut $C$, denoted by $c(C)$, is the number of edges belonging to the edge-set of $C$. A Steiner cut having the least capacity is known as a Steiner mincut of $G$, denoted by {\em $S$-mincut.} We denote the capacity of $S$-mincut by $\lambda_S$.  
By Definition \ref{def : steiner cut}, the two well-known minimum cuts of a graph are just two special cases of $S$-mincut, namely, \textit{global mincut} when $S=V$ and $(s,t)$-\textit{mincut} when $S=\{s,t\}$. 

There has been extensive research in designing compact data structures for minimum cuts \cite{gomory1961multi, dinitz1976structure, DBLP:journals/mp/PicardQ80, DBLP:journals/jacm/KawarabayashiT19, DBLP:journals/anor/ChengH91} as well as cuts of capacity near minimum \cite{DBLP:conf/stoc/DinitzN95, DBLP:conf/focs/Benczur95, DBLP:journals/jacm/KawarabayashiT19, DBLP:journals/talg/BaswanaBP23} 
that can efficiently answer the following fundamental query. \\ 

\noindent
   \textsc{cut($u,v, \kappa$)}: \textit{For any given pair of vertices $u,v$, determine if the least capacity cut separating $u$ and $v$ has capacity $\kappa$. Then, report a cut $C$ of capacity $\kappa$ separating $u,v$, if it exists.}\\



\noindent
Dinitz, Karzanov, and Lomonosov \cite{dinitz1976structure} designed an ${\mathcal O}(n)$ space cactus structure for storing all global mincuts in undirected weighted graphs. Likewise, Picard and Queyranne \cite{DBLP:journals/mp/PicardQ80} designed an ${\mathcal O}(m)$ space directed acyclic graph (dag) for storing all $(s,t)$-mincuts in directed weighted graphs. Both structures act as a data structure for answering query \textsc{cut} in ${\mathcal O}(n)$ time 
for their respective mincuts. 
It has been observed that a quite richer analysis was required to arrive at a compact structure for all $S$-mincuts. In a series of remarkable results by Dinitz and Vainshtein \cite{DBLP:conf/stoc/DinitzV94, DBLP:conf/soda/DinitzV95, DBLP:journals/siamcomp/DinitzV00}, they showed that there is an ${\mathcal O}(\min\{n\lambda_S,m\})$ space structure, known as \textit{Connectivity Carcass}, for storing all $S$-mincuts in an undirected multi-graph, which generalizes the results for $(s,t)$-mincut \cite{DBLP:journals/mp/PicardQ80} and global mincut \cite{dinitz1976structure}. They present an incremental algorithm for maintaining the structure until the capacity of $S$-mincut increases. This structure also helps in answering query \textsc{cut} for $S$-mincut in ${\mathcal O}(m)$ time and is the first single edge Sensitivity Oracle for $S$-mincut in undirected multi-graph. 

There also exist compact data structures for both global cut \cite{DBLP:conf/stoc/DinitzN95, DBLP:journals/jacm/KawarabayashiT19, DBLP:conf/focs/Benczur95} and $(s,t)$-cut \cite{DBLP:journals/talg/BaswanaBP23} of capacity beyond the minimum, and they are applied to establish many important algorithmic results \cite{DBLP:journals/jacm/KawarabayashiT19, DBLP:conf/focs/Benczur95, DBLP:conf/stoc/DinitzN95, DBLP:journals/talg/BaswanaBP23, DBLP:journals/talg/GoranciHT18}. Specifically for cuts of capacity minimum+1, Dinitz and Nutov \cite{DBLP:conf/stoc/DinitzN95} gave an ${\mathcal O}(n)$ space $2$-level cactus model that stores all global cuts of capacity minimum+1 in undirected multi-graphs. Their structure answers query \textsc{cut} in ${\mathcal O}(n)$ time and is used crucially for incremental maintenance of minimum+2 edge connected components \cite{DBLP:conf/stoc/DinitzN95}, which generalizes the results of Galil and Italiano \cite{DBLP:journals/siamcomp/GalilI93}, and Dinitz and Westbrook \cite{DBLP:conf/istcs/Dinitz93, DBLP:journals/algorithmica/DinitzW98}.
Similarly, for minimum+1 $(s,t)$-cuts in multi-graphs, Baswana, Bhanja, and Pandey \cite{DBLP:journals/talg/BaswanaBP23} designed an ${\mathcal O}(n^2)$ space data structure that answers query \textsc{cut} in ${\mathcal O}(n)$ time. This data structure acts as the key component in the design of a dual edge Sensitivity Oracle for $(s,t)$-mincut. 
Unfortunately, there does not exist any data structure for Steiner cuts of capacity beyond the minimum. Therefore, to bridge the gap between $(s,t)$-cut \cite{DBLP:journals/talg/BaswanaBP23} and global cut \cite{DBLP:conf/stoc/DinitzN95}, the following question is raised for every Steiner set $S$, $2\le |S| \le n$.
\begin{question} \label{ques : data structure}
    For any undirected multi-graph, does there exist a compact data structure that can efficiently answer query $\textsc{cut}$ for Steiner cuts of capacity minimum+1?
\end{question}
Observe that the failure of a pair of edges reduces the capacity of $S$-mincut if and only if at least one failed edge is contributing to an $S$-mincut or the pair of failed edges are contributing to a Steiner cut of capacity minimum+1. Therefore, the study of Steiner cuts of capacity minimum+1 has an important application to the problem of designing a \textit{dual edge Sensitivity Oracle for $S$-mincut} formally defined as follows.
\begin{definition} [Dual edge Sensitivity Oracle] \label{def : dual edge sensitivity oracle}
    A dual edge Sensitivity Oracle for $S$-mincut is a compact data structure that can efficiently report an $S$-mincut and its capacity after the failure or insertion of any pair of edges in $G$. 
\end{definition}
For $S$-mincuts, the existing Sensitivity Oracles can handle insertion/failure of only a single edge, and are based on the Connectivity Carcass of Dinitz and Vainshtein \cite{DBLP:conf/stoc/DinitzV94, DBLP:conf/soda/DinitzV95, DBLP:journals/siamcomp/DinitzV00}, which dates back to 2000. Baswana and Pandey \cite{DBLP:conf/soda/BaswanaP22}, using Connectivity Carcass as the foundation \cite{DBLP:conf/stoc/DinitzV94, DBLP:conf/soda/DinitzV95, DBLP:journals/siamcomp/DinitzV00}, designed an ${\mathcal O}(n)$ space single edge Sensitivity Oracle for $S$-mincut in undirected multi-graphs. It can report the capacity of $S$-mincut and an $S$-mincut for the resulting graph in ${\mathcal O}(1)$ and ${\mathcal O}(n)$ time respectively. 

It is also important to design Sensitivity Oracles that can handle the insertions/failures of multiple edges/vertices. For undirected multi-graphs, by using single edge Sensitivity Oracle for $S$-mincut \cite{DBLP:conf/soda/BaswanaP22}, observe that the trivial data structure that can handle any $f>0$ edge failures occupies ${\mathcal O}(m^{f-1}n)$ space. To design a Sensitivity Oracle for handling multiple edge insertions/failures, a natural and also quite commonly taken approach is to first design a Sensitivity Oracle that can handle the insertions/failures of a pair of edges. 
It has been observed that a dual edge Sensitivity Oracle for several fundamental graph problems, namely, reachability \cite{DBLP:conf/icalp/Choudhary16, DBLP:conf/icalp/ChakrabortyCC22}, breadth-first search \cite{DBLP:conf/podc/Parter15, DBLP:conf/icalp/GuptaK17}, shortest path \cite{DBLP:conf/soda/DuanP09a}, 
often requires a significantly richer analysis compared to the single edge Sensitivity Oracles. Moreover, it either reveals the difficulty or assists in several ways to solve the problem in its generality.
Specifically, in the area of minimum cuts, the following dual edge Sensitivity Oracle by Baswana, Bhanja, and Pandey \cite{DBLP:journals/talg/BaswanaBP23} is the only existing Sensitivity Oracle that can handle multiple insertions/failures. For $(s,t)$-mincut, there is an ${\mathcal O}(n^2)$ space data structure that, after the insertion/failure of any pair of edges from a multi-graph, can report an $(s,t)$-mincut and its capacity in ${\mathcal O}(n)$ and ${\mathcal O}(1)$ time respectively.  So, for $(s,t)$-mincut in multi-graphs, by designing a dual edge Sensitivity Oracle, they improved the space by a factor of $\Omega(\frac{m}{n})$ over the trivial data structure for handling any $f>0$ edge failures. The ${\mathcal O}(n)$ space 2-level cactus model of Dinitz and Nutov \cite{DBLP:conf/stoc/DinitzN95} for minimum+1 global cuts immediately leads to an ${\mathcal O}(n)$ space data structure for handling dual edge failure for global mincuts. It can report the global mincut and its capacity for the resulting graph in ${\mathcal O}(n)$ and ${\mathcal O}(1)$ time, respectively. Therefore, the existing results or the results that can be derived from the existing results are only for the two extreme cases of $S$-mincut. 
So, to provide a generalization of these results to any given Steiner set $S$, $2\le |S| \le n$, the following question arises naturally.
\begin{question} \label{ques : dual edge}
    For any undirected multi-graph, does there exist a compact data structure that can efficiently report the capacity of $S$-mincut and an $S$-mincut after the insertion or failure of any pair of edges in $G$?
\end{question}

\begin{note}
    The problem of handling dual edge insertions is provided in Appendix \ref{app : dual edge insertion}. Henceforth, we discuss only the failure of a pair of edges.
\end{note}





\subsection{Related Works:} 
Several data structures have been designed for cuts of capacity near minimum. Let $\lambda$ be the capacity of global mincut.  Bencz\'ur \cite{DBLP:conf/focs/Benczur95} provided an ${\mathcal O}(n^2)$ space data structure for cuts of capacity within $\frac{6}{5}\lambda$. Karger \cite{DBLP:journals/jacm/Karger00} designed an ${\mathcal O}(n^2)$ space data structure that improved the result of Benczur \cite{DBLP:conf/focs/Benczur95} from $\frac{6}{5}\lambda$ to $\frac{3}{2}\lambda$. Karger \cite{DBLP:journals/jacm/Karger00} also showed that there is an ${\mathcal O}(\alpha n^{\lfloor 2\alpha\rfloor})$ space data structure for cuts of capacity within $\alpha\lambda$.

\noindent
Single edge Sensitivity Oracle for $S$-mincut \cite{DBLP:conf/soda/BaswanaP22, DBLP:conf/stoc/DinitzV94, DBLP:conf/soda/DinitzV95, DBLP:journals/siamcomp/DinitzV00} not only includes $(s,t)$-mincut \cite{DBLP:journals/mp/PicardQ80, DBLP:journals/talg/BaswanaBP23} and global mincut \cite{dinitz1976structure} as special cases, but it also acts as the foundation of single edge Sensitivity Oracles for all-pairs mincut in undirected unweighted graphs \cite{DBLP:conf/soda/BaswanaP22}. 

\noindent
Providing a generalization from the two extreme scenarios of the Steiner set ($S=V$ and $|S|=2$) has also been addressed for various problems, namely, computing Steiner mincut \cite{DBLP:conf/stoc/DinitzV94, DBLP:journals/siamcomp/DinitzV00, DBLP:conf/stoc/ColeH03, DBLP:conf/soda/HeHS24, DBLP:journals/corr/abs-1912-11103}, Steiner connectivity augmentation and splitting-off \cite{DBLP:conf/soda/CenH0P23}, construction of a cactus graph for Steiner mincuts \cite{DBLP:journals/siamcomp/DinitzV00, DBLP:conf/soda/HeHS24}.



\subsection{Organization of this manuscript:} This manuscript is organized as follows. Our main results are stated in Section \ref{sec : our results}. Section \ref{sec : Preliminaries} contains basic preliminaries. A detailed overview of our results and techniques is provided in Section \ref{sec : overview}. The full version of our manuscript, containing all the omitted proofs, is given in the appendix. 
\section{Our Results} \label{sec : our results}
In this manuscript, we answer both Question \ref{ques : data structure} and Question \ref{ques : dual edge} in the affirmative. Moreover, the space and time bounds for our results also match with the existing result for the two extreme Scenarios of Steiner cuts -- global cut \cite{DBLP:conf/stoc/DinitzN95} ($|S|=n$) and $(s,t)$-cut \cite{DBLP:journals/talg/BaswanaBP23} $(|S|=2)$. 
Let us denote a Steiner cut of capacity minimum+1 by \textit{$(\lambda_S+1)$ cut}.

One of our key technical contributions is a generalization of the crucial \textsc{3-Star Lemma} from the Seminal work of Dinitz and Vainshtein \cite{DBLP:conf/stoc/DinitzV94, DBLP:conf/soda/DinitzV95, DBLP:journals/siamcomp/DinitzV00}. We believe that this result is of independent interest.  
By heavily exploiting the generalized version of \textsc{3-star Lemma}, we arrive at the following data structure for efficiently answering query \textsc{cut} for $(\lambda_S+1)$ cuts.
This result provides an affirmative answer to Question \ref{ques : data structure}. 

\begin{theorem}[Data Structure] \label{thm : data structure}
    Let $G=(V,E)$ be an undirected multi-graph on $n$ vertices and $S\subseteq V$ be any Steiner set of $G$. Let $\lambda_S$ be the capacity of $S$-mincut. There is an ${\mathcal O}(n(n-|S|+1))$ space data structure that, given any pair of vertices $u,v$ in $G$, can determine in ${\mathcal O}(1)$ time whether the least capacity Steiner cut separating $u$ and $v$ has capacity $\lambda_S+1$. Moreover, the data structure can report a $(\lambda_S+1)$ cut $C$ separating $u$ and $v$ in ${\mathcal O}(n)$ time, if it exists. 
\end{theorem}
Data structure in Theorem \ref{thm : data structure} occupies subquadratic, that is $o(n^2)$, space if $|S|=n-o(n)$.
For minimum+1 global cuts ($S=V$), it occupies only ${\mathcal O}(n)$ space, which matches with the $2$-level cactus model of Dintz and Nutov \cite{DBLP:conf/stoc/DinitzN95}. Moreover, for $S=\{s,t\}$ (the other extreme scenario), our data structure occupies ${\mathcal O}(n^2)$ space, which also matches with the existing best-known result on minimum+1 $(s,t)$-cut given by Baswana, Bhanja, and Pandey \cite{DBLP:journals/talg/BaswanaBP23}.  

For dual edge Sensitivity Oracle, we design a pair of data structures. Our first data structure helps in reporting an $S$-mincut in worst case optimal time; however, it requires ${\mathcal O}(|S|)$ time to report its capacity. If the aim is to report only the capacity of $S$-mincut, we present a more compact data structure that, interestingly, has a faster query time.
 These results on dual edge Sensitivity Oracle for $S$-mincut are formally stated in the following theorem, which answers Question \ref{ques : dual edge} in the affirmative.
\begin{theorem}[Sensitivity Oracle] \label{thm : dual edge failure}
     Let $G=(V,E)$ be an undirected multi-graph on $n=|V|$ vertices and $m=|E|$ edges. For any Steiner set $S\subseteq V$, 
     \begin{enumerate}
         \item  there is an ${\mathcal O}((n-|S|)^2+n)$ space data structure that can report the capacity of $S$-mincut for the resulting graph in ${\mathcal O}(1)$ time and 
         \item there is an ${\mathcal O}(n(n-|S|+1))$ space data structure that can report an $S$-mincut for the resulting graph in ${\mathcal O}(n)$ time
     \end{enumerate}
     after the insertion or failure of any pair of edges in $G$.
\end{theorem}
The dual edge Sensitivity Oracle in Theorem \ref{thm : dual edge failure} also achieves the same bound on space and query time for the existing best-known results on the two extreme scenarios of $S$-mincuts, namely, global mincut ($|S|=n$) \cite{DBLP:conf/stoc/DinitzN95} and $(s,t)$-mincut ($|S|=2$) \cite{DBLP:journals/talg/BaswanaBP23}.   

In contrast to Theorem \ref{thm : dual edge failure}, we provide the following lower bound for dual edge Sensitivity Oracle. This makes the data structures in Theorem \ref{thm : dual edge failure} tight almost for the whole range of $S$.
\begin{theorem} [Lower Bound] \label{thm : lower bound on dual edge failure}
    Let $G=(V,E)$ be a (un)directed multi-graph on $n$ vertices. For every Steiner set $S\subseteq V$, any data structure that, after the insertion or failure of any pair of edges in $G$, can report the capacity of $S$-mincut must occupy $\Omega((n-|S|)^2)$ bits of space in the worst case, irrespective of query time. 
\end{theorem}
For $S$-mincuts, the only existing lower bound for dual edge Sensitivity Oracle is for $|S|=2$ by Baswana, Bhanja, and Pandey \cite{DBLP:journals/talg/BaswanaBP23}. Moreover, it is conditioned on the Reachability Hypothesis \cite{DBLP:conf/wads/GoldsteinKLP17, DBLP:journals/siamcomp/Patrascu11}. Our lower bound in Theorem \ref{thm : lower bound on dual edge failure} not only generalizes their result from $|S|=2$ to any $S\subseteq V$, but it also does not depend on any hypothesis, hence, an unconditional lower bound.
\begin{note}
    We present the following comparison between the upper and lower bound of our dual edge Sensitivity Oracles. 
\begin{itemize}
     \item  The space occupied by our Sensitivity Oracles in Theorem \ref{thm : dual edge failure}(1) and Theorem \ref{thm : dual edge failure}(2) fail to match the lower bound in Theorem \ref{thm : lower bound on dual edge failure} if $|S|$ is only $n-o(\sqrt{n})$ and $n-o(n)$, respectively. 
    \item Observe that Sensitivity Oracles in Theorem \ref{thm : dual edge failure}(1) and Theorem \ref{thm : dual edge failure}(2) occupy space only linear in $n$ and subquadratic in $n$ if $|S|$ is only $n-o(\sqrt{n})$ and $n-o(n)$, respectively.
\end{itemize} 
\end{note}

All the results of this manuscript are compactly presented in Table \ref{tab : results}.

\section{Preliminaries} \label{sec : Preliminaries}
In this section, we define a set of notations and basic properties of cuts, which are to be used throughout this manuscript.
\begin{itemize}
    \item  A vertex $u$ is a \textit{Steiner vertex} if $u\in S$; otherwise $u$ is a \textit{nonSteiner} vertex.
    \item $G\setminus e$ denote the graph obtained from $G$ after the removal of edge $e$.
    \item A cut $C$ is said to \textit{subdivide} a set $X\subseteq V$ if both $C\cap X$ and $\overline{C}\cap X$ are nonempty.
    \item It follows from the definition of capacity of cuts that \textit{capacity of an empty set} is zero.
\end{itemize}



\begin{lemma}[Sub-modularity of Cuts \cite{DBLP:journals/dam/NagamochiI00}] \label{submodularity of cuts}
    For any two sets $A,B\subset V$,\\
    \text{$(1)$}  $c(A)+c(B)\ge c(A\cap B)+c(A\cup B)$ and \text{$(2)$} $c(A)+c(B)\ge c(A\setminus B)+c(B\setminus A)$.
\end{lemma}
The following lemma can be easily derived using Lemma \ref{submodularity of cuts}(1).
\begin{lemma} \label{lem : mincuts closed under int and uni}
    For a pair of $S$-mincut $C_1$ and $C_2$, if $C_1\cap C_2$ and $C_1\cup C_2$ are Steiner cuts, then both $C_1\cap C_2$ and $C_1\cup C_2$ are $S$-mincuts. 
\end{lemma}
 Dinitz and Vainshtein \cite{DBLP:conf/stoc/DinitzV94, DBLP:conf/soda/DinitzV95, DBLP:journals/siamcomp/DinitzV00} defined the following concept of \textit{$(\lambda_S+1)$-connectivity class}, which is also referred to as a $(\lambda_S+1)$ $(s,t)$-class for $S=\{s,t\}$ in \cite{DBLP:journals/talg/BaswanaBP23}. 
\begin{definition} \label{def : lambda+1 S-class}
    Let $R$ be an equivalence relation defined on the vertex set of $G$ as follows. A pair of vertices $u$ and $v$ are related by $R$ if and only if $u$ and $v$ are not separated by any $S$-mincut. Every equivalence class of relation $R$ is called a $(\lambda_S+1)$-connectivity class.
\end{definition}
For brevity we denote a $(\lambda_S+1)$-connectivity class by $(\lambda_S+1)$ class henceforth. A $(\lambda_S+1)$ class ${\mathcal W}$ is said to be a \textit{Singleton} $(\lambda_S+1)$ class if ${\mathcal W}$ has exactly one Steiner vertex of $G$. 

Let us define the following notion of crossing cuts initially introduced by Dinitz, Karzanov, and Lomonosov \cite{dinitz1976structure}, and the concept of nearest $(\lambda_S+1)$ cuts of a vertex. 
\begin{definition}[Crossing cuts and Corner sets]
    For a pair of cuts $A,B$ in $G$, each of the four sets, $A\cap B$, $A\setminus B$, $B\setminus A$, and $A\cup B$, is called a corner set of $A$ and $B$. Cuts $A$ and $B$ are said to be \textit{crossing} if each corner set is nonempty. 
\end{definition}
\begin{definition} [Nearest $(\lambda_S+1)$ cut of a vertex]  \label{def : nearest minimum+1 cut}
    For a vertex $u$, a $(\lambda_S+1)$ cut $C$ is said to be a nearest $(\lambda_S+1)$ cut of $u$ if $u\in C$ and there is no $(\lambda_S+1)$ cut $C'$ such that $u\in C'$ and $C'\subset C$. We denote the set of all nearest $(\lambda_S+1)$ cut of vertex $u$ by $N_S(u)$.
\end{definition}
   The \textit{nearest $(\lambda_S+1)$ cut from a vertex $u$ to a vertex $v$} is defined as a nearest $(\lambda_S+1)$ cut $C$ of $u$ such that $u\in C$ and $v\in \overline{C}$.

\section{An Overview of Our Results and Techniques} \label{sec : overview}
In this section, we present an overview of our results, including the limitations of the existing works and the key concepts/techniques established to arrive at our results.


\subsection{A Generalization of 3-Star Lemma}
Let $C_1,C_2,$ and $C_3$ be three Steiner cuts of $G$. Each of the three sets $C_1\cap \overline{C_2}\cap \overline{C_3}$, $\overline{C_1}\cap C_2 \cap \overline{C_3}$, $\overline{C_1}\cap \overline{C_2}\cap C_3$ is called a \textit{star} cut. For a pair of crossing $S$-mincuts, at least two of the four corner sets are also $S$-mincuts. 
Dinitz and Vainshtein \cite{DBLP:conf/stoc/DinitzV94, DBLP:conf/soda/DinitzV95, DBLP:journals/siamcomp/DinitzV00}, exploiting this property of $S$-mincuts, established the following result for the set of $S$-mincuts.  \\
\noindent
    \textsc{3-Star Lemma:} \textit{For any three $S$-mincuts $C_1,C_2$, and $C_3$, if each of the star cuts formed by $C_1,C_2$, and $C_3$ contains a vertex from $S$, then $c(C_1\cap C_2\cap C_3)=0$.}\\
\noindent
\textsc{3-Star Lemma} turned out to be an essential tool for designing the Connectivity Carcass \cite{DBLP:conf/stoc/DinitzV94, DBLP:conf/soda/DinitzV95, DBLP:journals/siamcomp/DinitzV00} for $S$-mincuts. Unfortunately, for a pair of crossing $(\lambda_S+1)$ cuts, it is quite possible that none of the four corner sets is a $(\lambda_S+1)$ cut (refer to cuts $C_1$ and $C_2$ in Figure \ref{fig : 1st figure in overview}($i$)).
As a result, \textsc{3-Star Lemma} no longer holds for the set of $(\lambda_S+1)$ cuts. 

Interestingly, to design a tool for addressing $(\lambda_S+1)$ cuts, exploiting the relation between a $(\lambda_S+1)$ class and $(\lambda_S+1)$ cuts, we generalize the \textsc{3-Star Lemma} for $(\lambda_S+1)$ cuts as follows (refer to Theorem \ref{thm : gen 3 star} in full version). 
\begin{theorem}[\textsc{Gen-3-Star Lemma}] For any three $(\lambda_S+1)$ cuts $C_1,C_2$, and $C_3$, if each of the star cuts formed by $C_1,C_2$, and $C_3$ contains a vertex from $S$, then $c(C_1\cap C_2\cap C_3)\le 3$. Moreover, for any $(\lambda_S+1)$ class ${\mathcal W}$, if exactly $k$ of the three star cuts formed by $C_1,C_2$, and $C_3$ subdivide ${\mathcal W}$, then $c(C_1\cap C_2\cap C_3)\le 3-k$.
\end{theorem}
    

\noindent
\textsc{Gen-3-Star Lemma} is used crucially throughout this manuscript for establishing several structural results for $(\lambda_S+1)$ cuts. This, in turn, leads to the design of a compact data structure for answering query \textsc{cut} for $(\lambda_S+1)$ cuts and the design of a dual edge Sensitivity Oracle for $S$-mincut. 
\subsection{A Data Structure for Reporting Minimum+1 Steiner Cuts}
Let $\lambda$ be the capacity of global mincut. For any pair of crossing global cuts $A,B$ of capacity $\lambda+1$, each of the four corner sets must have capacity at least $\lambda$ because it is also a global cut. For $\lambda\ge 3$, using this insight, Dinitz and Nutov \cite{DBLP:conf/stoc/DinitzN95} designed the ${\mathcal O}(n)$ space $2$-level cactus model for storing all global cuts of capacity $\lambda+1$. 
Unfortunately, the property no longer holds for $(\lambda_S+1)$ cuts because a corner set for a pair of crossing $(\lambda_S+1)$ cuts is not necessarily a Steiner cut; hence, a corner set can have capacity strictly less than $\lambda_S$. As a result, it does not seem possible to extend the approach of \cite{DBLP:conf/stoc/DinitzN95} for designing a data structure for $(\lambda_S+1)$ cuts.
On the other hand,
for $S=\{s,t\}$, since there are exactly two Steiner vertices, the intersection, as well as union, of a pair of $(\lambda_{\{s,t\}}+1)$ $(s,t)$-cut is always an $(s,t)$-cut, and has capacity at least $\lambda_{\{s,t\}}$. Baswana, Bhanja, and Pandey \cite{DBLP:journals/talg/BaswanaBP23} crucially exploit this fact and present the following result. 
 For any $(\lambda_{\{s,t\}}+1)$ class ${\mathcal W}$ of $G$, there is an ${\mathcal O}(|{\mathcal W}|^2+n)$ space data structure that answers query \textsc{cut} for any pair of vertices in ${\mathcal W}$ in ${\mathcal O}(n)$ time. For any Steiner set $S\subseteq V$, storing their data structure for every pair of vertices in $S$ results in an ${\mathcal O}(n^2|S|^2)$ space data structure, which can be $\Omega(n^4)$ for any $|S|=\Omega(n)$.

We now present the overview of our data structure in Theorem \ref{thm : data structure} that occupies only ${\mathcal O}(n(n-|S|+1))$ space and answers query \textsc{cut} for $(\lambda_S+1)$ cuts in ${\mathcal O}(n)$ time. Let $u$ and $v$ be any pair of vertices. By Definition \ref{def : lambda+1 S-class}, the cut of the least capacity that separates $u,v$ is $\lambda_S+1$ only if $u,v$ belong to the same $(\lambda_S+1)$ class of $G$.
Henceforth, to design a compact data structure for efficiently answering query \textsc{cut} for $(\lambda_S+1)$ cuts, we consider $(\lambda_S+1)$ classes of $G$. It turns out that the main challenge arises in handling the Singleton $(\lambda_S+1)$ classes of $G$. We first present a data structure for a Singleton $(\lambda_S+1)$ class. Later, to arrive at a compact data structure for any general $(\lambda_S+1)$ class (containing multiple or no vertex from $S$), we show that the Connectivity Carcass of Dinitz and Vainshtein \cite{DBLP:conf/stoc/DinitzV94, DBLP:conf/soda/DinitzV95, DBLP:journals/siamcomp/DinitzV00} can be suitably augmented with the data structure for Singleton $(\lambda_S+1)$ class. The procedure of augmentation requires establishing several new insights on $S$-mincuts, revealing close relationships between $S$-mincuts and $(\lambda_S+1)$ cuts, and a careful application of the \textit{covering technique} given by Baswana, Bhanja, and Pandey \cite{DBLP:journals/talg/BaswanaBP23} for $(s,t)$-cuts. 

  \begin{figure}
 \centering
    \includegraphics[width=\textwidth]{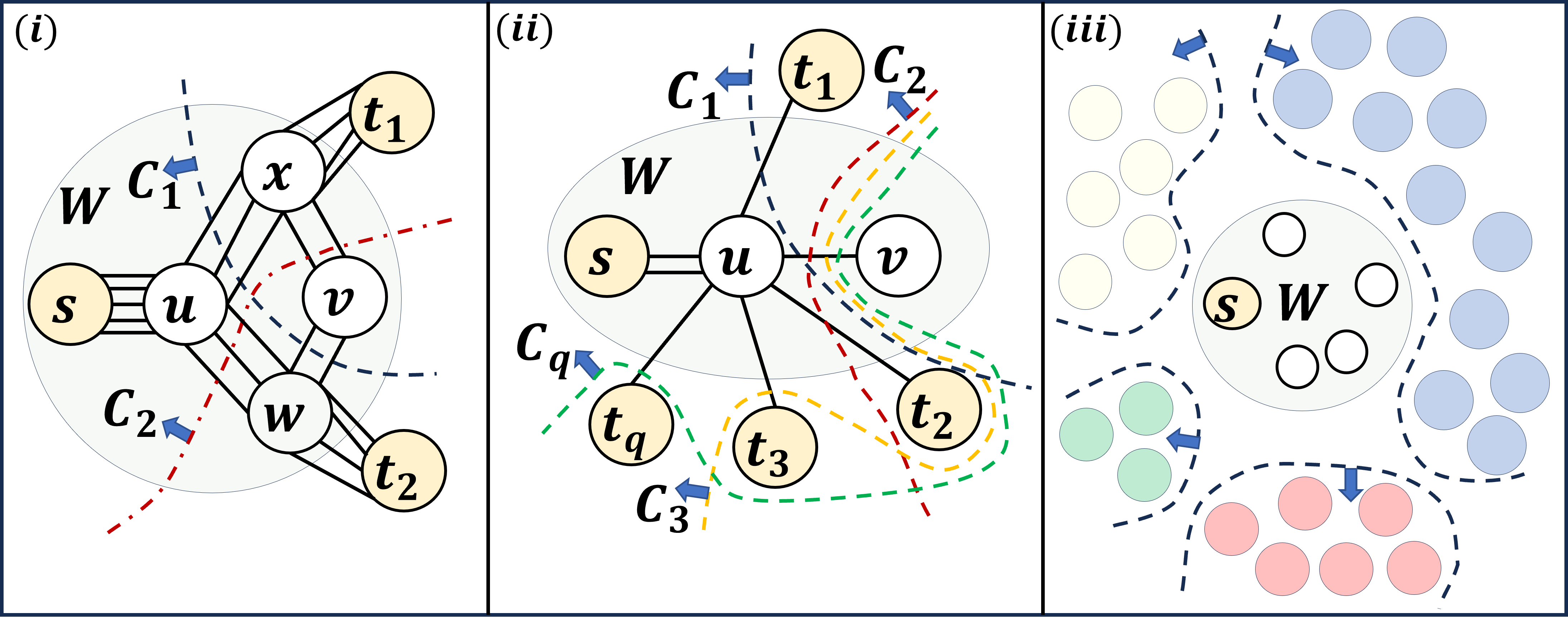} 
   \caption{The vertices 
   $\{s,t_1,t_2,\ldots t_q\}$ are Steiner vertices of $G({\mathcal W})$ and  $N_{S({\mathcal W})}(u)=\{C_1,C_2,\ldots, C_q\}$. $(i)$ Let $q=\Omega(|{\mathcal W}|)$. The capacity of global mincut of this graph is at least $4$ for $q\ge 2$. $C$ and $C'$ are Steiner cuts of capacity $2q+1$ and $2q$ respectively. The capacity of S-mincut is $2q$. Each cut $C_i$, $1\le i\le q$, has capacity $2q+1$; and moreover, $C_i$ has $t_2$ in $\overline{C_i}$. $(ii)$ Let $q=|S({\mathcal W})|$. For every pair of cuts $C_i,C_j$, $1\le i\ne j\le q$, from $N_{S({\mathcal W})}(u)$, a vertex $v$ from ${\mathcal W}$ belongs to $\overline{C_i\cup C_j}$. $(iii)$ Construction of $G({\mathcal W})$ from $G'$. The same color vertices not belonging to ${\mathcal W}$ are contracted into a Steiner vertex in $G({\mathcal W})$.}
  \label{fig : 1st figure in overview}
\end{figure}
\subsubsection{A Data Structure for a Singleton $(\lambda_S+1)$ class} \label{sec : overview : singleton data structure}
Suppose ${\mathcal W}$ is a Singleton $(\lambda_S+1)$ class and $s$ is the Steiner vertex belonging to ${\mathcal W}$. For any cut $C$, we assume that $s\in C$, otherwise consider $\overline{C}$. Let $C$ be a $(\lambda_S+1)$ cut that subdivides ${\mathcal W}$. Cut $C$ can cross multiple $S$-mincuts of $G$ and also contains vertices of $G$ that do not belong to ${\mathcal W}$. This makes it harder to analyze all the $(\lambda_S+1)$ cuts that subdivide ${\mathcal W}$. To circumvent this difficulty, we construct a \textit{smaller} graph $G({\mathcal W})$ from $G$ as follows. 
We first contract every $(\lambda_S+1)$ class, except ${\mathcal W}$, into a vertex. 
 Let $G'$ be the obtained graph. We construct graph $G({\mathcal W})$ from $G'$ as follows. Contract a pair of vertices $u,v$ if there is a Steiner mincut $C$ such that ${\mathcal W}\subseteq C$ and $u,v\in \overline{C}$. We repeat this contract operation in the resulting graph if it also has a Steiner mincut and a pair of vertices satisfying the same condition. 
The final obtained graph is $G({\mathcal W})$ (refer to Figure \ref{fig : 1st figure in overview}$(iii)$).
All vertices of $G({\mathcal W})$ that do not belong to ${\mathcal W}$, along with $s$, define the Steiner set of $G({\mathcal W})$, denoted by $S({\mathcal W})$.
Observe that in $G({\mathcal W})$, $\lambda_S$ is the capacity of Steiner mincut, and every Steiner mincut separates exactly one Steiner vertex from $s$.
It turns out that no $(\lambda_S+1)$ cut can cross a Steiner mincut in graph $G({\mathcal W})$ (refer to Figure \ref{fig : 1st figure in overview}). More importantly, we show that graph $G({\mathcal W})$ satisfies the following property.
\begin{lemma}[refer to Lemma \ref{lem : G(W)} in full version]
    A pair of vertices $u,v\in {\mathcal W}$ is separated by a $(\lambda_S+1)$ cut in $G$ if and only if there is a $(\lambda_S+1)$ cut in $G({\mathcal W})$ that separates $u,v$.
\end{lemma}


\noindent
Henceforth, we consider graph $G({\mathcal W})$ instead of $G$. 
By construction, $(\lambda_S+1)$ class ${\mathcal W}$ of $G$ also appears as a $(\lambda_S+1)$ class in $G({\mathcal W})$. To design a compact data structure for answering query \textsc{cut} for $(\lambda_S+1)$ cuts of $G({\mathcal W})$, we pursue an approach of compactly storing all the nearest $(\lambda_S+1)$ cuts of vertices in ${\mathcal W}$. This approach has been found to be quite useful for $S=\{s,t\}$ \cite{DBLP:journals/talg/BaswanaBP23}.  
 Let $u$ be a vertex in ${\mathcal W}$. Unlike for $|S|=2$ 
 \cite{DBLP:journals/talg/BaswanaBP23}, it turns out that the nearest $(\lambda_S+1)$ cut of $u$ to even a single vertex in ${\mathcal W}$ might not be unique. This is because the union of a pair of nearest $(\lambda_S+1)$ cuts of $u$ is not necessarily a Steiner cut (for example, cuts $C_1,C_2$ in Figure \ref{fig : 1st figure in overview}$(ii)$). 
To achieve our goal, by exploiting sub-modularity of cuts (Lemma \ref{submodularity of cuts}), we first establish the following property for $N_{S({\mathcal W})}(u)$ (refer to Lemma \ref{lem : pair of bunches} in full version).  
\begin{lemma} [\textsc{Property ${\mathcal P}_1$}]
    If the union of a pair of cuts $C_1,C_2\in N_{S({\mathcal W})}(u)$ is a Steiner cut, then no vertex of ${\mathcal W}$ can belong to ${\overline{C_1\cup C_2}}$.
\end{lemma}
 For $S=\{s,t\}$, the union of a pair of nearest $(\lambda_S+1)$ cuts is always a Steiner cut. Hence, every pair of nearest $(\lambda_S+1)$ cut satisfies \textsc{Property ${\mathcal P}_1$}.  
 This insight was used crucially for designing the compact data structure for minimum+1 $(s,t)$-cuts \cite{DBLP:journals/talg/BaswanaBP23}. However, as mentioned above, for any $S$, $2< |S|\le n$, for any pair of cuts from $N_{S({\mathcal W})}(u)$, their union might not be a Steiner cut. In fact, there can be $\Omega(|S({\mathcal W})|)$ number of cuts from $N_{S({\mathcal W})}(u)$ such that the complement of their union can contain a subset of vertices from ${\mathcal W}$, which can be $\Omega(n)$ in the worst case (refer to Figure \ref{fig : 1st figure in overview}$(ii)$). To overcome this difficulty, we address the problem in two different scenarios -- firstly, when $\lambda_{\mathcal W}\ge 4$, and secondly, when $\lambda_{\mathcal W} \le 3$, where $\lambda_{\mathcal W}$ is the global mincut capacity of $G({\mathcal W})$.

Suppose graph $G({\mathcal W})$ has global mincut capacity at least $4$. Exploiting \textsc{Gen-3-Star Lemma} and \textsc{Property ${\mathcal P}_1$}, we present the following generalization of \textsc{Property ${\mathcal P}_1$} for cuts of $N_{S({\mathcal W})}(u)$ even when their union is not a Steiner cut (refer to Lemma \ref{lem : three sets are nonempty} in full version). 

\begin{lemma} \label{lem : overview : three cuts intersection}
    If the capacity of global mincut of $G({\mathcal W})$ is at least $4$, then for any three cuts $C_1,C_2,C_3$ from $N_{S({\mathcal W})}(u)$, no vertex of ${\mathcal W}$ can belong to $\overline{C_1\cup C_2\cup C_3}$.
\end{lemma}
\noindent
Lemma \ref{lem : overview : three cuts intersection} essentially states that, for any vertex $u$ of ${\mathcal W}$, there are at most two cuts $C_1,C_2$ of $N_{S({\mathcal W})}(u)$ such that $u$ belongs to $\overline{C_1\cup C_2}$. Using this insight, we can store $\overline{C}\cap {\mathcal W}$ for all $C\in N_{S({\mathcal W})}(u)$ in only ${\mathcal O}(|{\mathcal W}|)$ space. Unfortunately, the problem still remains in compactly storing Steiner vertices $\overline{C}\cap S({\mathcal W})$ for all $C\in N_{S({\mathcal W})}(u)$. 

Observe that a Steiner vertex can appear in multiple cuts of $N_{S({\mathcal W})}(u)$ even if the global mincut capacity is at least $4$ (refer to Figure \ref{fig : 1st figure in overview}$(i)$). Therefore, trivially storing all cuts of $N_{S({\mathcal W})}(u)$ would occupy ${\mathcal O}(|{\mathcal W}||S({\mathcal W})|)$ space. Interestingly, we present an ${\mathcal O}(|{\mathcal W}|+|S({\mathcal W})|)$ space data structure that stores all cuts from $N_{S({\mathcal W})}(u)$.
We begin our analysis by classifying the set of all cuts of $N_{S({\mathcal W})}(u)$ into two sets as follows.\\

\noindent
\textbf{A Classification of all Cuts from $N_{S({\mathcal W})}(u)$:} A cut $C\in N_{S({\mathcal W})}(u)$ is said to be a \textit{$l$-cut} if it has exactly one Steiner vertex in $\overline{C}$; otherwise, it is called a $m$-cut.\\

\noindent
It follows that storing $\overline{C}$ for every $l$-cut $C$ of $N_{S({\mathcal W})}(u)$ occupies ${\mathcal O}(|{\mathcal W}|+|S({\mathcal W})|)$ space. However, the challenge arises in storing all $m$-cuts compactly. It seems quite possible that there are
\textit{many} Steiner vertices that belong to multiple $m$-cuts from $N_{S({\mathcal W})}(u)$. 
Exploiting the structural properties of $G({\mathcal W})$, we establish the following lemma for a pair of $m$-cuts from $N_{S({\mathcal W})}(u)$  (refer to Lemma \ref{lem : at most one bunch in common} in full version).
\begin{lemma}[\textsc{Property ${\mathcal P}_2$}] \label{lem : overview : property p2} For any pair of $m$-cuts $C_1,C_2\in N_{S({\mathcal W})}(u)$, there cannot exist more than one Steiner vertex in $\overline{C_1\cup C_2}$.
\end{lemma}
We design a bipartite graph $B$ using the set of $m$-cuts and the nonSteiner vertices belonging to ${\mathcal W}$. Exploiting the structure of $B$ and the well-known \textit{girth conjecture} by Erdos \cite{erdos1964extremal} and Bondy and Simonovits \cite{bondy1974cycles}, we show that \textsc{Property ${\mathcal P}_2$} can ensure an ${\mathcal O}(|{\mathcal W}|\sqrt{|S({\mathcal W})|})$ bound on space for storing all $m$-cuts of $N_{S({\mathcal W})}(u)$. 


To achieve a stricter bound on space, we further explore the relation of Steiner vertices with at least three $m$-cuts $C_1,C_2,$ and $C_3$ of $N_{S({\mathcal W})}(u)$.
Interestingly, by exploiting \textsc{Property ${\mathcal P}_1$},  \textsc{Property ${\mathcal P}_2$}, and the definition of $m$-cuts, we are able to show that every star cut of the three $m$-cuts $\overline{C_1}, \overline{C_2}$, and $\overline{C_3}$ is a Steiner cut that subdivides $(\lambda_S+1)$ class ${\mathcal W}$. Finally, by applying \textsc{Gen-3-Star Lemma}, 
the following crucial result for $m$-cuts is established, which acts as the key tool for storing all $m$-cuts compactly (refer to Lemma \ref{lem : a bunch belong to at most two cuts} in full version).  
\begin{lemma}
    For any three $m$-cuts $C_1,C_2,$ and $C_3$ from $N_{S({\mathcal W})}(u)$, there cannot exist any Steiner vertex $t$ in $S({\mathcal W})$ such that $t$ belongs to $\overline{C_1\cup C_2\cup C_3}$.
\end{lemma}
This result for $m$-cuts ensures that every Steiner vertex of $S({\mathcal W})$ appears in $\overline{C_1\cup C_2}$ for at most two $m$-cuts $C_1,C_2\in N_{S({\mathcal W})}(u)$. Therefore, it leads to a data structure occupying ${\mathcal O}(|{\mathcal W}|+|S({\mathcal W})|)$ space  for answering query \textsc{cut} for all cuts in $N_{S({\mathcal W})}(u)$. 

Now, we consider graph $G({\mathcal W})$ with global mincut capacity at most $3$. Using \textsc{Gen-3-Star Lemma}, we show that, for any three cuts from $N_{S({\mathcal W})}(u)$, at least one of the three star cuts of them is an empty set (refer to Lemma \ref{lem : storing at most two cuts are sufficient} in full version). This ensures that, for every vertex in ${\mathcal W}$, it is sufficient to store at most two cuts from $N_{S({\mathcal W})}(u)$. Therefore, ${\mathcal O}(|S({\mathcal W})|+|{\mathcal W}|)$ bound on space holds for this scenario as well.  

Finally, by storing this data structure on cuts from $N_{S({\mathcal W})}(u)$ for every $u\in {\mathcal W}$, we obtain a data structure ${\mathcal Q}_S({\mathcal W})$. This data structure occupies ${\mathcal O}(|{\mathcal W}|^2+|{\mathcal W}||S({\mathcal W})|)$ space and
answers query \textsc{cut} for $(\lambda_S+1)$ cuts of $G({\mathcal W})$.
Moreover, by storing an ${\mathcal O}(n)$ space mapping of vertices of $G$ to $G({\mathcal W})$, the data structure ${\mathcal Q}_S({\mathcal W})$ can report a cut $C$ in $G$ separating any given pair of vertices in ${\mathcal W}$ in ${\mathcal O}(n)$ time (formally stated in Theorem \ref{thm : data structure for singleton class}).
\begin{note} \label{note : k vertices}
    If $\lambda_{G({\mathcal W})}> 3$, then, given any set of vertices $\{u,v_1,v_2,\ldots, v_k\}\in {\mathcal W}$, the data structure ${\mathcal Q}_S({\mathcal W})$ can determine in ${\mathcal O}(k)$ time whether there is a cut $C\in N_{S({\mathcal W})}(u)$ such that $u\in C$ and $v_1\ldots,v_k\in \overline{C}$. However, if $\lambda_{G({\mathcal W})}\le 3$, then data structure ${\mathcal Q}_S({\mathcal W})$ can determine if only a pair of vertices from ${\mathcal W}$ are separated by a $(\lambda_S+1)$ cut.
\end{note}

 



\subsubsection{A Data Structure for General $(\lambda_S+1)$ classes}
In graph $G$, storing our data structure ${\mathcal Q}_S({\mathcal W})$ (designed in Section \ref{sec : overview : singleton data structure}) for every Singleton $(\lambda_S+1)$ class ${\mathcal W}$ occupies overall ${\mathcal O}(n(n-|S|+1))$ space.
However, a $(\lambda_S+1)$ class can also contain either no Steiner vertex or multiple Steiner vertices of $G$. For a $(\lambda_S+1)$ class ${\mathcal W}$ containing no Steiner vertex, we apply the \textit{Covering} technique of  \cite{DBLP:journals/talg/BaswanaBP23} to transform ${\mathcal W}$ into two Singleton $(\lambda_S+1)$ classes. We then construct the data structure ${\mathcal Q}_S$ for each of them. 

Suppose $(\lambda_S+1)$ class ${\mathcal W}$ contains multiple Steiner vertices. A $(\lambda_S+1)$ cut of $G({\mathcal W})$ either subdivides the Steiner set belonging to ${\mathcal W}$ or it does not. Using this insight, we construct the following pair of graphs $G({\mathcal W})_1$ and $G({\mathcal W})_2$ from graph $G$. Firstly, graph $G({\mathcal W})_1$ preserves each $(\lambda_S+1)$ cut of $G({\mathcal W})$ that does not subdivide the Steiner set belonging to ${\mathcal W}$. We show that the data structure ${\mathcal Q}_S$ is sufficient for graph $G({\mathcal W})_1$. Secondly, graph $G({\mathcal W})_2$ preserves the remaining $(\lambda_S+1)$ cuts of $G({\mathcal W})$. We ensure that, for graph $G({\mathcal W})_2$, the capacity of Steiner mincut is $\lambda_S+1$. So, we use a data structure for Steiner mincut given by Baswana and Pandey \cite{DBLP:conf/soda/BaswanaP22} for answering query \textsc{cut} in ${\mathcal O}(n)$ time for $G({\mathcal W})_2$ (refer to Theorem \ref{thm : reporting $S$-mincut using skeleton and projection mapping} in full version). We also show that the space occupied by the data structure for $G({\mathcal W})_2$ for all $(\lambda_S+1)$ classes ${\mathcal W}$ containing multiple Steiner vertices of $G$ occupies ${\mathcal O}(n)$ space. Therefore, given a pair of vertices $u,v$ in a $(\lambda_S+1)$ class ${\mathcal W}$, if ${\mathcal W}$ is not a Singleton $(\lambda_S+1)$ class, we query both the data structure -- one for $G({\mathcal W})_1$ and the other for $G({\mathcal W})_2$. The obtained data structure of this section is summarized in Theorem \ref{thm : data structure}.

\subsection{Dual Edge Sensitivity Oracle For S-mincut} \label{sec : dual failure overview}

We begin by providing a brief overview of the ${\mathcal O}(n)$ space single edge Sensitivity Oracle for $S$-mincut in unweighted undirected graphs \cite{DBLP:conf/soda/BaswanaP22}. 
An edge is said to \textit{belong to a $(\lambda_S+1)$ class} if both endpoints of the edge belong to the $(\lambda_S+1)$ class. By Definition \ref{def : lambda+1 S-class}, upon failure of an edge $e$, the capacity of $S$-mincut decreases by $1$ if and only if edge $e$ does not belong to any $(\lambda_S+1)$ class. 
Therefore, the ${\mathcal O}(n)$ space mapping of vertices to $(\lambda_S+1)$ classes helps in reporting the capacity of $S$-mincut in ${\mathcal O}(1)$ time after the failure of any edge. 
Baswana and Pandey \cite{DBLP:conf/soda/BaswanaP22} established an ordering among the $(\lambda_S+1)$ classes of $G$ that do not contain any Steiner vertex. They showed that this ordering, along with 
the Connectivity Carcass of Dinitz and Vainshtein \cite{DBLP:conf/stoc/DinitzV94, DBLP:conf/soda/DinitzV95, DBLP:journals/siamcomp/DinitzV00} (refer to Appendix \ref{sec : extended preliminaries} for details), is sufficient to design an ${\mathcal O}(n)$ space single edge Sensitivity Oracle for $S$-mincut for reporting an $S$-mincut in ${\mathcal O}(n)$ time (the result of Baswana and Pandey \cite{DBLP:conf/soda/BaswanaP22} is formally stated in Theorem \ref{thm : single edge Sensitivity Oracle}).  

In the case of two edge failures, both failed edges can belong to the same $(\lambda_S+1)$ class ${\mathcal W}$ and contribute to a single $(\lambda_S+1)$ cut that subdivides ${\mathcal W}$. In this case, the capacity of $S$-mincut definitely reduces by $1$. Unfortunately, the existing approach of designing a single edge Sensitivity Oracle does not reveal anything about the cuts that subdivide $(\lambda_S+1)$ classes. 
In the following, we present an overview of our dual edge Sensitivity Oracle from Theorem \ref{thm : dual edge failure}, which includes the establishment of several properties for $(\lambda_S+1)$ cuts and careful utilization of the data structure for Singleton $(\lambda_S+1)$ class (Theorem \ref{thm : data structure for singleton class}).  

\subsubsection{An $O(n(n-|S|+1))$ Space Data Structure with $O(|S|)$ Query Time} \label{sec : overview : dual edge first}
Let $e=(x,y)$ and $e'=(x',y')$ be the pair of failed edges in $G$. 
Depending on the relation between failed edges and $(\lambda_S+1)$ classes, the following four cases are possible.
\begin{itemize}
    \item Case 1: both failed edges belong to the same $(\lambda_S+1)$ class.
    \item Case 2: both failed edges belong to two different $(\lambda_S+1)$ classes.
    \item Case 3: Exactly one of the two failed edges belongs to a $(\lambda_S+1)$ class.
    \item Case 4: None of the failed edges belong to any $(\lambda_S+1)$ class.
\end{itemize}
An ${\mathcal O}(n)$ space mapping of vertices to $(\lambda_S+1)$ classes is sufficient to determine in ${\mathcal O}(1)$ time which of the four cases has occurred after the failure of edges $e,e'$. We handle each of the four possible cases as follows.

In Case 2, we show that the capacity of $S$-mincut never decreases after the failure of the two edges. In Case 3, by Definition \ref{def : lambda+1 S-class}, 
the failure of the edge that belongs to a $(\lambda_S+1)$ class cannot reduce the capacity of $S$-mincut. Therefore, this case is equivalent to handling the failure of the single edge that does not belong to any $(\lambda_S+1)$ class. Case 4 requires a richer analysis compared to Case 2 and Case 3. To handle this case, we establish new insights for $S$-mincuts, which help in showing that Connectivity Carcass of Dinitz and Vainshtein \cite{DBLP:conf/stoc/DinitzV94, DBLP:conf/soda/DinitzV95, DBLP:journals/siamcomp/DinitzV00}
is sufficient to design a data structure for handling dual edge failures. So, for all these three cases, we show that there is an ${\mathcal O}((n-|S|)^2+n)$ space data structure that, after the failure of any pair of edges,
can report an $S$-mincut and its capacity for the resulting graph in ${\mathcal O}(n)$ and ${\mathcal O}(1)$ time respectively. 
Now, the main hurdle comes in handling Case 1. 

In Case 1, both failed edges $e,e'$ belong to a $(\lambda_S+1)$ class ${\mathcal W}$. So, the capacity of $S$-mincut decreases by $1$ if and only if both failed edges are contributing to a single $(\lambda_S+1)$ cut that subdivides ${\mathcal W}$. We consider ${\mathcal W}$ to be a Singleton $(\lambda_S+1)$ class. In a similar way to the data structure in Theorem \ref{thm : data structure}, we also extend this result to any general $(\lambda_S+1)$ classes.\\

\noindent
\textbf{Handling Dual Edge Failure in a Singleton $(\lambda_S+1)$ class:} Let $s$ be the only Steiner vertex in ${\mathcal W}$. For every cut $C$, without loss of generality, assume that $s\in C$; otherwise, consider $\overline{C}$. Similar to the analysis of data structure in Theorem \ref{thm : data structure}, we show that it is sufficient to work with graph $G({\mathcal W})$ for ${\mathcal W}$ instead of $G$. 
The objective is to efficiently determine the existence of a $(\lambda_S+1)$ cut $C$ of $G({\mathcal W})$ such that both failed edges $e$ and $e'$ are contributing to $C$. The data structure in Theorem \ref{thm : data structure} can determine in ${\mathcal O}(1)$ time whether each of the failed edges is contributing to a $(\lambda_S+1)$ cut. Unfortunately, it fails to determine whether both $e,e'$ contribute to a single $(\lambda_S+1)$ cut. Suppose $(\lambda_S+1)$ cut $C$ exists and we assume, without loss of generality, that $x,x'\in C$ and $y,y'\in \overline{C}$.
The other three possible cases are handled in the same way.

 Let $C_1$ and $C_2$ be nearest $(\lambda_S+1)$ cuts from $x$ to $y$ and from $x'$ to $y'$, respectively. The following lemma states the necessary condition (refer to Lemma \ref{lem : new necessary condition} in full version). 
 \begin{lemma} [Necessary Condition] \label{lem : overview : necessary condition}
     If both edges $e,e'$ contribute to a single $(\lambda_S+1)$ cut, then $y'\notin C_1$ and $y\notin C_2$.
 \end{lemma}
 Let us assume that the capacity of global mincut for $G({\mathcal W})$ is at least $4$; later, we eliminate this assumption. The data structure in Theorem \ref{thm : data structure} can verify the necessary condition in Lemma \ref{lem : overview : necessary condition} in ${\mathcal O}(1)$ time (refer to Note \ref{note : k vertices}). 
 Suppose the conditions in Lemma \ref{lem : overview : necessary condition} holds for edges $e,e'$. Observe that $C_1\cap C_2$ is a cut that subdivides ${\mathcal W}$. Hence, it has capacity at least $\lambda_S+1$. For $S=\{s,t\}$, the union of cuts $C_1,C_2$ is always an $(s,t)$-cut and does not contain $y,y'$. Exploiting this fact, it is shown in \cite{DBLP:journals/talg/BaswanaBP23} that $C_1\cup C_2$ is an $(s,t)$-cut of capacity $\lambda_{\{s,t\}}+1$ in which both failed edges are contributing. Unfortunately, for any Steiner set $S\subseteq V$, there might not exist any Steiner vertex in $\overline{C_1\cup C_2}$. So, $C_1\cup C_2$ is not always a Steiner cut and can have capacity strictly less than $\lambda_S$.
 Moreover, even if $C_1\cup C_2$ is not a $(\lambda_S+1)$ cut, there might still exist a $(\lambda_S+1)$ cut in which both failed edges are contributing. Therefore, the data structure for nearest $(\lambda_S+1)$ cuts does not seem to work for answering dual edge failure queries. To overcome this hurdle, we establish the following crucial relation between vertices of ${\mathcal W}$ and a pair of crossing $(\lambda_S+1)$ cuts to which an edge belonging to ${\mathcal W}$ is contributing (refer to Lemma \ref{lem : property p3} in full version).
\begin{lemma}
    [\textsc{Property ${\mathcal P}_3$}] For any edge belonging to ${\mathcal W}$ and contributing to a pair of crossing $(\lambda_S+1)$ cuts $C$ and $C'$, neither $C\setminus C'$ nor $C'\setminus C$ contains any vertex from ${\mathcal W}$ if and only if each of the two sets $C\setminus C'$ and $C'\setminus C$ contains a Steiner vertex from $S({\mathcal W})$.
\end{lemma}
We exploit \textsc{Property ${\mathcal P}_3$} crucially to show that failed edges $e,e'$ cannot contribute to any single $(\lambda_S+1)$ cut if there is no Steiner vertex in $\overline{C_1\cup C_2}$. In other words, the existence of a Steiner vertex in $\overline{C_1\cup C_2}$ is indeed a sufficient condition, which leads to the following result (refer to Lemma \ref{lem : condition of dual failure} in full version).
\begin{lemma} \label{lem : conditions of dual edge failures}
    Both failed edges $e=(x,y)$ and $e'=(x',y')$ contribute to a single $(\lambda_S+1)$ cut if and only if there exist a nearest $(\lambda_S+1)$ cut $C_1$ from $x$ to $y$ and a nearest $(\lambda_S+1)$ cut $C_2$ from $x'$ to $y'$ such that $y'\notin C_1$, $y\notin C_2$, and there exists a Steiner vertex $t\in \overline{C_1\cup C_2}$.
\end{lemma}
It follows that if the capacity of global mincut of $G({\mathcal W})$ is at least $4$, the data structure in Theorem \ref{thm : data structure} is sufficient for verifying every condition of Lemma \ref{lem : conditions of dual edge failures}. 
Exploiting \textsc{Property ${\mathcal P}_1$}, \textsc{Property ${\mathcal P}_3$} and the following result, we show that data structure in Theorem \ref{thm : data structure} also works if the capacity of global mincut of $G({\mathcal W})$ is at most $3$ (refer to Lemma \ref{lem : necessary condition} in full version).
\begin{lemma}
    Let $(p,q)$ be any edge of $G({\mathcal W})$ and $u$ be a vertex such that $p,q,u\in {\mathcal W}$. Vertex $u$ belongs to a nearest $(\lambda_S+1)$ cut from $p$ to $q$ if and only if there is no $(\lambda_S+1)$ cut $C$ of $G({\mathcal W})$ in which edge $(p,q)$ contributes and $u\notin C$.
\end{lemma}
In conclusion, we have shown that for any capacity of global mincut of $G({\mathcal W})$, the data structure in Theorem \ref{thm : data structure} can determine the necessary condition in Lemma \ref{lem : overview : necessary condition} in ${\mathcal O}(1)$ time. To verify whether there is a Steiner vertex in $\overline{C_1\cup C_2}$, it requires ${\mathcal O}(|S|)$ time. Moreover, data structure in Theorem \ref{thm : data structure} can report $\overline{C_1}$ and $\overline{C_2}$ in ${\mathcal O}(n)$ time, which ensures that ${\overline{C_1}\cap \overline{C_2}}$ can be reported in ${\mathcal O}(n)$ time. Therefore, we have designed an ${\mathcal O}(n(n-|S|+1))$ space data structure that, after the failure of any pair of edges, can report the capacity of $S$-mincut and an $S$-mincut in ${\mathcal O}(|S|)$ time and ${\mathcal O}(n)$ time, respectively. This leads to Theorem \ref{thm : dual edge failure}(2).
\subsubsection{An $O((n-|S|)^2+n)$ Space Data Structure with $O(1)$ Query Time}
 After the failure of any pair of edges, the ${\mathcal O}(n(n-|S|+1))$ space data structure in Theorem \ref{thm : dual edge failure}(2) (designed in Section \ref{sec : overview : dual edge first}) takes ${\mathcal O}(|S|)$ time to report the capacity of $S$-mincut. For $S=V$, the query time is ${\mathcal O}(n)$, which is significantly inferior compared to the ${\mathcal O}(1)$ query time for reporting the capacity of global mincut after the failure of a pair of edges using the data structure in \cite{DBLP:conf/stoc/DinitzN95}. In this section, we present a data structure that is more compact than the data structure in Theorem \ref{thm : dual edge failure}(2). In addition, it can also report the capacity of $S$-mincut in ${\mathcal O}(1)$ time after the failure of a pair of edges; and hence, provides a generalization to the data structure of \cite{DBLP:conf/stoc/DinitzN95} from $S=V$ to any $S\subseteq V$. 


Recall that, only for handling Case 1 in Section \ref{sec : overview : dual edge first}, our data structure takes ${\mathcal O}(|S|)$ time to report the capacity of $S$-mincut. So, we consider the graph $G({\mathcal W})$ for a Singleton $(\lambda_S+1)$ class ${\mathcal W}$.  
Observe that the necessary condition (Lemma \ref{lem : overview : necessary condition}) can be verified in ${\mathcal O}(1)$ time. So, our objective is to determine in ${\mathcal O}(1)$ time whether there exists a Steiner vertex $t\in \overline{C_1\cup C_2}$ (the sufficient condition in Lemma \ref{lem : conditions of dual edge failures}).
For a single vertex $u\in {\mathcal W}$, \textsc{Property ${\mathcal P}_2$} (Lemma \ref{lem : overview : property p2}) ensures that there cannot be more than one Steiner vertex belonging to $\overline{C\cup C'}$, where $C,C'\in N_{S({\mathcal W})}(u)$. However, \textsc{Property ${\mathcal P}_2$} does not hold for a pair of vertices from ${\mathcal W}$. In other words, for a pair of vertices $u_1$ and $u_2$ in ${\mathcal W}$, there can be $\Omega(|S({\mathcal W})|)$ number of Steiner vertices that can belong to $\overline{C\cup C'}$, where $C\in N_{S({\mathcal W})}(u_1)$ and $C'\in N_{S({\mathcal W})}(u_2)$ (refer to Figure \ref{fig : 2nd figure in overview}). 
So, it might happen that the number of Steiner vertices belonging to each set $\overline{C_1}$ and $\overline{C_2}$ is ${\Omega}(|S({\mathcal W})|)$, but $(\overline{C_1\cup C_2})\cap S({\mathcal W})$ is ${\mathcal O}(1)$ only. Hence, $\Omega(|S({\mathcal W})|)$ time seems necessary to determine whether there is a Steiner vertex $t\in \overline{C_1\cup C_2}$, which is ${\mathcal O}(|S|)$ in the worst case.
In order to achieve ${\mathcal O}(1)$ query time, we provide a partition of the set of all cuts $\{N_{S({\mathcal W})}(u) ~|~ \forall u\in {\mathcal W}\}$ using the set of all $(\lambda_S+1)$ cuts and the Steiner set of $G({\mathcal W})$ as follows.

Let ${\mathcal C}_{\lambda_S+1}$ be the set of all $(\lambda_S+1)$ cuts of $G({\mathcal W})$ and let  ${\mathcal C}_s$ be the set of all cuts from $N_{S({\mathcal W})}(u)$ for all $u\in {\mathcal W}$. Our approach is to associate each cut of ${\mathcal C}_s$ with a \textit{small} set of Steiner vertices so that it helps in quickly determining whether a Steiner vertex is present in $ \overline{C_1\cup C_2}$. We show that there exists a subset $S^*$ of $S({\mathcal W})$ such that every cut in ${\mathcal C}_{\lambda_S+1}$ is associated with exactly one vertex from $S^*$ and satisfies the following interesting property 
\begin{lemma}[refer to Lemma \ref{lem : verifying third condition in constant} in full version]\label{lem : labeling of cuts}
Suppose $C_1$ and $C_2$ satisfy the necessary condition in Lemma \ref{lem : overview : necessary condition}, that is, $y'\notin C_1$ and $y\notin C_2$. Then, there is a Steiner vertex $t\in \overline{C_1\cup C_2}$ if and only if $C_1$ and $C_2$ are associated with the same Steiner vertex from $S^*$.
\end{lemma}
To verify the necessary condition in Lemma \ref{lem : overview : necessary condition}, we observe that the Steiner set $S({\mathcal W})$ is not required in the design of the data structure in Theorem \ref{thm : data structure}. This leads to an ${\mathcal O}((n-|S|)^2+n)$ space data structure that can answer the necessary condition in ${\mathcal O}(1)$ time.
Now, associating one vertex from $S^*$ to every cut in ${\mathcal C}_s$ occupies ${\mathcal O}((n-|S|)^2)$ space for all Singleton $(\lambda_S+1)$ classes of $G$. This helps in verifying the sufficient condition in ${\mathcal O}(1)$ time as well. Therefore, the capacity of $S$-mincut can be reported in ${\mathcal O}(1)$ time using an ${\mathcal O}((n-|S|)^2+n)$ space data structure after the failure of any pair of edges. This result, along with Theorem \ref{thm : dual edge failure}(2) and Theorem \ref{thm : dual edge insertion} in Appendix \ref{app : dual edge insertion} for handling insertions of any pair of edges, lead to Theorem \ref{thm : dual edge failure}. 



\subsection{Lower Bound for Dual Edge Sensitivity Oracle}
     It follows from Definition \ref{def : dual edge sensitivity oracle} that any dual edge Sensitivity Oracle for $S$-mincut answers the following query.\\
    \noindent
         \textsc{cap}$(e,e')$: \textit{Report the capacity of Steiner mincut after the failure of a pair of edges $e,e'$ in $G$.}\\     
     For $S=\{s,t\}$, Baswana, Bhanja, and Pandey \cite{DBLP:journals/talg/BaswanaBP23} established the following lower bound for answering query $\textsc{cap}$ conditioned on Reachability Hypothesis \cite{DBLP:conf/wads/GoldsteinKLP17, DBLP:journals/siamcomp/Patrascu11}. 
     For (un)directed multi-graph, any data structure that can answer query \textsc{cap} in ${\mathcal O}(1)$ time must occupy $\Tilde{ \Omega}(n^2)$ space in the worst case ($\Tilde{\Omega}$ hides polylogarithmic factors) unless Reachability Hypothesis \cite{DBLP:conf/wads/GoldsteinKLP17, DBLP:journals/siamcomp/Patrascu11} is violated.

      We first show that, for any given Steiner set $S\subseteq V$, any data structure that can answer query $\textsc{cap}$ in graph $G$ can also be used to answer query \textsc{cap} for $(s,t)$-mincut in a graph $G'=(V',E')$ on $|V'|={\mathcal O}(n-|S|)$ vertices, where $s,t\in V'$ are the only Steiner vertices of $G'$. This, along with the above-mentioned result of \cite{DBLP:journals/talg/BaswanaBP23}, can be used to establish the following conditional lower bound. For a (un)directed multi-graph, for every $S\subseteq V$, any data structure that can answer query cut in ${\mathcal O}(1)$ time must occupy $\Tilde{ \Omega}((n-|S|)^2)$ space in the worst case, unless Reachability Hypothesis \cite{DBLP:conf/wads/GoldsteinKLP17, DBLP:journals/siamcomp/Patrascu11} is violated. 
     
     Finally, and more importantly, we show by exploiting the structure of $G'$ that the Reachability Hypothesis \cite{DBLP:conf/wads/GoldsteinKLP17} is not required to establish an $\Omega((n-|S|)^2)$ lower bound on space for answering query \textsc{cap}. Hence, our obtained lower bound is unconditional. Moreover, the lower bound is also irrespective of the time taken for answering query \textsc{cap}. The result is formally stated in Theorem \ref{thm : lower bound on dual edge failure}.

\begin{table}[H]
\small
    \centering
    \begin{tabular}{|c|c|c|c|c|}
        \hline
         \textbf{Results} & Data Structure for & Query & Sensitivity Oracle & Sensitivity Oracle \\
         \textbf{} & Minimum+1 Cut (space) & \textsc{cut} (time) & for Reporting Capacity & for Reporting Cut \\
         \hline
           Global Cut \cite{DBLP:conf/stoc/DinitzN95}  &  &  & space: ${\mathcal O}(n)$ & space: ${\mathcal O}(n)$ \\
          ($|S|=n$) &  ${\mathcal O}(n)$ &  ${\mathcal O}(n)$ & time: ${\mathcal O}(1)$ & time: ${\mathcal O}(n)$\\
        
         \hline
         $(s,t)$-cut \cite{DBLP:journals/talg/BaswanaBP23} &   &  &  space: ${\mathcal O}(n^2)$ & space: ${\mathcal O}(n^2)$\\
          ($|S|=2$) & ${\mathcal O}(n^2)$ & ${\mathcal O}(n)$ & time: ${\mathcal O}(1)$ & time: ${\mathcal O}(n)$ \\
         \hline
        {\color{blue}\textbf{Steiner Cuts}} &  &  & {\color{blue}\textbf{space: $\mathbf{{\mathcal O}((n-|S|)^2+n)}$}} & {\color{blue}\textbf{space: $\mathbf{{\mathcal O}(n(n-|S|+1))}$}}  \\
        {\color{blue}\textbf{($2\le |S| \le n$)}} & {\color{blue} $\mathbf{{\mathcal O}(n(n-|S|+1))}$} & {\color{blue}$\mathbf{{\mathcal O}(n)}$} & {\color{blue}\textbf{time: $\mathbf{{\mathcal O}(1)}$}} & {\color{blue}\textbf{time:} $\mathbf{{\mathcal O}(n)}$}  \\
         \hline
    \end{tabular}    
    \caption{A comparison of our results with the existing results of extreme Scenarios of Steiner cut: global cut \cite{DBLP:conf/stoc/DinitzN95} and $(s,t)$-cut \cite{DBLP:journals/talg/BaswanaBP23}. Sensitivity Oracles are for handling insertion/failure of up to a pair of edges.}
    \label{tab : results}
\end{table}

\newpage
\section{Organization of the Full Version}
The full version of the manuscript is organized as follows. An extended preliminary consisting of an overview of Connectivity Carcass is provided in Appendix \ref{sec : extended preliminaries}. The generalization of $\textsc{3-Star-Lemma}$ is established in Appendix \ref{sec : 3 star lemma generalized}. A data structure for a Singleton $(\lambda_S+1)$ class is constructed in Appendix \ref{sec : data structure for singleton class}. Appendix \ref{sec : data structure complete} provides a data structure for answering query $\textsc{cut}$ for generic $(\lambda_S+1)$ classes. The dual edge Sensitivity Oracle is designed in Appendix \ref{sec : dual edge oracle} (dual edge failures) and Appendix \ref{app : dual edge insertion} (dual edge insertions). Appendix \ref{sec : lower bound oracle} contains the lower bound for dual edge Sensitivity Oracle. Finally, we conclude in Appendix \ref{sec : conclusion}. 

\section{Extended Preliminaries: An Overview of Connectivity Carcass} \label{sec : extended preliminaries}

In this section, we provide an overview of the ${\mathcal O}(\min\{n\lambda,m\})$ space Connectivity Carcass given by Dinitz and Vainshtein \cite{DBLP:conf/stoc/DinitzV94, DBLP:conf/soda/DinitzV95, DBLP:journals/siamcomp/DinitzV00}. 
Let ${\mathcal C}_S$ be the set of all $S$-mincuts of graph $G$.\\

\noindent
\textbf{Flesh Graph ${\mathcal F}_S$:} The Flesh graph is obtained from $G$ by contracting each $(\lambda_S+1)$ class of $G$ into a single node. It \textit{preserves} all the $S$-mincuts of $G$ and satisfies the following. A cut $C$ is an $S$-mincut in $G$ if and only if $C$ is a Steiner mincut in ${\mathcal F}_S$. We have a one-to-one mapping between nodes of ${\mathcal F}_S$ and $(\lambda_S+1)$ classes of $G$. \\


\noindent
\textbf{Quotient Mapping $\phi$:} The mapping of vertices from graph $G$ to nodes of Flesh graph ${\mathcal F}_S$ is called the quotient mapping and denoted by $\phi$.
 A node $\mu$ in ${\mathcal F}_S$ is called a \textit{Steiner node} if there exists at least one Steiner vertex $u$ in $G$ such that $\phi(u)=\mu$; otherwise, $\mu$ is called a \textit{nonSteiner node}. The following fact follows immediately from the construction of ${\mathcal F}_S$.
\begin{fact} \label{fact : quotient mapping}
    An edge $e=(u,v)$ in $G$ contributes to an $S$-mincut if and only if $\phi(u)\ne \phi(v)$. 
\end{fact}
We now define the following well-known cactus graph  and its minimal cuts.
\begin{definition}[Cactus Graph and its Minimal Cuts \cite{dinitz1976structure, DBLP:journals/siamcomp/DinitzV00}] \label{def : cactus and minimal cuts}
    An undirected graph is said to be a cactus graph if each edge in the graph belongs to at most one cycle. 
    An edge $e$ of a cactus is said to be a cycle-edge if $e$ is an edge of a cycle; otherwise, $e$ is called a tree-edge. A minimal cut of the cactus is either a tree-edge or a pair of cycle-edges from the same cycle.
\end{definition}
\begin{definition} [Proper Path in Cactus]
    A path in a cactus is said to be a proper path if every edge in the path belongs to at most one cycle in the cactus.
\end{definition}
The proper path between any pair of nodes $M, N$, denoted by $<M, N>$, in a cactus graph is unique, if it exists. \\

\noindent
\textbf{Skeleton ${\mathcal H}_S$:} It follows from Definition \ref{def : steiner cut} that every $S$-mincut $C$ partitions the Steiner set $S$ into two subsets $C\cap S$ and $\overline{C}\cap S$. There is an ${\mathcal O}(|S|)$ space cactus graph, known as \textit{Skeleton}, that compactly stores the Steiner partition formed by every $S$-mincut from ${\mathcal C}_S$. We denote the Skeleton by ${\mathcal H}_S$. \\

\noindent
Every Steiner node of Flesh graph is mapped to a unique node in the Skeleton. Note that there also exist nodes in the Skeleton to which no Steiner node is mapped. They are called the \textit{empty nodes} of the Skeleton. Skeleton satisfies the following property.  

\begin{lemma} [Lemma 8 and Theorem 3 of Section 2.4 in \cite{DBLP:conf/stoc/DinitzV94}]
    For any set $A\subseteq S$, $A$ is represented as a minimal cut in ${\mathcal H}_S$ if and only if there exists an $S$-mincut $C$ in $G$ such that $S\cap C=A$. 
    \label{lem : property of skeleton}
\end{lemma}
For any graph $H$, $\text{deg}(v)$ denote the number of edges adjacent to a vertex $v$ in $H$. Since the number of edges in ${\mathcal H}_S$ is ${\mathcal O}(|S|)$, the following lemma holds trivially for Skeleton.
\begin{lemma} \label{lem : sum of deg is order S}
    The sum of the degree of all nodes in ${\mathcal H}_S$ is ${\mathcal O}(|S|)$.
\end{lemma}

\noindent
We now introduce the following concept of $k$-junction, $k\ge 3$, which is used heavily in the analysis of the Sensitivity Oracle for dual edge insertions in Appendix \ref{app : dual edge insertion}.\\


\noindent
\textbf{$k$-junction:} A $k$-junction is defined on a cactus graph. A set of $k+1$, $k\ge 3$, nodes $\{u,u_1,u_2,\ldots,u_k\}$ is said to form a $k$-junction if there is a tree-edge between $(u_i,u)$. 
Node $u$ is called the \textit{core} node of the $k$-junction.\\

\noindent
Baswana and Pandey \cite{DBLP:conf/soda/BaswanaP22} introduced the following definition of intersection of two paths in a Skeleton.
\begin{definition} [Intersection of two paths (Definition 2.4 in \cite{DBLP:conf/soda/BaswanaP22})]
    A pair of paths $P_1$ and $P_2$ in Skeleton ${\mathcal H}_S$ are said to intersect if  
    \begin{itemize}
        \item there is a tree-edge in ${\mathcal H}_S$ belonging to both $P_1$ and $P_2$, or
        \item there is a cycle in ${\mathcal H}_S$ that shares edges with both $P_1$ and $P_2$
    \end{itemize}
\end{definition}
The following data structure is used to efficiently answer queries on intersections of a pair of paths. 
\begin{lemma} [Lemma 2.10 in \cite{DBLP:conf/soda/BaswanaP22}] \label{lem : lca queries on skeleton}
    There is an ${\mathcal O}(|S|)$ space data structure that, given any pair of paths $P_1$ and $P_2$ in ${\mathcal H}_S$, can determine in ${\mathcal O}(1)$ time whether $P_1$ intersects $P_2$. Moreover, if they intersect, the data structure can also report the intersection 
    (endpoints of an edge in common, endpoints of two cycle-edges in common, or a single node is in common)
    in ${\mathcal O}(1)$ time.      
\end{lemma}



\noindent
\textbf{Bunch:} Let $A\subset S$ and $(A,S\setminus A)$ be a partition of $S$ formed by a minimal cut of the Skeleton.  Let ${\mathcal B}$ be the set of all the $S$-mincuts $C$ with $C\cap S=A$. Set ${\mathcal B}$ is called a \textit{bunch} of graph $G$.  For a bunch ${\mathcal B}$, the corresponding \textit{Steiner partition} is denoted by $(S_{\mathcal B},S\setminus S_{\mathcal B})$.\\

\noindent
By Lemma \ref{lem : property of skeleton}, every minimal cut of the Skeleton corresponds to a bunch of $G$. The concept of a tight cut for a bunch is defined as follows. 
\begin{definition}[Tight cut for a bunch \cite{DBLP:journals/siamcomp/DinitzV00, DBLP:conf/soda/DinitzV95, DBLP:conf/stoc/DinitzV94}] \label{def : tight cut}
    An $S$-mincut $C$ belonging to a bunch ${\mathcal B}$ is said to be the tight cut of ${\mathcal B}$ from $S_{\mathcal B}$ to $S\setminus S_{\mathcal B}$ if $C\subseteq C'$ for every $S$-mincut $C'\in {\mathcal B}$ with $S_{\mathcal B}\subseteq C'$ and $S\setminus S_{\mathcal B} \subseteq \overline{C'}$. We denote this cut by $C(S_{\mathcal B})$.
\end{definition}
Let ${\mathcal H}_S^1$ and ${\mathcal H}_S^2$ be two subgraphs of Skeleton ${\mathcal H}_S$ formed after the removal of a tree-edge $(N,M)$ or a pair of cycle-edges $\{(N,M),(N',M')\}$ from a minimal cut $C$. Let ${\mathcal B}$ be the corresponding bunch of $C$. Without loss of generality, assume that for every vertex $s\in S_{\mathcal B}$, $\phi(s)$ is mapped to a node in ${\mathcal H}_S^1$. Without loss of generality, assume that node $N$ (likewise nodes $N$ and $N'$) belongs to ${\mathcal H}_S^1$, Then, we also denote the tight cut $C(S_{\mathcal B})$ by $C(N,(N,M))$ (likewise $C(N,N', (N,M), (N',M'))$). \\


\noindent
 \textbf{A Node of Flesh Distinguished by a Bunch:} For a node $\mu$ in Flesh graph, there may exist a bunch ${\mathcal B}$ such that $\mu$ appears in between two tight cuts of ${\mathcal B}$, that is, $\mu\in \overline{C(S_{\mathcal B})}\cap \overline{C(S_{\mathcal B})}$.  
Such a node $\mu$ is said to be \textit{distinguished} by bunch ${\mathcal B}$.  Depending on whether a node of the flesh graph is distinguished by a bunch or not, the set of nodes of Flesh is partitioned into two types of units defined as follows (refer to Section 2.3 in \cite{DBLP:conf/stoc/DinitzN95}).
\begin{definition} [Terminal unit and Stretched unit] \label{def : terminal and nonterminal unit}
    For any node $\mu$ in Flesh, if there exists a minimal cut in the Skeleton for which the corresponding bunch distinguishes $\mu$, then $\mu$ is called a \textit{Streched unit}; otherwise, $\mu$ is called a \textit{terminal unit}.
\end{definition}
It follows from Definition \ref{def : terminal and nonterminal unit} that every node of the Flesh graph is either a terminal unit or a Stretched unit. Observe that Steiner nodes cannot be distinguished by any bunch. Hence, they are always terminal units. Note that there might exist nonSteiner nodes that are not distinguished by any bunch; hence, by Definition \ref{def : terminal and nonterminal unit}, they are terminal units as well. \\  



\noindent
\textbf{Projection Mapping $\pi$:} Dinitz and Vainshtein \cite{DBLP:journals/siamcomp/DinitzV00} defines the following mapping, denoted by $\pi$,  between units of Flesh graph ${\mathcal F}_S$ and the Skeleton ${\mathcal H}_S$.
A unit $\mu$ of the Flesh is said to be projected to an edge $e$ in Skeleton, that is $e\in \pi(\mu)$, if one of the two holds -- (1) $e$ is a tree-edge and the bunch corresponding to the minimal cut defined by $e$ distinguishes $\mu$ or (2) $e$ is a cycle-edge of a cycle $O$ and, for every other edge $e'$ of cycle $O$, the bunch corresponding to the minimal cut defined by edges $e,e'$ distinguishes $\mu$.  
This mapping is called the \textit{projection mapping}. \\

\noindent
Every Stretched unit is projected to at least one edge in the Skeleton. Interestingly, the following property of projection mapping shows that for any stretched unit $\mu$, the edges belonging to $\pi(\mu)$ do not appear arbitrarily in the Skeleton. This property acts as the key tool for designing the compact structure \cite{DBLP:conf/soda/BaswanaP22} for efficiently answering single-edge sensitivity queries for Steiner mincut.
\begin{lemma} [Theorem 4 in \cite{DBLP:conf/stoc/DinitzV94} or Theorem 4.4 in \cite{DBLP:conf/soda/DinitzV95}] \label{lem : projection mapping of terminal and stretched units}
    For any node $\mu$ in ${\mathcal F}_S$, 
    \begin{enumerate}
        \item if $\mu$ is a terminal unit, then $\pi(\mu)$ is a unique node in ${\mathcal H}_S$. 
        \item if $\mu$ is a stretched unit, then $\pi(\mu)$ is a proper path in  ${\mathcal H}_S$. 
    \end{enumerate} 
\end{lemma}

\noindent
Since the proper path between a pair of nodes is unique, therefore, for any stretched unit $\mu$, it is sufficient to store only the endpoints of $\pi(\mu)$ (referred to as two coordinates of $\mu$ in Section 2.4 in \cite{DBLP:conf/stoc/DinitzV94}). This establishes the following lemma.
\begin{lemma} [\cite{DBLP:journals/siamcomp/DinitzV00, DBLP:conf/soda/DinitzV95, DBLP:conf/stoc/DinitzV94}] \label{lem : projection mapping}
    There is an ${\mathcal O}(n)$ space data structure that, given any unit $\mu$ of the Flesh graph, can report $\pi(\mu)$ in ${\mathcal O}(1)$ time. 
\end{lemma}
For any vertex $u$ of $G$, $\phi(u)$ is a unique unit in the Flesh. Therefore, without causing any ambiguity, we can say that a vertex $u$ is projected to a proper path $P$ in ${\mathcal H}_S$ if $\pi(\phi(u))=P$. \\

\noindent
\textbf{Projection of an edge:} Let $e_1$ be an edge in the Skeleton. Let ${\mathcal B}$ be a bunch corresponding to a minimal cut $C$ of the Skeleton such that $e_1$ belongs to the edge-set of $C$. An edge $e$ of $G$ is said to be \textit{projected} to edge $e_1$ if $e$ contributes to an $S$-mincut belonging to ${\mathcal B}$. Then, by extending Lemma \ref{lem : projection mapping of terminal and stretched units}, the following property of edges of the Flesh graph is established in \cite{DBLP:conf/soda/BaswanaP22}.
\begin{lemma}[Section 2.3 in \cite{DBLP:conf/soda/BaswanaP22}]  \label{lem : edge is projected to a proper path}
    For any edge $e=(u,v)$ in Flesh graph ${\mathcal F}_S$, $\pi(e)$ is a proper path in ${\mathcal H}_S$. Moreover, $\pi(e)$ is the proper path that has $\pi(\phi(u))$ as its prefix and $\pi(\phi(v))$ as its suffix or vice versa. 
\end{lemma}
Suppose an edge $e$ of ${\mathcal F}_S$ is projected to a cycle-edge $e_1$ of a cycle in Skeleton. In this case, for every minimal cut $C$ defined by edge $e_1$, edge $e$ contributes to at least one $S$-mincut belonging to the bunch corresponding to minimal cut $C$.\\

The following lemma helps in reporting a tight cut for a bunch. 
\begin{lemma} [Proposition 1 in  \cite{DBLP:conf/stoc/DinitzV94}] \label{lem : tight cut reporting}
    For any vertex $u\in V$, $\pi(\phi(u))$ is a subgraph of ${\mathcal H}_S^1$ if and only if $u\in C(S_{\mathcal B})$. 
\end{lemma}
It follows from Lemma \ref{lem : tight cut reporting} and Lemma \ref{lem : projection mapping} that, by querying the projection $\pi(\phi(u))$ of every vertex $u\in V$, it requires ${\mathcal O}(n)$ time to report $C(S_{\mathcal B})$ or $C(S\setminus S_{\mathcal B})$ for any bunch ${\mathcal B}$.

\noindent
Flesh, Skeleton, and Projection mapping together form the Connectivity Carcass.

\section{A Generalization of \textsc{3-Star Lemma}} \label{sec : 3 star lemma generalized}
In this section, we present a generalization of \textsc{3-Star Lemma} for $S$-mincuts given by Dinitz and Vainshtein \cite{DBLP:conf/stoc/DinitzV94, DBLP:conf/soda/DinitzV95, DBLP:journals/siamcomp/DinitzV00} to $(\lambda_S+1)$ cuts. This is used crucially to establish various structural and algorithmic results in the following sections. 

We begin by stating the following lemma, which can be seen as a generalization of the sub-modularity of cuts (Lemma \ref{submodularity of cuts}) property.
\begin{lemma} [Four Component Lemma \cite{DBLP:conf/focs/Benczur95}] \label{lem : four component lemma}
    For any three sets $A,B,C\subset V$, $c(A)+c(B)+c(C)\ge c(A\cap B\cap C)+c(\overline{A}\cap \overline{B}\cap C)+c(\overline{A}\cap B \cap \overline{C})+c(A\cap \overline{B}\cap \overline{C})$.
\end{lemma}
Let $C_1, C_2,$ and $C_3$ be any three $(\lambda_S+1)$ cuts of $G$ (refer to Figure \ref{fig : regions of three cuts}). These three cuts define the following eight disjoint subsets of $V$: $V_a=C_1\cap \overline{C_2}\cap \overline{C_3}$, $V_b=\overline{C_1}\cap C_2 \cap \overline{C_3}$, $V_c=\overline{C_1}\cap \overline{C_2}\cap C_3$, $V_{ab}=C_1\cap C_2\cap \overline{C_3}$, $V_{bc}=\overline{C_1}\cap C_2\cap C_3$, $V_{ac}=C_1\cap \overline{C_2}\cap C_3$, $V_{abc}=C_1\cap C_2\cap C_3$, and $V_{\Phi}=\overline{C_1\cup C_2\cup C_3}$.    
  \begin{figure}
 \centering
    \includegraphics[width=0.5\textwidth]{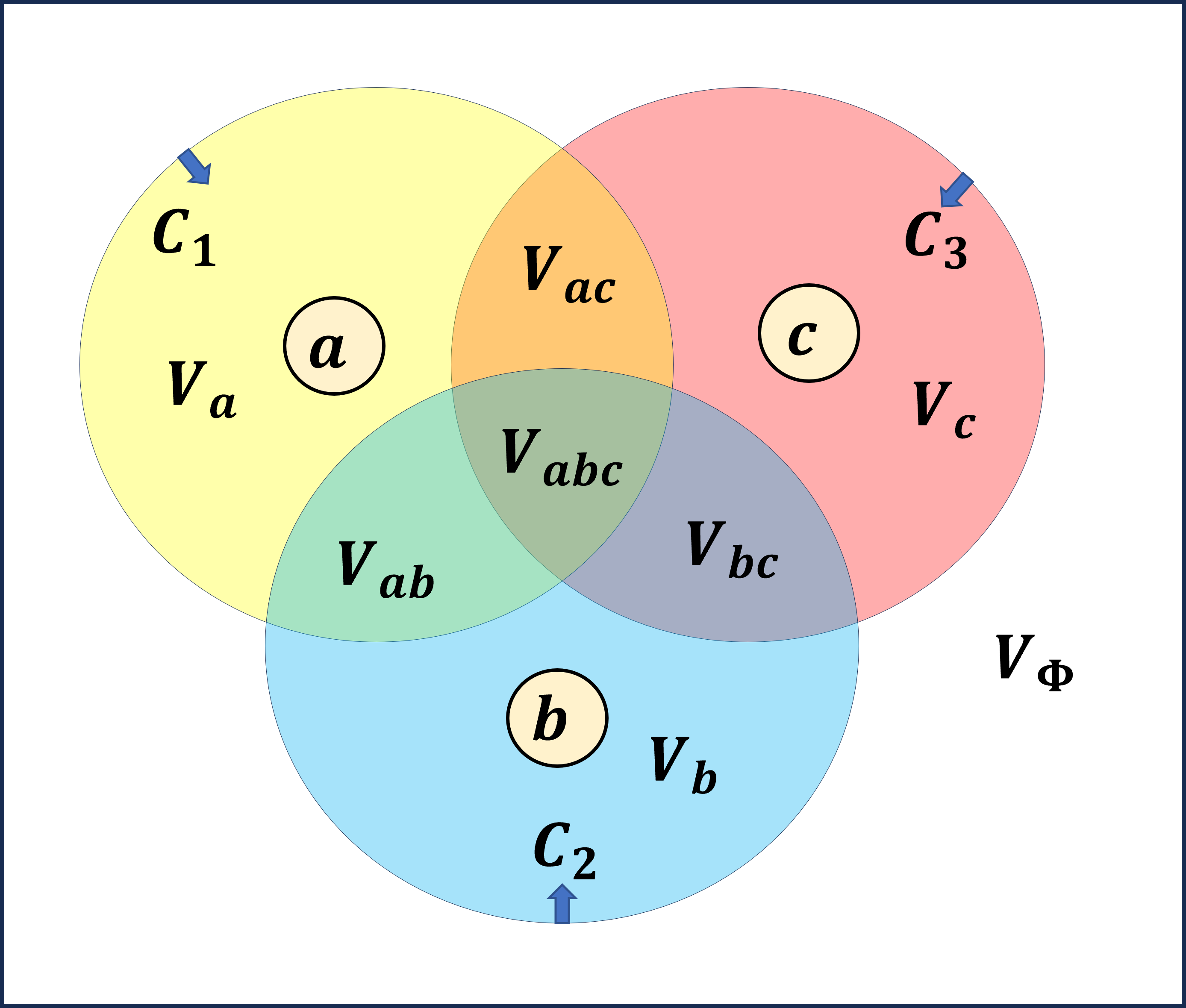} 
   \caption{Depicting all the eight regions formed by three $(\lambda_S+1)$ cuts $C_1,C_2,$ and $C_3$ of $G$.}
  \label{fig : regions of three cuts}. 
\end{figure}
\begin{theorem}[\textsc{Gen-3-Star Lemma}] \label{thm : gen 3 star} Suppose there are three Steiner vertices $a,b,c$ such that $a\in V_a$, $b\in V_b$, and $c\in V_c$.
Then, the following three properties hold.
    \begin{enumerate}
        \item $c(C_1\cap C_2 \cap C_3)\le 3$.
        \item For a $(\lambda_S+1)$ class ${\mathcal W}$ and an integer $k\in [0,3]$, if $(i)$ $u\in V_{\Phi}\cap {\mathcal W}$ 
        and $(ii)$ exactly $k$ of the three sets $V_a,V_b,$ and $V_c$ contains a vertex from ${\mathcal W}$, then $c(C_1\cap C_2\cap C_3)\le 3-k$. 
        \item Let $x,y,z\in \{a,b,c\}$ and $x\ne y\ne z$. 
        Vertices of $V_{xy}$ are adjacent to vertices of $V_{xy}, V_x$ and $V_y$. Moreover, there are at most two edges between $V_{xy}$ and $V\setminus (V_x\cup V_y\cup V_{xy})=V_z\cup V_{xz} \cup V_{yz}\cup V_{xyz}\cup V_{\Phi}$. 
    \end{enumerate}
\end{theorem}
\begin{proof}
Each of the three sets $V_a$, $V_b$, and $V_c$ contains a Steiner vertex. Moreover, observe that $V_a$, $V_b$, and $V_c$ are disjoint sets. 
Therefore, each of the three sets $V_a, V_b$, and $V_c$ defines a Steiner cut. 
This implies that each cut of the three cuts $V_a$, $V_b$, and $V_c$ has capacity at least $\lambda_S$. For cuts $C_1,C_2,$ and $C_3$, the following equation follows from Lemma \ref{lem : four component lemma}, 
\begin{equation} \label{eq : equation for 3 sets}
c(C_1)+c(C_2)+c(C_3)\ge c(C_1\cap C_2 \cap C_3) + c(V_a) + c(V_b) + c(V_c)
\end{equation} 
Since $c(C_i)= \lambda_S+1$ for each $i\in \{1,2,3\}$, it follows from Equation \ref{eq : equation for 3 sets} that $c(C_1\cap C_2 \cap C_3)\le 3$. This completes the proof of property (1). 

Suppose, without loss of generality, cut $V_a$ contains a vertex from ${\mathcal W}$. It is given that there is a vertex $u\in V_{\Phi}\cap {\mathcal W}$. Therefore, $V_a$ is a Steiner cut that subdivides the $(\lambda_S+1)$ class ${\mathcal W}$. Hence, the capacity of Steiner cut $V_a$ is at least $\lambda_S+1$. So, any $k\ge 0$ sets of the three sets, $V_a$, $V_b$, and $V_c$, containing a vertex from ${\mathcal W}$ has a capacity at least $\lambda_S+1$. This, using Equation \ref{eq : equation for 3 sets}, implies that $c(C_1\cap C_2\cap C_3)\le 3-k$. This completes the proof of property (2). The proof of property (3) is deferred to  Appendix \ref{app : gen 3 star extended} for better readability.
\end{proof}

We establish the following lemma as a corollary of Theorem \ref{thm : gen 3 star}(2).

\begin{lemma} \label{lem : cor of gen 3 star}
    For any three $(\lambda_S+1)$ cuts $C_1,C_2,$ and $C_3$, if each of the three sets $V_a, V_b,$ and $V_c$ contains a Steiner vertex and a vertex from the same $(\lambda_S+1)$ class ${\mathcal W}$ of $G$, then $C_1\cap C_2 \cap C_3=\emptyset$.
\end{lemma}
\begin{proof}
    Each of the three sets $\{V_a,V_b,V_c\}$ contains a Steiner vertex and a vertex from ${\mathcal W}$. Since $V_a,V_b,$ and $V_c$ are disjoint sets, therefore, each of them defines a Steiner cut and subdivides ${\mathcal W}$. Therefore, observe that for these three Steiner cuts $C_1,C_2,C_3$, Theorem \ref{thm : gen 3 star}(2) holds even if $C_1,C_2,C_3$ do not satisfy condition $(i)$ of Theorem \ref{thm : gen 3 star}(2) (there is no vertex $u\in V_{\Phi}\cap {\mathcal W}$). Now, by assigning the value of $k=3$ in Theorem \ref{thm : gen 3 star}(2), we get $c(C_1\cap C_2\cap C_3)=0$. Since graph $G$ is a connected graph, therefore, $C_1\cap C_2 \cap C_3$ is an empty set. 
\end{proof}

\section{A Data Structure for a Singleton $(\lambda_S+1)$ class} \label{sec : data structure for singleton class}
Let ${\mathcal W}$ be a Singleton $(\lambda_S+1)$ class that contains exactly one Steiner vertex $s\in S$. Without loss of generality, we assume that for every cut $C$, $s\in C$; otherwise, consider $\overline{C}$. In this section, we design a compact data structure that can answer the following query efficiently for any given pair of vertices $u,v\in {\mathcal W}$.
\begin{center}
    $\textsc{cut}_{\mathcal W}(u,v)\text{: report a $(\lambda_S+1)$ cut $C$ such that $s,u\in C$ and $v\in \overline{C}$, if exists.}$
\end{center}
\noindent
To design a compact data structure for answering query $\textsc{cut}_{\mathcal W}$, we aim to store all the nearest cuts $N_S(u)$ for every vertex $u\in {\mathcal W}$ (refer to Definition \ref{def : nearest minimum+1 cut}). Since $(\lambda_S+1)$ class ${\mathcal W}$ contains one Steiner vertex, it is represented by a Steiner node $\mu$ of the Flesh graph. 
By Lemma \ref{lem : projection mapping of terminal and stretched units}(1), let $N$ be the node of the Skeleton ${\mathcal H}_S$ such that $\pi(\mu)=N$. 
A node in a Skeleton can be adjacent to ${\mathcal O}(|S|)$ edges. Every tree-edge or pair of cycle-edges from the same cycle defines a minimal cut of ${\mathcal H}_S$ (refer to Definition \ref{def : cactus and minimal cuts}). By definition, for every minimal cut $\tilde C$ of ${\mathcal H}_S$, there is a bunch 
corresponding to $\tilde C$. Henceforth, without causing any ambiguity, we state that every node of Skeleton ${\mathcal H}_S$ is \textit{adjacent} to ${\mathcal O}(|S|)$ bunches. 

Suppose node $N$ is adjacent to exactly one bunch ${\mathcal B}$.
We establish the following property for $N_S(u)$ for every vertex $u\in {\mathcal W}$. 
\begin{lemma} \label{lem : for a single bunch data structure}
    Let $u$ be any vertex in ${\mathcal W}$. For a pair of cuts $C_1,C_2\in N_S(u)$, $\overline{C_1}\cap \overline{C_2}\cap {\mathcal W}=\phi$. 
\end{lemma}
\begin{proof}
    Recall that, for bunch ${\mathcal B}$, the corresponding partition of Steiner set is $(S_{\mathcal B},S\setminus S_{\mathcal B})$. Without loss of generality, assume that $s\in S_{\mathcal B}$. Let $t$ be a Steiner vertex belonging to $S\setminus S_{\mathcal B}$. Since the cuts $C_1$ and $C_2$ are Steiner cuts and ${\mathcal B}$ is the only bunch that is adjacent to $N$, so $s\in C_1\cap C_2$ and $t\notin C_1\cup C_2$. We now give a proof by contradiction. Assume to the contrary that there is a vertex $v\in {\mathcal W}$ such that $v\in \overline{C_1}\cap \overline{C_2}$. Since $s,u\in C_1\cap C_2$ and $v\notin C_1\cup C_2$, therefore, $C_1\cap C_2$ as well as $C_1\cup C_2$ subdivides $(\lambda_S+1)$ class ${\mathcal W}$. Hence, both $C_1\cap C_2$ and $C_1\cup C_2$ have capacity at least $\lambda_S+1$. By sub-modularity of cuts (Lemma \ref{submodularity of cuts}(1)), $c(C_1\cap C_2)+c(C_1\cup C_2)\le 2\lambda_S+2$. This implies that both $C_1\cap C_2$ and $C_1\cup C_2$ are $(\lambda_S+1)$ cuts. Therefore, we have a $(\lambda_S+1)$ cut $C_1\cap C_2$ with $u\in C_1\cap C_2$, which is a proper subset of $C_1$, $C_2$, or both $C_1$ and $C_2$, a contradiction.
\end{proof}
Let $u$ be any vertex in ${\mathcal W}$ and $N_S(u)=\{C_1,C_2,\ldots,C_k\}$. Lemma \ref{lem : for a single bunch data structure} ensures that $\overline{C_i}\cap {\mathcal W}$ for all $i\in [k]$ can be stored using ${\mathcal O}(|{\mathcal W}|)$ space. This leads to an ${\mathcal O}(|{\mathcal W}|^2)$ space data structure that, given any pair of vertices $u,v$ in ${\mathcal W}$, can report $\overline{C}\cap {\mathcal W}$ in ${\mathcal O}(|\overline{C}\cap {\mathcal W}|)$ time such that $s,u\in C$ and $v\in \overline{C}$. However, in general, as discussed above, node $N$ can be adjacent to ${\mathcal O}(|S|)$ bunches. 
Therefore, it requires ${\mathcal O}(|{\mathcal W}|^2|S|)$ space in the worst case. 

To achieve a more compact space, we explore the relation between the nearest $(\lambda_S+1)$ cuts of $u$ and the bunches adjacent to node $N$. Without loss of generality, we consider that for any bunch ${\mathcal B}$, $s\in S_{\mathcal B}$. A $(\lambda_S+1)$ cut $C$ is said to \textit{subdivide a bunch} ${\mathcal B}$ if $C$ subdivides $\overline{C(S_{\mathcal B})}$ ($C(S_{\mathcal B})$ is a \textit{tight cut} for bunch ${\mathcal B}$, defined in Definition \ref{def : tight cut} in Appendix \ref{sec : extended preliminaries}).  
A $(\lambda_S+1)$ cut $C$ subdividing ${\mathcal W}$ may subdivide multiple bunches adjacent to node $N$. 
Therefore, to overcome this difficulty in analyzing $(\lambda_S+1)$ cuts subdividing ${\mathcal W}$, we construct the following graph $G({\mathcal W})$ from $G$. \\

 \noindent
\textbf{Construction of graph $G({\mathcal W})$:}
Let us select a bunch ${\mathcal B}_1$ adjacent to node $N$ and $C_1=\overline{C(S_{{\mathcal B}_1})}$. 
For any other bunch ${\mathcal B}_2$ adjacent to $N$, let $C_2=\overline{C(S_{{\mathcal B}_2})}$. It follows from the definition of bunches and construction of Skeleton that $C_2$ (likewise $C_1$) cannot subdivide the Steiner set $\overline{S_{{\mathcal B}_1}}$ (likewise $\overline{S_{{\mathcal B}_2}}$). But, $C_2$ can cross $C_1$, and there are only nonSteiner vertices in $C_1\cap C_2$. So, $\overline{S_{{\mathcal B}_1}}\subseteq C_1\cap \overline{C_2}$ and $\overline{S_{{\mathcal B}_2}}\subseteq \overline{C_1}\cap C_2$. Therefore, $C_1\cap \overline{C_2}$ and $C_1\cup \overline{C_2}$ are Steiner cuts. By Lemma \ref{lem : mincuts closed under int and uni}, $C_1\cup \overline{C_2}$ is a Steiner mincut and belongs to bunch ${\mathcal B}_2$.
We contract all the vertices belonging to $C_1$ into a single Steiner vertex. It follows that after contraction, bunch ${\mathcal B}_2$ and its Steiner partition remains unchanged. Moreover, the capacity of Steiner mincut remains the same. Now, we select another bunch adjacent to node $N$ in the resulting graph and repeat the contraction procedure. Finally, after processing every bunch adjacent to node $N$ in the similar way, we obtain graph $G({\mathcal W})$.\\


\noindent
A vertex $v$ of $G({\mathcal W})$ is a Steiner vertex if $v$ is a Steiner vertex in $G$ or $v$ is a contracted vertex to which a Steiner vertex of $G$ is mapped. We denote the Steiner set of $G({\mathcal W})$ by $S({\mathcal W})$ and nonSteiner vertices of $G({\mathcal W})$ by $\overline{S({\mathcal W})}$.  Observe that the capacity of $S$-mincut in $G({\mathcal W})$ is $\lambda_S$ and every nonSteiner vertex of $G({\mathcal W})$ is a nonSteiner vertex of $G$ that belongs to ${\mathcal W}$. 
Hence, without causing any ambiguity, we denote a Steiner cut of capacity $(\lambda_{S({\mathcal W})}+1)$ in $G({\mathcal W})$ as $(\lambda_S+1)$ cut. 
The following property is satisfied by graph $G({\mathcal W})$. 
\begin{lemma} \label{lem : G(W)}
   A $(\lambda_S+1)$ cut subdivides ${\mathcal W}$ into $({\mathcal W}_1,{\mathcal W}\setminus {\mathcal W}_1)$ in $G$ if and only if there exists a $(\lambda_S+1)$ cut that subdivides ${\mathcal W}$ into $({\mathcal W}_1,{\mathcal W}\setminus {\mathcal W}_1)$ in $G({\mathcal W})$. 
\end{lemma}
\begin{proof}
    Suppose $(\lambda_S+1)$ cut $C$ subdivides ${\mathcal W}$ into $({\mathcal W}_1,{\mathcal W}\setminus {\mathcal W}_1)$ in $G$. 
    We now state the following assertion.
    
    \noindent
    ${\mathcal A}(i):$ If $C$ subdivides $i$ number of bunches adjacent to $N$, then there is a $(\lambda_S+1)$ cut in $G({\mathcal W})$ that subdivides ${\mathcal W}$ into $({\mathcal W}_1,{\mathcal W}\setminus {\mathcal W}_1)$.\\
    We prove, by induction on $i$, that ${\mathcal A}(i)$ holds for all $i\ge 1$. We first establish the base case, that is $i=1$.  
    Suppose $C$ subdivides $\overline{C(S_{\mathcal B})}$ for exactly one bunch ${\mathcal B}$ adjacent to $N$.
    Assume to the contrary that there is no $(\lambda_S+1)$ cut in $G({\mathcal W})$ that subdivides ${\mathcal W}$ into $({\mathcal W}_1,{\mathcal W}\setminus {\mathcal W}_1)$.
    Let  $C'=\overline{C(S_{\mathcal B})}$ and we know that $C'$ is an $S$-mincut. Since both $C$ and $C'$ are Steiner cuts, observe that either $C\cap C'$ and $C\cup C'$ are Steiner cuts or $C\setminus C'$ and $C'\setminus C$ are Steiner cuts. Without loss of generality, assume that $C\cap C'$ and $C\cup C'$ are Steiner cuts; otherwise, we can consider $\overline{C}$ instead of $C$. So, $c(C\cap C')$ and $c(C\cup C')$ are at least $\lambda_S$. Observe that either $C\cap C'$ or $C\cup C'$ subdivides ${\mathcal W}$. Suppose $C\cap C'$ subdivides ${\mathcal W}$.  The proof is similar if $C\cup C'$ subdivides ${\mathcal W}$. Since $C\cap C'$ subdivides ${\mathcal W}$, $c(C\cap C')$ is at least $\lambda_S+1$. By sub-modularity of cuts (Lemma \ref{submodularity of cuts}(1)), $c(C\cap C')+c(C\cup C')\le 2\lambda_S+1$. It follows that $c(C\cap C')=\lambda_S+1$ and $c(C\cup C')=\lambda$. Therefore, we have a $(\lambda_S+1)$ cut $C\cap C'$ that subdivides ${\mathcal W}$ into $({\mathcal W}_1,{\mathcal W}\setminus {\mathcal W}_1)$. Moreover, it follows from our construction of $G({\mathcal W})$ that after contracting $\overline{C(S_{\mathcal B})}$, in the resulting graph, $C\cap C'$ does not subdivide any other bunches adjacent to $N$. Therefore, $C\cap C'$ remains unaffected after constructing $G({\mathcal W})$, a contradiction.
    
    \noindent
    Suppose ${\mathcal A}(j)$ holds for all $j<i$. We now establish ${\mathcal A}(i)$. In the procedure of constructing $G({\mathcal W})$, let ${\mathcal B}$ be the first bunch for which $\overline{C(S_{\mathcal B})}$ is contracted out of
    all the $i$ bunches that cut $C$ subdivides. In a similar way to the base case, using cut $\overline{C(S_{\mathcal B})}$, we can show the following. In the resulting graph after contracting $\overline{C(S_{\mathcal B})}$, there is a $(\lambda_S+1)$ cut $C'$ that subdivides ${\mathcal W}$ into $({\mathcal W}_1,{\mathcal W}\setminus {\mathcal W}_1)$ and does not subdivide bunch ${\mathcal B}$. 
    So, in this graph, $C'$ subdivides at most $i-1$ bunches. Observe that cut $C'$ is also a $(\lambda_S+1)$ cut in $G$ and subdivides at most $i-1$ bunches adjacent to $N$. Therefore, by Induction Hypothesis, there is a $(\lambda_S+1)$ cut in $G({\mathcal W})$ that subdivides ${\mathcal W}$ into $({\mathcal W}_1,{\mathcal W}\setminus {\mathcal W}_1)$. Hence ${\mathcal A}(i)$ holds.    
    
    The proof of the converse part is evident. 
\end{proof}
It follows from Lemma \ref{lem : G(W)} that we can work with $G({\mathcal W})$ instead of $G$ for answering query $\textsc{cut}_{\mathcal W}$.
Observe that ${\mathcal W}$ appears as a $(\lambda_S+1)$ class of $G({\mathcal W})$ which contains Steiner vertex $s$.

For graph $G({\mathcal W})$, let $u\in {\mathcal W}$ and $N_{S({\mathcal W})}(u)=\{C_1,C_2,\ldots,C_k\}$, where $N_{S({\mathcal W})}(u)$ contains nearest $(\lambda_S+1)$ cuts for $u$ to every nonSteiner vertex in ${\mathcal W}$. 
The following lemma can be established in the same way as of Lemma \ref{lem : for a single bunch data structure}. 
\begin{lemma} [\textsc{Property ${\mathcal P}_1$}] \label{lem : pair of bunches}
   For any pair of cuts $C_1,C_2\in N_{S({\mathcal W})}(u)$, if there exists a Steiner vertex $t$ such that $t\in \overline{C_1}\cap \overline{C_2}$, then $\overline{C_1}\cap \overline{C_2}\cap {\mathcal W}=\emptyset$.  
\end{lemma}
Lemma \ref{lem : pair of bunches} essentially states that the complement set of a pair of nearest $(\lambda_S+1)$ cuts for a vertex does not contain any vertex from ${\mathcal W}$ only if their union is a Steiner cut. Unfortunately, there may exist a pair of cuts $A,A'$ from $N_{S({\mathcal W})}(u)$ such that $A\cup A'$ is not a Steiner cut. Moreover, it is possible that there exist vertices belonging to ${\mathcal W}$ which belong to $\overline{A\cup A'}$, that is, $\overline{A}\cap \overline{A'}\cap {\mathcal W}\ne \emptyset$ (refer to Figure \ref{fig : 1st figure in overview}$(ii)$). 
For every cut $C\in N_{S({\mathcal W})}(u)$, there can be ${\mathcal O}(|\overline{S({\mathcal W})}|)$ number of vertices. Moreover, $|N_{S({\mathcal W})}(u)|$ is ${\mathcal O}(|\overline{S({\mathcal W})}|)$. Therefore, to store all cuts from $\{N_{S({\mathcal W})}(u)~|~\forall u\in {\mathcal W}\}$ explicitly, it seems that ${\mathcal O}(|\overline{S({\mathcal W})}|^3+|S({\mathcal W})||\overline{S({\mathcal W})}|^2)$ space is required, which is $\Omega(n^3)$ in the worst case. Interestingly, exploiting \textsc{Gen-3-Star Lemma} (Theorem \ref{thm : gen 3 star}), we establish the following insight that helps in compactly storing all the vertices from $\overline{S({\mathcal W})}$. 
  \begin{figure}
 \centering
    \includegraphics[width=\textwidth]{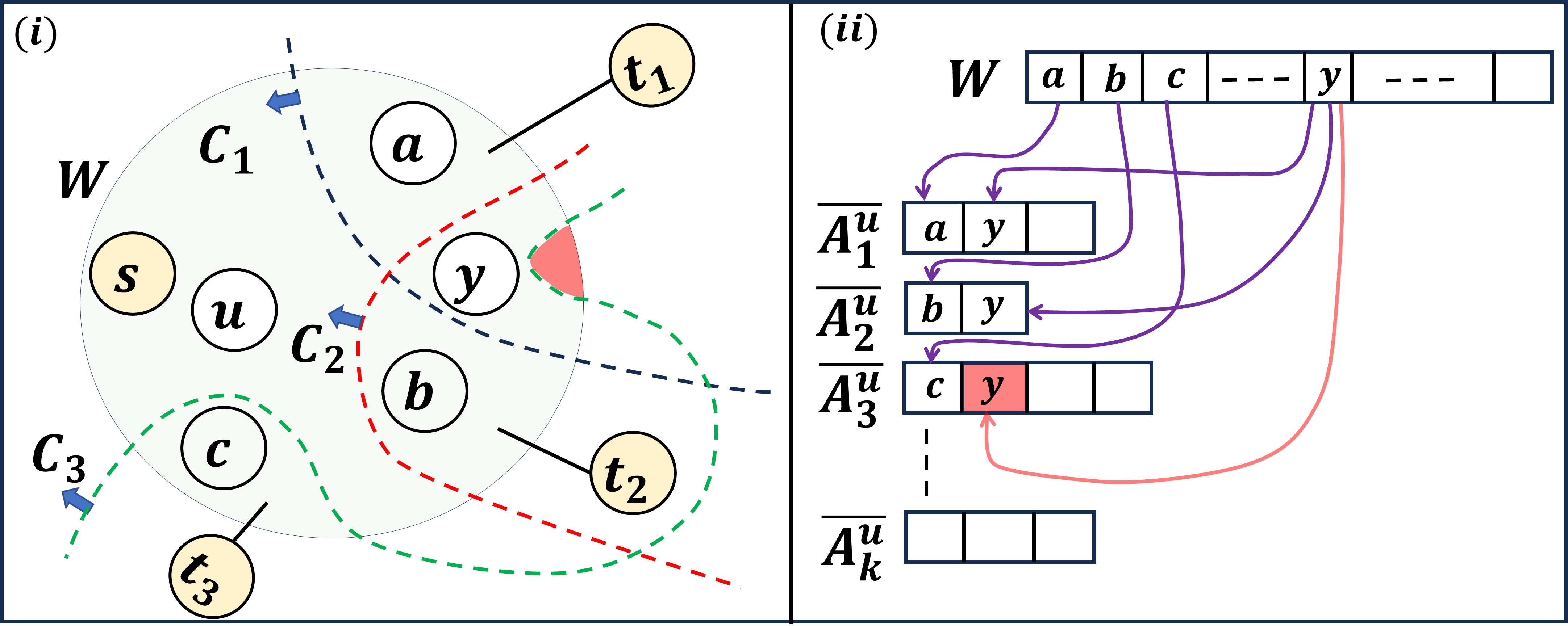} 
   \caption{$(i)$ The set $\overline{C_1\cup C_2\cup C_3}$
   of any three cuts $C_1,C_2,C_3\in N_{S({\mathcal W})}(u)$ is empty, shown in red region. $(ii)$ Data structure ${\mathcal N}_{{\mathcal W}}^{\ge 4}(u)$ for ${\mathcal W}$. Vertex $y$ cannot belong to more than two arrays.}
  \label{fig : data structure}. 
\end{figure}
\begin{lemma} \label{lem : three sets are nonempty}
    Let $C_1,C_2,$ and $C_3$ are any three cuts from $N_{S({\mathcal W})}(u)$. If the capacity of global mincut in $G({\mathcal W})$ is at least $4$, then $\overline{C_1}\cap \overline{C_2} \cap \overline{C_3}\cap {\mathcal W}$ is an empty set. 
\end{lemma}
\begin{proof}
    Assume that for cuts $C_1,C_2$, and $C_3$, the set $\overline{C_1}\cap \overline{C_2} \cap \overline{C_3}\cap {\mathcal W}$ is nonempty (refer to Figure \ref{fig : data structure}$(i)$). So, for any pair of cuts $C_i,C_j$, $1\le i\ne j\le 3$, $\overline{C_i}\cap \overline{C_j}$ contains at least one vertex from ${\mathcal W}$. By Lemma \ref{lem : pair of bunches}, there cannot exist Steiner vertex $t\ne s$ such that $t\in \overline{C_i}\cap \overline{C_j}$. Note that $C_1,C_2,$ and $C_3$ are Steiner cuts. This implies that each of the three sets $\overline{C_1}\cap C_2 \cap C_3$, $C_1\cap \overline{C_2}\cap C_3$, and $C_1\cap C_2 \cap \overline{C_3}$ must contain at least one Steiner vertex. So, each of the three sets $\overline{C_1}\cap C_2 \cap C_3$, $C_1\cap \overline{C_2}\cap C_3$, and $C_1\cap C_2 \cap \overline{C_3}$ defines a Steiner cut. Therefore, for three cuts $\overline{C_1},\overline{C_2},$ and $\overline{C_3}$, by Theorem \ref{thm : gen 3 star}(1), we have $c(\overline{C_1}\cap \overline{C_2} \cap \overline{C_3})\le 3$, a contradiction because the capacity of global mincut in $G({\mathcal W})$ is at least $4$.      
\end{proof}
With the help of Lemma \ref{lem : three sets are nonempty}, we design the following data structure for all cuts $N_{S({\mathcal W})}(u)$.
\subsection*{Construction of Data Structure ${\mathcal N}^{\ge 4}_{\mathcal W}(u)$}
Suppose capacity of global mincut of $G({\mathcal W})$ is at least $4$. Data structure ${\mathcal N}^{\ge 4}_{\mathcal W}(u)$ for storing $N_{S({\mathcal W})}(u)$ is constructed in Algorithm \ref{alg : nearest cut data structure}. 
\begin{algorithm}[H]
\caption{Construction of Data structure ${\mathcal N}^{\ge 4}_{\mathcal W}(u)$}
\label{alg : nearest cut data structure}
\begin{algorithmic}[1]
\Procedure{\textsc{DataStuctureConstruction}$(u)$}{}
    \State $i=1$;
       \For{every cut $C_i$ in $N_{S({\mathcal W})}(u)$ with 
 $s\in C_i$}
          \State store $(\overline{C_i}\cap {\mathcal W})\cup (\overline{C_i}\cap S({\mathcal W}))$ in array ${\mathcal A}_i^u$; 
              
        \State $i\gets i+1$
    \EndFor
\EndProcedure
\end{algorithmic}
\end{algorithm}
\noindent
\textbf{Analyzing the space occupied by ${\mathcal N}^{\ge 4}_{\mathcal W}(u)$:} 
For any three cuts from $N_{S({\mathcal W})}(u)$, it follows from Lemma \ref{lem : three sets are nonempty} that 
every vertex in $\overline{S({\mathcal W})}$ appears in at most two arrays ${\mathcal A}_i^u, {\mathcal A}_j^u$, $1\le i\ne j\le k$, in the data structure ${\mathcal N}^{\ge 4}_{\mathcal W}(u)$ (refer to Figure \ref{fig : data structure}$(ii)$). 
This ensures that the data structure ${\mathcal N}^{\ge 4}_{\mathcal W}(u)$ occupies ${\mathcal O}(|\overline{S({\mathcal W})}|)$ space for storing vertices from $\overline{S({\mathcal W})}$.\\ 

We now analyze the space occupied by ${\mathcal N}_{\mathcal W}^{\ge 4}(u)$ for storing Steiner vertices belonging to $\overline{C}$ for each cut $C\in N_{S({\mathcal W})}(u)$. 
It is easy to observe that there can be at most $2|\overline{S({\mathcal W})}|$ cuts in $N_{S({\mathcal W})}(u)$. If, for every cut $C\in N_{S({\mathcal W})}(u)$, $\overline{C}$ contains ${\mathcal O}(|S({\mathcal W})|)$ Steiner vertices, then the space occupied by ${\mathcal N}_{\mathcal W}^{\ge 4}(u)$ is ${\mathcal O}(|\overline{S({\mathcal W})}||S({\mathcal W})|)$. So, it can be observed that, for all vertices in ${\mathcal W}$, the space occupied by ${\mathcal N}_{\mathcal W}^{\ge 4}$ is ${\mathcal O}((n-|S|)^2|S|)$ in the worst case. 
Interestingly, we show in two phases that the space occupied by ${\mathcal N}_{\mathcal W}^{\ge 4}(u)$ is ${\mathcal O}(|\overline{S({\mathcal W})}|+|S({\mathcal W})|)$.
\subsection{Achieving ${\mathcal O}(|\overline{S({\mathcal W})}|\sqrt{|S|})$ space bound for ${\mathcal N}_{\mathcal W}^{\ge 4}(u)$} \label{app : root S bound} 
We establish the following insight that reduces the space to ${\mathcal O}(|\overline{S({\mathcal W})}|\sqrt{|S({\mathcal W})|})$. 
\begin{lemma} [\textsc{Property ${\mathcal P}_2$}] \label{lem : at most one bunch in common}
    For any pair of cuts $C,C'\in N_{S({\mathcal W})}(u)$, there can be at most one Steiner vertex $t$ such that $t\in \overline{C} \cap \overline{C'}$.
\end{lemma}
\begin{proof}
    Suppose there is a pair of Steiner vertices $t_1,t_2$ such that $t_1,t_2\in \overline{C}\cap \overline{C'}$. Since $s$ belongs to both $C$ and $C'$,  so, $C\cap C'$ and $C\cup C'$ are Steiner cuts. This implies that $c(C\cap C')$, as well as $c(C\cup C')$, is at least $\lambda_S$. Since $C$ and $C'$ belong to $N_{S({\mathcal W})}(u)$, by Definition \ref{def : nearest minimum+1 cut}, $C$ and $C'$ cannot be subsets of each other. Since $u\in C\cap C'$ and $C\cap C'$ is a proper subset of both $C$ and $C'$, therefore, $c(C\cap C')\ge \lambda_S+2$.  By sub-modularity of cuts (Lemma \ref{submodularity of cuts}(1)), we have $c(C\cap C')+c(C\cup C')\le 2\lambda_S+2$. This implies that $c(C\cup C')= \lambda_S$, a contradiction because, by construction of $G({\mathcal W})$, there is no $S$-mincut $C$ such that $s\in C$ and $t_1,t_2\in \overline{C}$ or vice versa.   
\end{proof}
Let us construct a bipartite graph $B$ with vertex set $V_B=V_1 \cup V_2$. The vertex sets $V_1$ and $V_2$ of $B$ are defined as follows. For every cut $C$ from $N_{S({\mathcal W})}(u)$, there is a vertex in $V_1$. Similarly, for every Steiner vertex $t\ne s$ of $G({\mathcal W})$, there is a vertex in $V_2$. So, $B$ is a ${\mathcal O}(|\overline{S({\mathcal W})}|) \times {\mathcal O}(|S({\mathcal W})|)$ bipartite graph. We now define the edge set of $B$. Let $C$ be a cut in $N_{S({\mathcal W})}(u)$ represented by a vertex $v\in V_1$. Similarly, let $t\ne s$ be a Steiner vertex represented by a vertex $b\in V_2$. If $t\in \overline{C}$, then there is an edge between $v$ and $b$.

The \textit{girth} of a graph is the number of edges belonging to the smallest cycle of the graph. It follows from Lemma \ref{lem : at most one bunch in common} that there cannot exist any cycle of length $4$ in $B$. So, since $B$ is a bipartite graph, the girth of $B$ is strictly greater than $4$. Erdos \cite{erdos1964extremal} and Bondy and Simonovits \cite{bondy1974cycles} established the following property on the number of edges of a bipartite graph. 
\begin{lemma} [\cite{erdos1964extremal, bondy1974cycles}] \label{lem : girth is 4}
    Let ${\mathcal C}_B$ be the class of all $P$ by $Q$ bipartite graphs. For any graph $H\in {\mathcal C}_B$, if $H$ has ${\Omega}(P\sqrt{Q})$ edges, then $H$ has girth at most $4$.    
\end{lemma}
Since girth of graph $B$ is strictly greater than $4$, therefore, by Lemma \ref{lem : girth is 4}, the bipartite graph $B$ can have ${\mathcal O}(|\overline{S({\mathcal W})}|\sqrt{|S({\mathcal W})|})$ edges. 
Hence, data structure ${\mathcal N}_{\mathcal W}^{\ge 4}$ occupy ${\mathcal O}(|\overline{S({\mathcal W})}|\sqrt{|S({\mathcal W})|})$ space in the worst case.

\subsection{Achieving ${\mathcal O}(|\overline{S({\mathcal W})}|+|S({\mathcal W})|)$ space bound for ${\mathcal N}_{\mathcal W}^{\ge 4}(u)$}
To establish a stricter bound on the space for ${\mathcal N}_{\mathcal W}^{\ge 4}(u)$, we present a classification of the set of all cuts from $N_{S({\mathcal W})}(u)$ into two types, namely, \textit{$l$-cut} and \textit{$m$-cut}. \\

\noindent
\textbf{A Classification of All Cuts From $N_{S({\mathcal W})}(u)$:} A cut $C\in N_{S({\mathcal W})}(u)$ is said to be a $l$-cut if there is exactly one Steiner vertex $t$ such that $t\in \overline{C}$, otherwise $C$ is called a \textit{$m$-cut}. \\

Since every $l$-cut $C_i$ stores at most one Steiner vertex in the array ${\mathcal A}_i^u$ (Algorithm \ref{alg : nearest cut data structure}), therefore, it is easy to observe that the space occupied by storing all the $l$-cuts of $N_{S({\mathcal W})}(u)$ is only ${\mathcal O}(|\overline{S({\mathcal W})}|+|S({\mathcal W})|)$ in data structure ${\mathcal N}_{\mathcal W}^{\ge 4}(u)$. 

We now consider the set of all $m$-cuts. It seems that there still exist multiple Steiner vertices that belong to more than one $m$-cuts. Interestingly, for every Steiner vertex, the following crucial lemma provides a strict bound on the number of $m$-cuts to which it does not belong.
\begin{lemma} \label{lem : a bunch belong to at most two cuts}
    For any Steiner vertex $t$, there can be at most two $m$-cuts $C,C'$ such that $t\in \overline{C}\cap \overline{C'}$. 
\end{lemma}
\begin{proof}
    Suppose there exist three $m$-cuts $C_1,C_2$, and $C_3$ such that $t\in \overline{C_1}\cap \overline{C_2}\cap \overline{C_3}$. By definition, for every $m$-cut $C\in \{C_1,C_2,C_3\}$, there are at least two Steiner vertices $t_1,t_2$ such that $t_1,t_2\in \overline{C}$. For any pair of integers $i,j\in [1,3]$ and $i\ne j$, by Lemma \ref{lem : at most one bunch in common}, there cannot be more than one Steiner vertex in $\overline{C_i} \cap \overline{C_j}$. Therefore, for the three cuts, $C_1,C_2,C_3$, there must exist at least three Steiner vertices $t_1,t_2,t_3$ such that $t_1\in \overline{C_1}\cap C_2 \cap C_3$, $t_2\in C_1 \cap \overline{C_2}\cap C_3$, and $t_3\in C_1\cap C_2 \cap \overline{C_3}$. Now, since $C_1,C_2,C_3$ are nearest cuts of $u$ that subdivide ${\mathcal W}$, there must exist vertices $v_1,v_2,v_3\in {\mathcal W}$ such that $v_1\in \overline{C_1}$, $v_2\in \overline{C_2}$, and $v_3\in \overline{C_3}$. It is also assumed that $t\in \overline{C_1}\cap \overline{C_2}\cap \overline{C_3}$. So, for any pair of integers $i,j\in [1,3]$ and $i\ne j$, by Lemma \ref{lem : pair of bunches}, there cannot be any vertex from ${\mathcal W}$ in $\overline{C_i}\cap \overline{C_j}$.    
    This implies that $v_1\in \overline{C_1}\cap C_2\cap C_3$, $v_2\in C_1\cap \overline{C_2} \cap C_3$ and $v_3\in C_1\cap C_2\cap \overline{C_3}$. So, we have ensured that each of the three cuts $\overline{C_1},\overline{C_2},$ and $\overline{C_3}$ contains at least one Steiner vertex of $G({\mathcal W})$ and at least one vertex from the same class ${\mathcal W}$. Therefore, for three cuts $\overline{C_1}, \overline{C_2},$ and $\overline{C_3}$, it follows from Lemma \ref{lem : cor of gen 3 star} that $\overline{C_1}\cap \overline{C_2}\cap \overline{C_2}=\emptyset$, a contradiction.
\end{proof}

It follows from Lemma \ref{lem : a bunch belong to at most two cuts} that a Steiner vertex can appear in at most two arrays ${\mathcal A}_i^u, {\mathcal A}_j^u$, $1\le i\ne j \le k$, in ${\mathcal N}_{\mathcal W}^{\ge 4}({\mathcal W})$. Therefore, it requires only ${\mathcal O}(|S({\mathcal W})|)$ space to store all Steiner vertices for all $m$-cuts. So, overall space occupied by ${\mathcal N}_{\mathcal W}^{\ge 4}(u)$ is ${\mathcal O}(|\overline{S({\mathcal W})}|+|S({\mathcal W})|)$. \\

\noindent
\textbf{Answering existential query:} Let us mark a vertex $v$ in ${\mathcal W}$ by $i$ if $u$ and $v$ are separated by $i^{th}$ cut in $N_{S({\mathcal W})}(u)$. Suppose $v$ and $w$ are any pair of vertices in ${\mathcal W}$. By verifying whether marks of $v$ intersect marks of $w$, the data structure ${\mathcal N}^{\ge 4}_{\mathcal W}(u)$ can answer the following query. 
\begin{equation*}
\textsc{belong}(u,v,w)=\begin{cases}
            1 \quad &\text{if there exists a cut $C\in N_{S({\mathcal W})}(u)$ such that $v,w\in \overline{C}$}  \\
            0 \quad &\text{otherwise.}
     \end{cases}
\end{equation*}
It follows from the construction (Lemma \ref{lem : three sets are nonempty}) of ${\mathcal N}^{\ge 4}_{\mathcal W}(u)$ that every vertex can be marked with at most two values. Hence, for any fixed vertex $u\in {\mathcal W}$, the query \textsc{belong} can be answered using ${\mathcal N}^{\ge 4}_{\mathcal W}(u)$ in ${\mathcal O}(1)$ time. This establishes the following lemma.
\begin{lemma} \label{lem : data structure for a single vertex at least 4}
    Suppose the capacity of the global mincut of graph $G$ is at least $4$. Let ${\mathcal W}$ be any $(\lambda_S+1)$ class in $G$ and $s,u\in {\mathcal W}$ where $s\in S$. There exists an $\mathcal{O}(|\overline{S({\mathcal W})}|+|S({\mathcal W})|)$ space data structure ${\mathcal N}^{\ge 4}_{\mathcal W}(u)$ that can determine in ${\mathcal O}(k)$ time whether there exists a $(\lambda_S+1)$ cut $C$ such that $s,u\in C$ and $v_1,\ldots,v_k\in \overline{C}$ for any $v_1,\ldots,v_k \in {\mathcal W}$. If $C$ exists, the data structure can output $\overline{C}\cap {\mathcal{W}}$ in ${\mathcal{O}}(|\overline{C}\cap {\mathcal{W}}|)$ time and $\overline{C}\cap S({\mathcal{W}})$ in ${\mathcal{O}}(|\overline{C}\cap S({\mathcal{W}})|)$ time. 
\end{lemma}
\subsection{Extension to General Graphs} \label{app : extension to general graphs}
Suppose the capacity of global mincut of $G({\mathcal W})$ is less than $4$. It turns out that more than two cuts may exist from $N_{S({\mathcal W})}(u)$ such that the complement of their union has a nonempty intersection with ${\mathcal W}$ (refer to Figure \ref{fig : 1st figure in overview}$(ii)$). The following lemma ensures that even if multiple cuts exist, storing at most two cuts from $N_{S({\mathcal W})}(u)$ for every vertex is sufficient for answering query $\textsc{cut}_{\mathcal W}$. 

\begin{lemma} \label{lem : storing at most two cuts are sufficient}
    For any three cuts $C_0,C_1,$ and $C_2$ in $N_{S({\mathcal W})}(u)$, if $\overline{C_0}\cap \overline{C_1}\cap \overline{C_2}\cap {\mathcal W}$ is nonempty, then there exists an $i\in \{0,1,2\}$ such that $\overline{C_i}\cap C_{(i+1)\%3}\cap C_{(i+2)\%3}\cap {\mathcal W}$ is an empty set. 
\end{lemma}
\begin{proof}
    Suppose there does not exist any $i\in [0,2]$ such that $\overline{C_i}\cap C_{(i+1)\%3}\cap C_{(i+2)\%3}\cap {\mathcal W}$ is an empty set. Let $P=C_0\cap C_1\cap \overline{C_2}$, $Q=C_0\cap \overline{C_1}\cap C_2$, and $R=\overline{C_0}\cap C_1\cap C_2$. Since for every $i,j\in\{0,1,2\}$, $\overline{C_i}\cap \overline{C_j}$ is nonempty, therefore, by Lemma \ref{lem : pair of bunches}, there cannot exist any Steiner vertex $t$ such that $t\in \overline{C_i}\cap \overline{C_j}$. Moreover, cuts $C_0,C_1,$ and $C_2$ are Steiner cuts.  Hence, each of the three sets $P,Q,$ and $R$ contains a Steiner vertex. This implies that each of the three sets $P,Q,$ and $R$ defines a Steiner cut. Since we assumed that each of the three cuts $P,Q,$ and $R$ contains at least one vertex from ${\mathcal W}$, so each of them subdivides ${\mathcal W}$. By Lemma \ref{lem : cor of gen 3 star}, for cuts $\overline{C_0}, \overline{C_1},$ and $\overline{C_2}$, we have $\overline{C_0}\cap \overline{C_1}\cap \overline{C_2}=\emptyset$, a contradiction.
\end{proof}
By Lemma \ref{lem : storing at most two cuts are sufficient}, for any three cuts $C_0,C_1,$ and $C_2$ in $N_{S({\mathcal W})}(u)$, if $\overline{C_0}\cap \overline{C_1}\cap \overline{C_2}\cap {\mathcal W}$ is nonempty, then we do not have to store cut $C_i$, $i\in \{0,1,2\}$, for which $\overline{C_i}\cap C_{(i+1)\%3}\cap C_{(i+2)\%3}\cap {\mathcal W}$ is an empty set. This is because for each vertex $u\in \overline{C_i}\cap {\mathcal W}$, we have already stored at least one of the two cuts $C_{(i+1)\%3}$ and $ C_{(i+2)\%3}$. The remaining construction and analysis is the same as that of data structure ${\mathcal N}^{\ge 4}_{\mathcal W}(u)$, stated in Algorithm \ref{alg : nearest cut data structure}, because they do not depend on the capacity of global mincut of $G({\mathcal W})$. 
This leads to the data structure in the following lemma. 

\begin{lemma} \label{lem : data structure for a single vertex at most 3}
    Suppose the capacity of the global mincut of graph $G$ is at most $3$. Let ${\mathcal W}$ be any $(\lambda_S+1)$ class in $G$ and $s,u\in {\mathcal W}$ where $s\in S$. There exists an $\mathcal{O}(|\overline{S({\mathcal W})}|+|S({\mathcal W})|)$ space data structure ${\mathcal N}^{\le 3}_{\mathcal W}(u)$ that can determine in ${\mathcal O}(1)$ time whether there exists a $(\lambda_S+1)$ cut $C$ such that $s,u\in C$ and $v\in \overline{C}$ for any $v \in {\mathcal W}$. If $C$ exists, the data structure can output $\overline{C}\cap {\mathcal W}$ in ${\mathcal{O}}(|\overline{C}\cap {\mathcal W}|)$ time and $\overline{C}\cap S({\mathcal{W}})$ in ${\mathcal{O}}(|\overline{C}\cap S({\mathcal{W}})|)$ time. 
\end{lemma}
We now design data structure ${\mathcal Q}_S({\mathcal W})$ using data structures in Lemma \ref{lem : data structure for a single vertex at least 4} and Lemma \ref{lem : data structure for a single vertex at most 3}.\\

\noindent
\textbf{Description of Data Structure ${\mathcal Q}_S({\mathcal W})$:} For every vertex $u\in {\mathcal W}$, store ${\mathcal N}^{\le 3}_{\mathcal W}(u)$ if capacity of global mincut is at most $3$, otherwise store ${\mathcal N}^{\ge 4}_{\mathcal W}(u)$.\\

Let us analyze the space occupied by ${\mathcal Q}_S({\mathcal W})$. For the only Steiner vertex $s\in {\mathcal W}$, by Lemma \ref{lem : data structure for a single vertex at most 3} and Lemma \ref{lem : data structure for a single vertex at least 4}, the space occupied by $N_{S({\mathcal W})}(s)$ is ${\mathcal O}(|\overline{S({\mathcal W})}|+|S({\mathcal W})|)$. Similarly, ${\mathcal O}(|\overline{S({\mathcal W})}|^2+|\overline{S({\mathcal W})}||S({\mathcal W}|))$ space is occupied by all nonSteiner vertices in $\overline{S({\mathcal W})}$ of ${\mathcal W}$.
Therefore, the data structure ${\mathcal Q}_S({\mathcal W})$ occupies ${\mathcal O}(|\overline{S({\mathcal W})}|^2+|\overline{S({\mathcal W})}||S({\mathcal W})|)$ space.\\

\noindent
\textbf{Answering query} $\textsc{cut}_{\mathcal W}:$ Let $u$ and $v$ are any given pair of vertices from ${\mathcal W}$. It follows from Lemma \ref{lem : data structure for a single vertex at least 4}, Lemma \ref{lem : data structure for a single vertex at most 3}, and the construction of data structure ${\mathcal Q}_S({\mathcal W})$ that, using Query \textsc{belong}, data structure $Q_S({\mathcal W})$ can determine in ${\mathcal O}(1)$ time if there is a cut $C_i\in N_{S({\mathcal W})}(u)$ such that $v \in \overline{C_i}\cap {\mathcal W}$ and $s,u\in C_i\cap {\mathcal W}$. By using array ${\mathcal A}_i^u$ (Algorithm \ref{alg : nearest cut data structure}), it can report $C_i$ in $G({\mathcal W})$ in ${\mathcal O}(|{\overline{S({\mathcal W})}}|+|S({\mathcal W})|)$ time.
So, 
by storing an ${\mathcal O}(n)$ space mapping of vertices from $G$ to $G({\mathcal W})$, we can report the cut $C_i$ in $G$ in ${\mathcal O}(n)$ time. Therefore, the data structure for answering query $\textsc{cut}_{\mathcal W}$ occupies ${\mathcal O}(n+|\overline{S({\mathcal W})}|^2+|\overline{S({\mathcal W})}||S({\mathcal W})|)$ space. The results of this section are summarized in the following theorem. 
\begin{theorem} \label{thm : data structure for singleton class}
    Let $G=(V,E)$ be an undirected multi-graph on $n$ vertices with a Steiner set $S\subseteq V$. Let ${\mathcal W}$ be a $(\lambda_S+1)$ class of $G$ containing exactly one Steiner vertex $s$ and let $\overline{S({\mathcal W})}$ be the set of nonSteiner vertices of $G$ belonging to ${\mathcal W}$. Let $\lambda$ be the capacity of global mincut of $G$. For  $(\lambda_S+1)$ class ${\mathcal W}$, there is an ${\mathcal O}(|\overline{S({\mathcal W})}|^2+|S||\overline{S({\mathcal W})}|)$ space data structure ${\mathcal Q}_S({\mathcal W})$ that,
    \begin{enumerate}
        \item given vertices $u,v_1,\ldots,v_k$ from ${\mathcal W}$, can determine in ${\mathcal O}(k)$ time whether there exists a $(\lambda_S+1)$ cut $C$ such that $s,u\in C$ and $v_1,\ldots,v_k\in \overline{C}$ if $\lambda\ge 4$, and
        \item given vertices $u,v\in {\mathcal W}$, can determine in ${\mathcal O}(1)$ time whether there exists a $(\lambda_S+1)$ cut $C$ such that $s,u\in C$ and $v\in \overline{C}$ if $\lambda\le 3$.
    \end{enumerate}
    If $C$ exists, the data structure can report $\overline{C}\cap {\mathcal W}$ in ${\mathcal O}(|\overline{C}\cap {\mathcal W}|)$ time, and cut $C$ in ${\mathcal O}(n)$ time using an auxiliary ${\mathcal O}(n)$ space. 
\end{theorem}

\section{Data Structure for Reporting $(\lambda_S+1)$ cuts} \label{sec : data structure complete}
In this section, for graph $G$, we design a compact data structure that can answer query \textsc{cut} for $(\lambda_S+1)$ cuts efficiently.
Exploiting the data structure in Theorem \ref{thm : data structure for singleton class}, we first construct a data structure ${\mathcal D}_s$ for all Singleton $(\lambda_S+1)$ classes of $G$. Later, using data structure ${\mathcal D}_s$, we design data structures for generic $(\lambda_S+1)$ classes of $G$. A generic $(\lambda_S+1)$ class can contain either multiple Steiner vertices or no Steiner vertices, which are mapped to Steiner and NonSteiner nodes of Flesh graph ${\mathcal F}_S$ of $G$, respectively.
\subsection{A Data Structure for All Singleton $(\lambda_S+1)$ classes} \label{sec : data structure for all singleton classes}
We construct the data structure ${\mathcal D}_s$ for all Singleton $(\lambda_S+1)$ classes of $G$ as follows.\\

\noindent
\textbf{Construction of Data Structure ${\mathcal D}_s$:} For every Singleton $(\lambda_S+1)$ class ${\mathcal W}_s$ of $G$, we store a data structure ${\mathcal Q}_S({\mathcal W}_s)$ occupying ${\mathcal O}(|\overline{S({\mathcal W}_s)}|^2+|\overline{S({\mathcal W}_s)}||S({\mathcal W}_s)|)$ space from Theorem \ref{thm : data structure for singleton class}, where $S({\mathcal W}_s)$ is the Steiner set and $\overline{S({\mathcal W}_s)}$ is the set of nonSteiner vertices of $G({\mathcal W}_s)$. For each Singleton $(\lambda_S+1)$ class of $G$, storing an auxiliary ${\mathcal O}(n)$ space mapping of vertices of Theorem \ref{thm : data structure for singleton class} leads to ${\mathcal O}(n^2)$ space in the worst case. Hence, we do not store this mapping. To report a $(\lambda_S+1)$ cut, we store the quotient mapping ($\phi$) and projection mapping $(\pi)$ of units of Flesh graph of $G$ (refer to Appendix \ref{sec : extended preliminaries}). Finally, for every Singleton $(\lambda_S+1)$ class ${\mathcal W}_s$ of $G$, we store the following mapping. Let $N$ be the node in Skeleton ${\mathcal H}_S$ that represents the Singleton $(\lambda_S+1)$ class ${\mathcal W}_s$. Observe that, by construction of $G({\mathcal W}_s)$, for every bunch ${\mathcal B}$ adjacent to node $N$, there is a Steiner vertex $t$ in $G({\mathcal W}_s)$. We store at vertex $t$ the pair of cycle-edges or the tree-edge of ${\mathcal H}_S$ that represents bunch ${\mathcal B}$. \\ 

\noindent
\textbf{Analyzing the Space occupied by ${\mathcal D}_s:$} Storing projection mapping $\pi$ and quotient mapping $\phi$ occupies ${\mathcal O}(n)$ space (refer to Appendix \ref{sec : extended preliminaries} and Lemma \ref{lem : projection mapping}). \\
Let ${\mathbb W}$ be the set of all Singleton $(\lambda_S+1)$ classes of $G$. Let ${\mathcal W}_i\in {\mathbb W}$. By construction of Connectivity Carcass, there is a node $N_i$ in ${\mathcal H}_S$ that represents ${\mathcal W}_i$. Moreover, for a pair of Singleton $(\lambda_S+1)$ classes, they are mapped to different nodes in ${\mathcal H}_S$.
Therefore, by Lemma \ref{lem : sum of deg is order S}, it is ensured that storing a pair of cycle-edges or a tree-edge for each adjacent bunch of node $N_i$ for every Singleton $(\lambda_S+1)$ class ${\mathcal W}_i\in {\mathbb W}$ requires ${\mathcal O}(|S|)$ space.\\
For a $(\lambda_S+1)$ class ${\mathcal W}\in {\mathbb W}$, observe that, except the auxiliary ${\mathcal O}(n)$ space for reporting a cut in $G$, the data structure ${\mathcal Q}_S({\mathcal W})$ from Theorem \ref{thm : data structure for singleton class} occupies ${\mathcal O}(|\overline{S({\mathcal W})}|^2+|S||\overline{S({\mathcal W})}|)$ space. 
Since $\Sigma_{{\mathcal W}\in {\mathbb W}}(|\overline{S({\mathcal W})}|) ={\mathcal O}(n-|S|)$ and  $\Sigma_{{\mathcal W}\in {\mathbb W}}(|\overline{S({\mathcal W})}|^2)\le (\Sigma_{{\mathcal W}\in {\mathbb W}}(|\overline{S({\mathcal W})}|))^2={\mathcal O}((n-|S|)^2)$, therefore, overall space occupied by ${\mathcal D}_s$ is ${\mathcal O}(n+(n-|S|)^2+|S|(n-|S|))$, which is ${\mathcal O}(n(n-|S|+1))$.\\

\noindent
\textbf{Answering Query} \textsc{cut} \textbf{for all Singleton $(\lambda_S+1)$ class:} Let ${\mathcal W}$ be any Singleton $(\lambda_S+1)$ class with a Steiner vertex $s$. Let $N$ be the node in ${\mathcal H}_S$ that represents the $(\lambda_S+1)$ class ${\mathcal W}$. For any given pair of vertices $u$ and $v$ from ${\mathcal W}$, using Theorem \ref{thm : data structure for singleton class}, the data structure ${\mathcal D}_s$ can determine in ${\mathcal O}(1)$ time whether there is a $(\lambda_S+1)$ cut $C$ such that $s,u\in C$ and $v\in \overline{C}$. \\ 
Suppose there exists a $(\lambda_S+1)$ cut $C$ such that $s,u\in C$ and $v\in \overline{C}$. Recall that cut $C$ is for graph $G({\mathcal W})$. By using cut $C$, we now explain the procedure of reporting a $(\lambda_S+1)$ cut $C'$ of $G$ such that $s,u\in C'$ and $v\notin C'$.
It follows from Theorem \ref{thm : data structure for singleton class} that ${\mathcal D}_s$ can report all vertices of $G$ belonging to $\overline{C}\cap {\mathcal W}$ in ${\mathcal O}(|\overline{C}\cap {\mathcal W}|)$ time. It can also report each Steiner vertex $t$ of $G({\mathcal W})$ such that  $t\in \overline{C}$. Steiner vertex $t\in \overline{C}$ corresponds to a minimal cut (tree-edge or a pair of cycle-edges from the same cycle adjacent to node $N$) in ${\mathcal H}_S$, which is also stored at $t$. 
We remove the tree-edge/pair of cycle-edges from ${\mathcal H}_S$ and mark the subgraph of the Skeleton that does not contain node $N$ (containing Steiner vertex $s$). We repeat this process for every Steiner vertex $t$ for which $t\in  \overline{C}$. It is easy to observe that marking of every subgraph requires overall ${\mathcal O}(|S|)$ time because subgraphs not containing node $N$ are disjoint. Finally, we get a set of marked subgraphs of ${\mathcal H}_S$.

Let ${\mathcal B}$ be a bunch represented by a minimal cut corresponding to a Steiner vertex $t\in \overline{C}$ of $G({\mathcal W})$. 
By using the mapping of vertices from $G$ to $G({\mathcal W})$, for cut $C$ in $G({\mathcal W})$, let $C_G$ be the corresponding $(\lambda_S+1)$ cut in $G$. The following lemma helps in reporting a $(\lambda_S+1)$ cut in $G$ given cut $C$.
\begin{lemma} \label{lem : intersection is also a desired cut}
    Let $C_m=C(S_{\mathcal B})$. If $C_m$ and $C_G$ crosses then $C_m\cap C_G$ is a $(\lambda_S+1)$ cut in $G$ such that $s,u\in C_m\cap C_G$ and $v\in \overline{C_m\cap C_G}$.
\end{lemma}
\begin{proof}
   We have $s\in C_m\cap C_G$. Since $t\in \overline{C}$, by following our construction of $G({\mathcal W})$, $\overline{S_{\mathcal B}}\subseteq \overline{C_m\cup C_G}$. Therefore, $C_m\cap C_G$ and $C_m\cup C_G$ are Steiner cuts. $C_G$ separates vertices $u$ and $v$ belonging to ${\mathcal W}$. So, $C_m\cap C_G$ subdivides ${\mathcal W}$. Therefore, $c(C_G\cap C_m)\ge \lambda_S+1$. Since $C_m\cap C_G$ is a Steiner cut, $c(C_m\cup C_G)\ge \lambda_S$. By sub-modularity of cuts (Lemma \ref{submodularity of cuts}(1)), $c(C_m\cap C_G)+c(C_m\cup C_G)\le 2\lambda_S+1$. It follows that $C_m\cap C_G$ is a $(\lambda_S+1)$ cut. We also have $s,u\in C_m$ and $v\notin C_G$. Therefore, $C_m\cap C_G$ is a $(\lambda_S+1)$ cut such that $s,u\in C_m\cap C_G$ and $v\in \overline{C_m\cap C_G}$.   
\end{proof}
In graph $G({\mathcal W})$, suppose there are $k$ Steiner vertices that belong to $\overline{C}$ and the corresponding bunches are ${\mathcal B}_1,{\mathcal B}_2,\ldots, {\mathcal B}_k$. Let $A=\bigcap_{i=1}^k C(S_{{\mathcal B}_i})$. By using Lemma \ref{lem : intersection is also a desired cut} and an induction on the number of Steiner vertices belonging to $\overline{C}$ in $G({\mathcal W})$, it is a simple exercise to establish the following result.
\begin{lemma} \label{lem : reporting desired cut}
    Cut $C_G\cap A$ is a $(\lambda_S+1)$ cut such that $s,u\in C_G\cap A$ and $v\in \overline{C_G\cap A}$.
\end{lemma}
Given cut $C$, by Lemma \ref{lem : reporting desired cut}, our aim is now to report $\overline{C_G\cap A}$. To achieve this goal, we use the marked subgraphs of the Skeleton as follows. For every vertex $u$ of $G$, if $\pi(\phi(u))$ intersects any marked subgraph corresponding to a bunch ${\mathcal B}$, then report $u$. This is because, by definition of projection mapping $\pi$, $u\in \overline{C(S_{\mathcal B})}$, and hence, $u\in \overline{A}$. By Lemma \ref{lem : projection mapping}, this process requires ${\mathcal O}(n)$ time for all vertices of $G$. Therefore, we can report every vertex of $G$ that belongs to $\overline{C_G\cap A}$ in ${\mathcal O}(n)$ time. \\


We now state the following lemma for all Singleton $(\lambda_S+1)$ classes of $G$.
\begin{lemma} \label{lem : data structure for singleton class}
    There exists an ${\mathcal O}(n(n-|S|+1))$ space data structure ${\mathcal D}_s$ that, given any pair of vertices $u$ and $v$ from any Singleton $(\lambda_S+1)$ class containing one Steiner vertex $s$, can determine whether there exists a $(\lambda_S+1)$ cut $C$ such that $s,u\in C$ and $v\in \overline{C}$ in ${\mathcal O}(1)$ time. Moreover, it can report cut $C$ in ${\mathcal O}(n)$ time.
\end{lemma}
In the following, we consider the case when a $(\lambda_S+1)$ class contains no Steiner vertex. Later, in Appendix \ref{sec : generic steiner node}, we discuss the case when a $(\lambda_S+1)$ class contains more than one Steiner vertices.
\subsection{A Data Structure for All NonSteiner Nodes}
A $(\lambda_S+1)$ class containing no Steiner vertex is mapped to a nonSteiner node in the Flesh graph (refer to quotient mapping in Appendix \ref{sec : extended preliminaries}). Let $\mu$ be a NonSteiner node of the Flesh graph and ${\mathcal W}_{\mu}$ be the corresponding $(\lambda_S+1)$ class. By Definition \ref{def : terminal and nonterminal unit}, there are two possible cases -- $(i)$ $\mu$ is a Stretched unit or $(ii)$ $\mu$ is a terminal unit. We now discuss each case separately in Appendix \ref{sec : stretched unit data structure} and Appendix \ref{sec : terminal unit data structure}. 
\subsubsection{A Data Structure for All Stretched Units} \label{sec : stretched unit data structure}
 Suppose $\mu$ is a Stretched unit of $G$. 
It follows from Lemma \ref{lem : projection mapping of terminal and stretched units}(2) that $\mu$ is projected to a proper path $P$ in the Skeleton. Let us consider an edge $e_{\mu}$ from the proper path $P$. There exists a bunch ${\mathcal B}$ corresponding to a minimal cut of the Skeleton ${\mathcal H}_S$ containing edge $e_{\mu}$. 
By definition of projection mapping, there exists a pair of $S$-mincuts $C,C'$ in ${\mathcal B}$ such that ${\mathcal W}_{\mu}\subseteq C\setminus C'$. We now establish the following lemma to construct a compact data structure for Stretched units.
\begin{lemma} \label{lem : cuts are stored}
    There is a $(\lambda_S+1)$ cut $C$ that subdivides $(\lambda_S+1)$ class ${\mathcal W}_{\mu}$ into $({\mathcal W}_1,{\mathcal W}_{\mu}\setminus {\mathcal W}_1)$ if and only if there exists a $(\lambda_S+1)$ cut $C'$ that subdivides ${\mathcal W}_{\mu}$ into $({\mathcal W}_1,{\mathcal W}_{\mu}\setminus {\mathcal W}_1)$ and $C'$ does not subdivide both $C(S_{\mathcal B})$ and $C(S\setminus S_{\mathcal B})$.
\end{lemma}
\begin{proof}
     Suppose there is a $(\lambda_S+1)$ cut $C$ that subdivides $(\lambda_S+1)$ class ${\mathcal W}_{\mu}$ into $({\mathcal W}_1,{\mathcal W}_{\mu}\setminus {\mathcal W}_1)$. Assume to the contrary that there is no $(\lambda_S+1)$ cut that subdivides ${\mathcal W}_{\mu}$ into $({\mathcal W}_1,{\mathcal W}_{\mu}\setminus {\mathcal W}_1)$ and does not subdivide both $C(S_{\mathcal B})$ and $C(S\setminus S_{\mathcal B})$. So, cut $C$ must subdivide at least one of $C(S_{\mathcal B})$ and $C(S\setminus S_{\mathcal B})$. We first show that if $C$ subdivides $C(S_{\mathcal B})$, then $C$ must subdivide the Steiner set $S_{\mathcal B}$. The statement holds for $C(S\setminus S_{\mathcal B})$ as well. Assume to the contrary that $C$ does not subdivide $S_{\mathcal B}$. Let $C'=C(S_{\mathcal B})$. By Definition \ref{def : tight cut}, both $C(S_{\mathcal B})$ and $C(S\setminus S_{\mathcal B})$ are $S$-mincuts. Without loss of generality, assume that $S_{\mathcal B}\subseteq C$; otherwise, consider $\overline{C}$. So, there must exist a Steiner vertex $t\in \overline{C}$ such that $t$ is from $S\setminus S_{\mathcal B}$. This implies that $S_{\mathcal B}\subseteq C\cap C'$ and there is a Steiner vertex $t$ such that $t\notin C\cup C'$. Therefore, both $C\cap C'$ and $C\cup C'$ are Steiner cuts. For $S$-mincut $C'$ and $(\lambda_S+1)$ cut $C$, by sub-modularity of cuts (Lemma \ref{submodularity of cuts}(1)), we have the following equation.
    \begin{equation} \label{eq : in lemma 51}
        c(C\cap C')+c(C\cup C')\le 2\lambda_S+1
    \end{equation}     
    Moreover, observe that $C\cup C'$ subdivides ${\mathcal W}_{\mu}$. So, $C\cup C'$ has capacity at least $\lambda_S+1$. By Equation \ref{eq : in lemma 51}, it implies that $C\cap C'$ is an $S$-mincut. We also have $S_{\mathcal B}\subseteq C\cap C'$ and since $S\setminus S_{\mathcal B}\subseteq \overline{C'}$, there is no Steiner vertex from $S\setminus S_{\mathcal B}$ that belongs to $C\cap C'$. Therefore, $C\cap C'$ belongs to bunch ${\mathcal B}$. However, we have at least one vertex from $C'=C(S_{\mathcal B})$ that is not in $C\cap C'$ because $C$ subdivides $C(S_{\mathcal B})$. Therefore, $C\cap C'$ is a proper subset of $C(S_{\mathcal B})$ that belongs to ${\mathcal B}$, which is a contradiction because of Definition \ref{def : tight cut}.

    It follows that if $C$ subdivides $C(S_{\mathcal B})$ (likewise $C(S\setminus S_{\mathcal B})$), then $C$ also subdivides $S_{\mathcal B}$ (likewise $S\setminus S_{\mathcal B}$). Now there are two possibilities -- (a) $C$ subdivides exactly one of $C(S_{\mathcal B})$ and $C(S\setminus S_{\mathcal B})$ and (b) $C$ subdivides both $C(S_{\mathcal B})$ and $C(S\setminus S_{\mathcal B})$. We now prove for each case.\\

    \noindent
    \textbf{Case (a):} Suppose $C$ subdivides $C(S_{\mathcal B})$. The proof is along a similar line if $C$ subdivides $S_{\mathcal B}$. We know that $C$ must subdivide $S_{\mathcal B}$. Therefore, for cuts $C$ and $C'$, there is a Steiner vertex $s\in (C\cap C') \cap S_{\mathcal B}$ and $S\setminus S_{\mathcal B}\subseteq \overline{C\cup C'}$. Also, cut $C\cup C'$  subdivides ${\mathcal W}_{\mu}$. Since $C'$ does not subdivide ${\mathcal W}_{\mu}$, it is easy to observe that the partition of ${\mathcal W}_{\mu}$ formed by $C\cup C'$ is the same as the partition formed by $C$. Observe that $S$-mincut $C'$ and $(\lambda_S+1)$ cut $C$ satisfies Equation \ref{eq : in lemma 51}.
    Since there is no $(\lambda_S+1)$ cut with the Steiner partition ${\mathcal W}_{\mu}$ same as $C$ and  $C\cup C'$ does not subdivide any of $C(S_{\mathcal B})$ and $C(S\setminus S_{\mathcal B})$, $C\cup C'$ must have capacity at least $\lambda_S+2$. It follows from Equation \ref{eq : in lemma 51} that $c(C\cap C')<\lambda_S$. So, we have a Steiner cut of capacity strictly less than the capacity of $S$-mincut, which is a contradiction. \\

    \noindent
    \textbf{Case (b):}  Suppose $C$ subdivides both $C(S_{\mathcal B})$ and $C(S\setminus S_{\mathcal B})$. This implies that $C$ must subdivide $S_{\mathcal B}$ as well as $S\setminus S_{\mathcal B}$. Let us consider cut $C'=C(S_{\mathcal B})$. It follows that we have a Steiner vertex $s$ from $S_{\mathcal B}$ such that $s\in C\cup C'$. Similarly, since $S\setminus S_{\mathcal B}\subseteq \overline{C'}$, there also exist a Steiner vertex $t$ from $S\setminus S_{\mathcal B}$ such that $t\in \overline{C\cup C'}$. This ensures that $C\cap C'$ and $C\cup C'$ are Steiner cuts. Cut $C$ subdivides ${\mathcal W}_{\mu}$. So, $C\cup C'$ has capacity at least $\lambda_S+1$. Moreover, $C\cap C'$ has capacity at least $\lambda_S$. Observe that $C$ and $C'$ satisfies Equation \ref{eq : in lemma 51}. So, $C\cap C'$ is an $S$-mincut and $C\cup C'$ is a $(\lambda_S+1)$ cut. So, we have a $(\lambda_S+1)$ cut that subdivides ${\mathcal W}_{\mu}$ and does not subdivide $C(S_{\mathcal B})$. Therefore, $C\cup C'$ subdivide only $C(S\setminus S_{\mathcal B})$. In this case, we have shown in Case (a) that there must exist at least one $(\lambda_S+1)$ cut $C''$ with the same partition of ${\mathcal W}_{\mu}$ as $C$ such that $C''$ does not subdivide any of $C(S_{\mathcal B})$ and $C(S\setminus S_{\mathcal B})$. This completes the proof of the forward direction. 

        The proof of the converse part is immediate. 
\end{proof} 
Similar to the construction of graph $G({\mathcal W})$ for a Singleton $(\lambda_S+1)$ class ${\mathcal W}$ in Appendix \ref{sec : data structure for singleton class}, we construct graph $G({\mathcal W}_{\mu})$ for Stretched unit $\mu$ as follows. Graph $G({\mathcal W}_{\mu})$ is obtained from graph $G$ by contracting set $C(S_{\mathcal B})$ into a single vertex $s$ and $C(S\setminus S_{\mathcal B})$ into a single vertex $t$. Note that $s$ and $t$ are Steiner vertices in $G({\mathcal W}_{\mu})$. Lemma \ref{lem : cuts are stored} ensures that, for every partition of ${\mathcal W}_{\mu}$ formed by a $(\lambda_S+1)$ cut of $G$, there is a $(\lambda_S+1)$ cut in $G({\mathcal W}_{\mu})$ with the same partition. However, observe that the $(\lambda_S+1)$ class ${\mathcal W}_{\mu}$ in $G({\mathcal W}_{\mu})$ still remains a Stretched unit. Note that if we can convert ${\mathcal W}_{\mu}$ to a Singleton $(\lambda_S+1)$ class, then by using Lemma \ref{lem : data structure for singleton class}, a data structure can be constructed for $\mu$. To materialize this idea, we apply a covering technique of Baswana, Bhanja, and Pandey \cite{DBLP:journals/talg/BaswanaBP23} on graph $G({\mathcal W}_{\mu})$. A pair of graphs  $G({\mathcal W}_{\mu})^I$ and $G({\mathcal W}_{\mu})^U$ are constructed from $G({\mathcal W}_{\mu})$ as follows.\\

\noindent
\textbf{Construction of $G({\mathcal W}_{\mu})^I$ and $G({\mathcal W}_{\mu})^U$:} Let $x$ be any vertex in ${\mathcal W}_{\mu}$. Graph $G({\mathcal W}_{\mu})^I$ is obtained by adding a pair of edges between vertex $s$ and vertex $x$. Similarly, graph $G({\mathcal W}_{\mu})^U$ is obtained by adding a pair of edges between vertex $x$ and vertex $t$. Without causing any ambiguity, we denote the Steiner set by $S$ for each of the three graphs $G({\mathcal W}_{\mu})$, $G({\mathcal W}_{\mu})^I$, and $G({\mathcal W}_{\mu})^U$. The following property is ensured by Baswana, Bhanja, and Pandey \cite{DBLP:journals/talg/BaswanaBP23}.
\begin{lemma} [Lemma 3.1 in \cite{DBLP:journals/talg/BaswanaBP23}]
    $C$ is $(\lambda_S+1)$ cut in $G({\mathcal W}_{\mu})$ if and only if $C$ is a $(\lambda_S+1)$ cut either in $G({\mathcal W}_{\mu})^I$ or in $G({\mathcal W}_{\mu})^U$.
\end{lemma}
Observe that, in the Flesh graph of  $G({\mathcal W}_{\mu})^I$ (likewise of $G({\mathcal W}_{\mu})^U$), the $(\lambda_S+1)$ class containing ${\mathcal W}_{\mu}\cup \{s\}$ (likewise ${\mathcal W}_{\mu}\cup \{t\}$) is a Singleton $(\lambda_S+1)$ class. Hence, in a similar way as Singleton $(\lambda_S+1)$ class, we store a pair of data structure ${\mathcal D}_s$ of Lemma \ref{lem : data structure for singleton class} -- one for $G({\mathcal W}_{\mu})^I$ and the other for $G({\mathcal W}_{\mu})^U$ for each $(\lambda_S+1)$ class ${\mathcal W}_{\mu}$ of $G$, where the the corresponding node of ${\mathcal W}_{\mu}$ in Flesh is a Stretched unit. 
Therefore, the following lemma holds for all Stretched units of $G$.
\begin{lemma} \label{lem : data structure for stretched units}
    There exists an ${\mathcal O}(n(n-|S|+1))$ space data structure ${\mathcal S}_s$ that, given any pair of vertices $u$ and $v$ belonging to any $(\lambda_S+1)$ class corresponding to a Stretched unit, can determine whether there exists a $(\lambda_S+1)$ cut $C$ such that $u\in C$ and $v\in \overline{C}$ in ${\mathcal O}(1)$ time. Moreover, it can report cut $C$ in ${\mathcal O}(n)$ time.
\end{lemma}

\subsubsection{A Data Structure for All Terminal Units with no Steiner Vertex} \label{sec : terminal unit data structure} 
Let $\mu$ be a terminal unit and the corresponding $(\lambda_S+1)$ class of ${\mathcal W}_{\mu}$ contains no Steiner vertex. Since $\mu$ is a terminal unit, by Lemma \ref{lem : projection mapping of terminal and stretched units}(1), there is a node $N$ in the Skeleton that represents $\mu$. If node $N$ has degree $1$ with edge $e$ adjacent to node $N$, then it must be a Steiner node because the bunch corresponding to minimal cut defined by edge $e$ represents $S$-mincuts. Suppose node $N$ has degree $2$. The two adjacent edges of $N$ cannot be two cycle-edges from the same cycle because of the same reason. So, it is adjacent to two different tree-edges. In this case, observe that $\mu$ becomes a Stretched unit. Therefore, it follows from the construction of Skeleton \cite{DBLP:conf/stoc/DinitzV94, DBLP:conf/soda/DinitzV95, DBLP:journals/siamcomp/DinitzV00} that node $N$ has degree at least $3$. Using the bunches adjacent to node $N$, we construct the graph $G({\mathcal W}_{\mu})$ (Appendix \ref{sec : data structure for singleton class}) for ${\mathcal W}_{\mu}$. Let ${\mathcal W}$ be the $(\lambda_S+1)$ class of $G({\mathcal W}_{\mu})$ to which all vertices of $(\lambda_S+1)$ class to ${\mathcal W}_{\mu}$ of $G$ is mapped. It follows that ${\mathcal W}$ does not contain any Steiner vertex. However, by construction of $G({\mathcal W}_{\mu})$, for each edge $(u,v)$ in $G({\mathcal W}_{\mu})$ with $u\in {\mathcal W}$ and $v\notin {\mathcal W}$, $v$ is a Steiner vertex. 

Let $x$ be any vertex in ${\mathcal W}$ and $s,t$ be a pair of Steiner vertices in $G({\mathcal W}_{\mu})$. Similar to the Stretched unit case, we construct a pair of graphs $G({\mathcal W}_{\mu})^I$ and $G({\mathcal W}_{\mu})^U$ using the covering technique of Baswana, Bhanja, and Pandey \cite{DBLP:journals/talg/BaswanaBP23} as follows. Graph $G({\mathcal W}_{\mu})^I$ is obtained by adding a pair of edges between vertex $s$ and vertex $x$. Similarly, graph $G({\mathcal W}_{\mu})^U$ is obtained by adding a pair of edges between vertex $x$ and vertex $t$. The remaining construction is the same as of Stretched units (Lemma \ref{lem : data structure for stretched units}), which leads to the following lemma.

\begin{lemma}
    \label{lem : data structure for terminal nonSteiner units}
    Let ${\mathbb W}_T$ be the set of all $(\lambda_S+1)$ classes of $G$ such that for any $(\lambda_S+1)$ class ${\mathcal W}\in {\mathbb W}_T$, ${\mathcal W}\cap S=\emptyset$ and ${\mathcal W}$ corresponds to a terminal unit of $G$. There exists an ${\mathcal O}(n(n-|S|+1))$ space data structure ${\mathcal S}_T$ that, given any pair of vertices $u$ and $v$ from any $(\lambda_S+1)$ class belonging to ${\mathbb W}_T$, can determine whether there exists a $(\lambda_S+1)$ cut $C$ such that $u\in C$ and $v\in \overline{C}$ in ${\mathcal O}(1)$ time. Moreover, it can report cut $C$ in ${\mathcal O}(n)$ time.
\end{lemma}

\subsection{A Data Structure for All Steiner Nodes} \label{sec : generic steiner node}
Let $\mu$ be a Steiner node and ${\mathcal W}$ be the corresponding $(\lambda_S+1)$ class. If ${\mathcal W}$ contains exactly one Steiner vertex, then we have data structure ${\mathcal D}_s$ of Lemma \ref{lem : data structure for singleton class}. So, we assume that ${\mathcal W}$ contains at least two Steiner vertices. 

Let ${\mathcal C}({\mathcal W})$ be the set of all $(\lambda_S+1)$ cuts of $G$ that subdivides ${\mathcal W}$. A $(\lambda_S+1)$ cut $C\in {\mathcal C}_{\mathcal W}$ can be of two types -- either $C$ separates at least one pair of Steiner vertices belonging to ${\mathcal W}$ or $C$ keeps all the Steiner vertices belonging to ${\mathcal W}$ on the same side. Based on this insight, we can form two disjoint subsets ${\mathcal C}({\mathcal W})_1$ and ${\mathcal C}({\mathcal W})_2$ of ${\mathcal C}({\mathcal W})$ such that ${\mathcal C}({\mathcal W})_1 \cup {\mathcal C}({\mathcal W})_2={\mathcal C}({\mathcal W})$. A cut $C\in {\mathcal C}({\mathcal W})$ belongs to ${\mathcal C}({\mathcal W})_2$ if it separates at least a pair of Steiner vertices belonging to ${\mathcal W}$, otherwise $C$ belongs to ${\mathcal C}({\mathcal W})_1$. We now construct a pair of graphs $G({\mathcal W})_1$ and $G({\mathcal W})_2$ such that $G({\mathcal W})_1$ preserves all cuts from ${\mathcal C}({\mathcal W})_1$ and $G({\mathcal W})_2$ preserves all cuts from ${\mathcal C}({\mathcal W})_2$.\\

\noindent
\textbf{Construction of $G({\mathcal W})_1$ and $G({\mathcal W})_2$:} Let $N$ be the node of the Skeleton ${\mathcal H}_S$ that represents the $(\lambda_S+1)$ class ${\mathcal W}$. Using the bunches adjacent to node $N$, we construct graph $G({\mathcal W})$ (Appendix \ref{sec : data structure for singleton class}) for ${\mathcal W}$. Graph $G({\mathcal W})_1$ is obtained from $G({\mathcal W})$ by contracting all the Steiner vertices of ${\mathcal W}$ into a single Steiner vertex $s$. Graph $G({\mathcal W})_2$ is the same graph as $G({\mathcal W})$ but the Steiner set for $G({\mathcal W})_2$ contains only the Steiner vertices that belong to ${\mathcal W}$. \\

\noindent
Let ${\mathcal W}_1$ be the $(\lambda_S+1)$ class of $G({\mathcal W})_1$ to which all vertices of ${\mathcal W}$ are mapped. 
By construction and Lemma \ref{lem : G(W)}, graph $G({\mathcal W})_1$ has exactly one Steiner vertex $s$ in ${\mathcal W}_1$ and satisfies the following property.
\begin{lemma} \label{lem : all cuts in gw1}
    $C$ is a $(\lambda_S+1)$ cut in ${\mathcal C}({\mathcal W})_1$ if and only if $C$ is a  Steiner cut of capacity $(\lambda_S+1)$ that subdivides ${\mathcal W}_1$.     
\end{lemma}
Since $s$ is the only Steiner vertex in ${\mathcal W}_1$, therefore, in $G({\mathcal W})_1$, ${\mathcal W}_1$ is a Singleton $(\lambda_S+1)$ class. Let ${\mathbb W}$ be the set of all  $(\lambda_S+1)$ classes of $G$ that contain more than one Steiner vertices. So, similar to ${\mathcal W}_1$, we can construct ${\mathcal W}_1^i$ for every $(\lambda_S+1)$ classes ${\mathcal W}_1^i\in {\mathbb W}$. Since each ${\mathcal W}_1^i$ is a Singleton $(\lambda_S+1)$ class, we store the data structure ${\mathcal D}_s$ of Lemma \ref{lem : data structure for singleton class} for all of them. This, using Lemma \ref{lem : all cuts in gw1} and Lemma \ref{lem : data structure for singleton class}, completes the proof of the following lemma. 
\begin{lemma} \label{lem : data structure for G1}
    There exists an ${\mathcal O}(n(n-|S|+1))$ space data structure that, given any pair of vertices $u,v$ from any $(\lambda_S+1)$ class ${\mathcal W}$ corresponding to a Steiner node, can determine in ${\mathcal O}(1)$ time whether there exists a $(\lambda_S+1)$ cut $C$ of $G$ that does not subdivide the Steiner set belonging to ${\mathcal W}$ and  separates vertices $u,v$. Moreover, $C$ can be reported in ${\mathcal O}(n)$ time. 
\end{lemma} 
We now consider graph $G({\mathcal W})_2$. Let $S_2$ be the Steiner set of $G({\mathcal W})_2$. By construction, the capacity of Steiner mincut of $G({\mathcal W})_2$ is $\lambda_S+1$ and satisfies the following property.
\begin{lemma} \label{lem : all cuts GW2}
    $C$ is a $(\lambda_S+1)$ cut in ${\mathcal C}({\mathcal W})_2$ if and only if $C$ is a Steiner mincut in $G({\mathcal W})_2$.
\end{lemma}
Exploiting Lemma \ref{lem : all cuts GW2} and the following data structure by Baswana and Pandey \cite{DBLP:conf/soda/BaswanaP22}, we construct a data structure for reporting a $(\lambda_S+1)$ cut from ${\mathcal C}({\mathcal W})_2$. Unfortunately, the following theorem is not explicitly mentioned in \cite{DBLP:conf/soda/BaswanaP22}; hence, a proof is provided in Appendix \ref{app : proof of theorem 58 for reporting S mincut separating vertices} for completeness.
\begin{theorem} [\cite{DBLP:conf/soda/BaswanaP22}] \label{thm : reporting $S$-mincut using skeleton and projection mapping}
    For any undirected multi-graph $G=(V,E)$ on $n$ vertices and $m$ edges with a Steiner Set $S\subseteq V$, the ${\mathcal O}(n)$ space Skeleton ${\mathcal H}(G)$, Projection Mapping $\pi(G)$ and an Ordering $\tau(G)$ of nonSteiner vertices can be used to report in ${\mathcal O}(1)$ time whether there exists an $S$-mincut $C$ that separates $u,v$ and can report $C$ in ${\mathcal O}(n)$ time.
\end{theorem}

\noindent
\textbf{Construction of Data structure $R({\mathcal W})$:} For graph $G({\mathcal W})_2$, we store the Skeleton ${\mathcal H}({G({\mathcal W})_2})$ on Steiner set $S_2$. Similarly, the projection mapping $\pi({G({\mathcal W})_2})$ and the Ordering $\tau({G({\mathcal W})_2})$ of all nonSteiner vertices is also stored for graph ${G({\mathcal W})_2}$ from Theorem \ref{thm : reporting $S$-mincut using skeleton and projection mapping}. Evidently, Skeleton ${\mathcal H}({G({\mathcal W})_2})$ occupies ${\mathcal O}(|S_2|)$ space. The number of nonSteiner vertices of $G({\mathcal W})_2$ is the sum of the number of the nonSteiner vertices belonging to ${\mathcal W}$ in $G$ and $\text{deg}(N)$ in ${\mathcal H}_S$, where node $N$ represents ${\mathcal W}$ in ${\mathcal H}_S$. Let $n_2$ be the number of nonSteiner vertices of ${\mathcal W}$. Therefore, projection mapping $\pi({G({\mathcal W})_2})$ and the Ordering $\tau({G({\mathcal W})_2})$ occupies ${\mathcal O}(n_2+\text{deg}(N))$ space. Recall that, in the construction of graph $G({\mathcal W})$, for every bunch ${\mathcal B}$ adjacent to node $N$ in ${\mathcal H}_S$, there is a Steiner vertex $b$ in $G({\mathcal W})$. By construction of $G({\mathcal W})_2$, $b$ is a nonSteiner vertex in $G({\mathcal W})_2$. For each such vertex $b$, similar to data structure ${\mathcal D}_s$ in Lemma \ref{lem : data structure for singleton class}, we also keep the mapping from $b$ to the tree-edge/a pair of cycle-edges that represents the bunch ${\mathcal B}$. 
It occupies ${\mathcal O}(\text{deg}(N))$ space. So, the space occupied by the obtained data structure $R({\mathcal W})$ is ${\mathcal O}(|S_2|+n_2+\text{deg}(N))$.  

By Theorem \ref{thm : reporting $S$-mincut using skeleton and projection mapping} and Lemma \ref{lem : all cuts GW2}, the data structure $R({\mathcal W})$, given any pair of vertices $u,v$, can determine in ${\mathcal O}(1)$ time whether there exists a Steiner mincut $C$ of $G({\mathcal W})_2$ that separates $u$ and $v$. Moreover, it can report $C$ in ${\mathcal O}(|S_2|+n_2+\text{deg}(N))$ time for graph $G({\mathcal W})_2$. By using cut $C$, we now report a $(\lambda_S+1)$ cut $C'$ in $G$ such that $C'$ separates $u,v$.   
Cut $C$ contains every nonSteiner vertex $x$ of $G({\mathcal W})_2$ which stores tree-edge or pair of cycle-edges of ${\mathcal H}_S$. Cut $C'$ contains all the vertices belonging to $C$ except these nonSteiner vertices. If $x\in C$ and $x$ stores a tree-edge or a pair of cycle-edges of ${\mathcal H}_S$, then we remove those edges from ${\mathcal H}_S$ and mark the subgraph of ${\mathcal H}_S$ that does not contain node $N$. It requires ${\mathcal O}(n)$ time as discussed in Appendix \ref{sec : data structure for all singleton classes}. Finally, for every vertex $x$ of $G$, if $\pi(\phi(x))$ has a nonempty intersection with any marked subgraph, then report $x$ as $x\in C'$. Therefore, the overall time taken for reporting $C'$ is ${\mathcal O}(n)$.

Let ${\mathbb W}$ be the set of all $(\lambda_S+1)$ classes of $G$ that contain more than one Steiner vertices. We augment data structure $R({\mathcal W})$ for graph $G({\mathcal W})_2$ for every $(\lambda_S+1)$ class ${\mathcal W}\in {\mathbb W}$ of $G$. For each ${\mathcal W}_i\in {\mathbb W}$, let $S_i$ be the number of Steiner vertices in ${\mathcal W}_i$, $n_i$ be the number of nonSteiner vertices of ${\mathcal W}_i$, and $N_i$ be the node of the Skeleton ${\mathcal H}_S$ that represents ${\mathcal W}_i$. Since $\Sigma_{\forall{\mathcal W}\in {\mathbb W}} |S_i|={\mathcal O}(|S|)$, $\Sigma_{\forall{\mathcal W}\in {\mathbb W}} |n_i|={\mathcal O}(n-|S|)$, and $\Sigma_{\forall{\mathcal W}\in {\mathbb W}} \text{deg}(N_i)={\mathcal O}(|S|)$, the overall space occupied by the data structure is ${\mathcal O}(n)$. This completes the proof of the following lemma. 
\begin{lemma} \label{lem : data structure for G2}
    There exists an ${\mathcal O}(n)$ space data structure that, given any pair of vertices $u,v$ from any $(\lambda_S+1)$ class ${\mathcal W}$, can determine in ${\mathcal O}(1)$ time whether there exists a $(\lambda_S+1)$ cut $C$ of $G$ that subdivides the Steiner set belonging to ${\mathcal W}$ and separates $u,v$. Moreover, $C$ can be reported in ${\mathcal O}(n)$ time. 
\end{lemma}
Lemma \ref{lem : data structure for G1} and Lemma \ref{lem : data structure for G2} leads to the following lemma for Steiner nodes.
\begin{lemma} \label{lem : data structure for generic steiner node}
       There exists an ${\mathcal O}(n(n-|S|+1))$ space data structure that, given any pair of vertices $u$ and $v$ belonging to any $(\lambda_S+1)$ class corresponding to a Steiner node, can determine whether there exists a $(\lambda_S+1)$ cut $C$ such that $u\in C$ and $v\in \overline{C}$ in ${\mathcal O}(1)$ time. Moreover, it can report cut $C$ in ${\mathcal O}(n)$ time.
\end{lemma}
Finally, data structures in Lemma \ref{lem : data structure for singleton class}, Lemma \ref{lem : data structure for stretched units}, Lemma \ref{lem : data structure for terminal nonSteiner units}, and Lemma \ref{lem : data structure for generic steiner node} complete the proof of Theorem \ref{thm : data structure}.

\section{Sensitivity Oracle: Dual Edge Failure} \label{sec : dual edge oracle}
In this Section, we address the problem of designing a compact data structure that, after the failure of any pair of edges in $G$, can efficiently report an $S$-mincut and its capacity. 






Let $e=(x,y)$ and $e'=(x',y')$ be the two failed edges in $G$. As stated in Section \ref{sec : introduction}, the capacity of $S$-mincut decreases after the failure of a pair of edges if and only if at least one of the failed edges contributes to an $S$-mincut or both failed edges contribute to a single $(\lambda_S+1)$ cut. 
An edge is said to belong to a node in Flesh ${\mathcal F}_S$ if both endpoints of the edge are mapped to the node.
Depending on the quotient mapping ($\phi$) of the endpoints of both failed edges, the following four cases are possible.
\begin{itemize}
    \item \textbf{Case 1:} $\phi(x)=\phi(y)=\phi(x')=\phi(y')$, both edges belong to the same node of the Flesh. 
    \item \textbf{Case 2:} $(\phi(x)=\phi(y))\ne(\phi(x')=\phi(y'))$, both edges belong to different nodes of the Flesh.
    \item \textbf{Case 3:} Either $\phi(x)\ne \phi(y)$ or $\phi(x') \ne \phi(y')$, exactly one of the two edges belongs to a node of the Flesh.
    \item \textbf{Case 4:} $\phi(x) \ne \phi(y)$ and $\phi(x') \ne \phi(y')$, none of the edges belong to any node of the Flesh.
\end{itemize}
We now handle each case separately.\\

\noindent
In \textbf{Case 1}, the failed edges are belonging to the same node or the same $(\lambda_S+1)$ class. In this case, the $S$-mincut capacity decreases by $1$ if and only if both failed edges are contributing to a single $(\lambda_S+1)$ cut. Recall that a node of the Flesh graph can be a Steiner or a nonSteiner node. We handle Steiner and nonSteiner nodes separately in Appendix \ref{sec : Steiner node} and Appendix \ref{sec : nonSteiner node}.\\


\noindent
In \textbf{Case 2}, the two failed edges belong to different $(\lambda_S+1)$ classes.
It follows from the construction of the Flesh graph that any Steiner cut that subdivides a $(\lambda_S+1)$ class must have capacity at least $\lambda_S+1$. Exploiting this fact and sub-modularity of cuts (Lemma \ref{submodularity of cuts}(1)), we show that, in Case 2, the capacity of $S$-mincut remains unchanged.
\begin{lemma} \label{lem : subdivides at most one class}
    There is no $(\lambda_S+1)$ cut that subdivides more than one $(\lambda_S+1)$ class.
\end{lemma}
\begin{proof}
    Assume to the contrary that $C$ is a $(\lambda_S+1)$ cut that subdivides two $(\lambda_S+1)$ classes ${\mathcal W}_1$ and $\mathcal W_2$. Since ${\mathcal W}_1$ is different from ${\mathcal W}_2$, there must exist an $S$-mincut $C'$ that separates them. Observe that either $C\cap C'$ and $C\cup C'$ are Steiner cuts or $C\setminus C'$ and $C'\setminus C$ are Steiner cuts; otherwise, it can be shown that one of the two cuts $C$ and $C'$ is not a Steiner cut. Without loss of generality, assume that $C\cap C'$ and $C\cup C'$ are Steiner cuts; otherwise, we can consider the Steiner cut $\overline{C}$ instead of $C$. Since $C'$ keeps ${\mathcal W}_1$ and ${\mathcal W}_2$ on different sides, it follows that one of the two Steiner cuts $C\cap C'$ and $C\cup C'$ subdivides the $(\lambda_S+1)$ class ${\mathcal W}_1$ and the other Steiner cut subdivides the $(\lambda_S+1)$ class ${\mathcal W}_2$. Hence, both cuts $C\cap C'$ and $C\cup C'$ have capacity at least $\lambda_S+1$. Therefore, $c(C\cap C')+c(C\cup C')\ge 2\lambda+2$. However, for cuts $C,C'$, by sub-modularity of cuts (Lemma \ref{submodularity of cuts}(1)), $c(C\cap C')+c(C\cup C')\le 2\lambda_S+1$, a contradiction.  
\end{proof}

 By Lemma \ref{lem : subdivides at most one class}, since there is no $S$-mincut or $(\lambda_S+1)$ cut that subdivides more than one $(\lambda_S+1)$ class, therefore, it is immediate that the capacity of $S$-mincut remains the same. \\

\noindent
In \textbf{Case 3}, exactly one of the two failed edges does not belong to any $(\lambda_S+1)$ class. By Fact \ref{fact : quotient mapping}, there is an $S$-mincut in which the other failed edge, say $e'$, is contributing. Therefore, the capacity of $S$-mincut decreases by exactly $1$. To report an $S$-mincut for the resulting graph, we use the following single edge Sensitivity Oracle given by Baswana and Pandey \cite{DBLP:conf/soda/BaswanaP22}. (complete details are also provided in Appendix \ref{app : single edge insertion} and Appendix \ref{app : proof of theorem 58 for reporting S mincut separating vertices}).
\begin{theorem} [\cite{DBLP:conf/soda/BaswanaP22}] \label{thm : single edge Sensitivity Oracle} 
    For any undirected multi-graph $G$ on $n$ vertices with a Steiner set $S$, there is an ${\mathcal O}(n)$ space data structure that can report the capacity of $S$-mincut and an $S$-mincut in ${\mathcal O}(1)$ time and ${\mathcal O}(n)$ time, respectively, after the failure of any edge in $G$. 
\end{theorem}
Theorem \ref{thm : single edge Sensitivity Oracle} ensures that in Case 3, there is an ${\mathcal O}(n)$ space data structure that 
can report the capacity of $S$-mincut in ${\mathcal O}(1)$ time and an $S$-mincut in ${\mathcal O}(n)$ time.\\

\noindent
In \textbf{Case 4},
similar to Case $3$, the capacity of $S$-mincut definitely decreases by at least $1$.
 Note that the capacity of $S$-mincut decreases by $2$ if and only if both failed edges are contributing to a single $S$-mincut. Therefore, we want to determine if both failed edges $e,e'$ are contributing to a single $S$-mincut. Otherwise, if $S$-mincut decreases by exactly $1$, then single edge Sensitivity Oracle in Theorem \ref{thm : single edge Sensitivity Oracle} is sufficient to report an $S$-mincut for the resulting graph in ${\mathcal O}(n)$ time. We design a data structure in Appendix \ref{sec : projected to skeleton} that can determine whether there exists a single $S$-mincut $C$ such that both $e,e'$ are contributing to $C$. Moreover, it can also efficiently report a resulting $S$-mincut, if $C$ exists. 




\subsection{Both Failed Edges are Projected to Skeleton} \label{sec : projected to skeleton}
Suppose none of the failed edges $e,e'$ belong to any node of the Flesh graph. 
By Lemma \ref{lem : edge is projected to a proper path}, both of them are projected to proper paths in the Skeleton. To determine whether both failed edges are contributing to a single $S$-mincut, by Lemma \ref{lem : property of skeleton} and Lemma \ref{lem : edge is projected to a proper path}, the necessary condition is the following. There is at least one tree-edge/cycle-edge of ${\mathcal H}_S$ that belongs to both $\pi(e)$ and $\pi(e')$ or there is a pair of cycle-edges from the same cycle in ${\mathcal H}_S$ such that one edge belongs to $\pi(e)$ and the other edge belongs to $\pi(e')$. By Lemma \ref{lem : lca queries on skeleton}, this can be determined in ${\mathcal O}(1)$ time.
Moreover, the data structure in Lemma \ref{lem : lca queries on skeleton} can also report in ${\mathcal O}(1)$ time one of the required tree-edge $e_1$ or a pair of cycle-edges $e_1=(A,A'),~e_2=(B,B')$ from the same cycle. 

We first show that if the pair of cycle-edges $e_1,e_2$ are different, then there is a single $S$-mincut in which both failed edges $e,e'$ are contributing.  
However, in the other cases, it is not guaranteed that they always contribute to a single $S$-mincut. Let ${\mathcal B}$ be the bunch corresponding to a minimal cut $\tilde C$ of ${\mathcal H}_S$. Without loss of generality, assume that $A,B\in \tilde C$, $A',B'\notin \tilde C$, and Steiner vertices in $S_{\mathcal B}$ are mapped to nodes that belong to $\tilde C$. We say that an $S$-mincut $C'$ from bunch ${\mathcal B}$ contains a node $P$ of ${\mathcal H}_S$ if $P\in \tilde C$, otherwise, $C'$ does not contain $P$.
\begin{lemma} \label{lem : different cycle edges}
    Let there be a pair of cycle-edges $e_1=(A,A'),~e_2=(B,B')$, $e_1\ne e_2$, from the same cycle $O$ of Skeleton ${\mathcal H}_S$ such that $e_1\in \pi(e)$ and $e_2\in \pi(e')$. Then, there is an $S$-mincut of $G$ in which both edges $e,e'$ are contributing. 
\end{lemma}
\begin{proof}
    Let ${\mathcal B}$ be the bunch corresponding to the minimal cut $\tilde C$ defined by the cycle-edges $\{e_1,e_2\}$ of Skeleton ${\mathcal H}_S$.  Without loss of generality, assume that $A,B\in \tilde C$ and $A',B'\notin {\tilde C}$ and the nodes of the cycle appear in the order $A,A',B'$, and $B$.
    Since $e_1\in \pi(e)$ and $e_2\in \pi(e')$, by definition of projection of an edge, 
    there is at least one $S$-mincut $C_1$ from ${\mathcal B}$ in which edge $e$ is contributing and at least one $S$-mincut $C_2$ from ${\mathcal B}$ in which $e'$ is contributing. 
    If $e'$ is contributing to $C_1$ (likewise $e$ is contributing to $C_2$), then $C_1$ (likewise $C_2$) is an $S$-mincut in which both edges are contributing. Suppose this is not the case. 
    Without loss of generality, assuming node $A',B'\in \overline{C_1}$, observe that $C_1$ must keep both $x'$ and $y'$ either in $C_1$ or in $\overline{C_1}$. 
    
    Suppose $x',y'\in C_1$. Since $e_1,e_2$ are two different cycle-edges from the same cycle, without loss of generality, assume that $A'\ne B'$. Let $e_3$ be a cycle-edge belonging to the path between $A'$ and $B'$ that does not have nodes $A$ and $B$. Let ${\mathcal B}'$ be the bunch corresponding to the minimal cut $C''$ 
    defined by the cycle-edges $e_2$ and $e_3$.
    By definition of projection mapping of an edge, there is an $S$-mincut $C_3$ belonging to ${\mathcal B}'$ such that $e'$ is contributing to $C_3$. Without loss of generality, assume that $x',B$ is in $C_3$ and $y',B'$ does not belong to $C_3$. Observe that $A,B\in C_1\cap C_3$ and $B'\notin C_1\cup C_3$. So, Steiner vertices belonging to tight cut $C(B,(B,B'),(B,B''))$ (refer to Definition \ref{lem : tight cut reporting}) also belong to $C_1\cap C_3$, where $(B,B'')$ is the other adjacent cycle-edge of node $B$ in cycle $O$. Similarly, there also exist Steiner vertices in $C(B',(B',B),(B',C))$ that belong to $\overline{C_1\cup C_3}$, where $(B',C)$ is the other adjacent cycle-edge of node $B'$ in cycle $O$. Hence, both $C_1\cap C_3$ and $C_1\cup C_3$ are Steiner cuts. Moreover, by sub-modularity of cuts (Lemma \ref{submodularity of cuts}), we have $c(C_1\cap C_3)+c(C_1\cup C_3)\le \lambda_S$. So, $C_1\cap C_3$ and $C_1\cup C_3$ are $S$-mincuts. Therefore, we get an $S$-mincut $C_1\cap C_3$ in which both edges $e$ and $e'$ are contributing. \\ 
    Now suppose $x',y'\in \overline{C_1}$. Without loss of generality, assume that $A,B\in C_2$. There are two possibilities -- either $x,y\in C_2$ or $x,y\notin C_2$. If $C_2$ does not contain $x,y$ then, in a similar way as we established that $C_1\cup C_3$ is an $S$-mincut, we can also argue that $C_1\cup C_2$ is an $S$-mincut in which both edges $e,e'$ are contributing. We now consider the other case when $C_2$ contains both $x,y$. We can consider edge $e_3$ again. By definition of projection mapping of an edge, there is an $S$-mincut $C_4$ belonging to a bunch ${\mathcal B}''$ such that edge $e$ is contributing. In a similar way as the proof of the case when $x',y'\in C_1$, 
    we can also show that 
    $C_1\cap C_4$ is an $S$-mincut in which both edges $e,e'$ are contributing. Hence, it completes the proof.  
\end{proof}
It follows from Lemma \ref{lem : different cycle edges} that if $e$ and $e'$ are projected to two different cycle-edges of the same cycle in ${\mathcal H}_S$, then the capacity of $S$-mincut decreases by exactly $2$. In this case, we now explain the procedure for reporting an $S$-mincut in which both edges $e$ and $e'$ are contributing. Let ${\mathcal B}$ be the bunch represented by the minimal cut $C_{e_1,e_2}$, where $C_{e_1,e_2}$ is defined by the two cycle-edges $e_1$ and $e_2$ of the Skeleton ${\mathcal H}_S$ from the same cycle. By Lemma \ref{lem : tight cut reporting}, we can report the set $C(S_{\mathcal B})$ in ${\mathcal O}(n)$ time. 
If $x,x'\in C(S_{\mathcal B})$ or $y,y'\in C(S_{\mathcal B})$, then $C(S_{\mathcal B})$ is an $S$-mincut in which both edges $e,e'$ are contributing. Similarly, if $x,x'\in C(S\setminus S_{\mathcal B})$ or $y,y'\in C(S\setminus S_{\mathcal B})$, then $C(S\setminus S_{\mathcal B})$ is an $S$-mincut in which both edges $e,e'$ are contributing.
It is quite possible that either both endpoints of at least one of the two failed edges $e,e'$ are projected to the two edges $e_1,e_2$ or exactly one endpoint of one edge belongs to $C(S_{\mathcal B})$ and exactly one endpoint of the other edge belongs to $C(S\setminus S_{\mathcal B})$.
However, we do not have any nontrivial data structure to report an $S$-mincut for these two cases. To overcome this hurdle, we design a compact data structure as follows.
We begin by defining the nearest $S$-mincut of a vertex from a Steiner vertex $s$ to a Steiner vertex $t$.
\begin{definition} [Nearest $(s,t)$-mincut of a vertex] \label{def : nearest s,t mincut of a vertex}
     An $S$-mincut $C$ is said to be a nearest $(s,t)$-mincut for a vertex $u$ if $s,u\in C$ and $t\in \overline{C}$ and there exists no $S$-mincut $C'$ such that $C'\subset C$ with $s,u\in C'$ and $t\in \overline{C'}$.
\end{definition}
For a pair of $S$-mincut, the following fact is immediate.
\begin{fact}[Fact 2.2 in \cite{DBLP:journals/siamcomp/DinitzV00}] \label{fact : closed under intersection and union}
    For a pair of $S$-mincuts $C_1,C_2$ with a Steiner vertex $s_1\in C_1\cap C_2$ and a Steiner vertex $s_2\in \overline{C_1\cup C_2}$, both $C_1\cap C_2$ and $C_1\cup C_2$ are $S$-mincuts. 
\end{fact}
Exploiting Fact \ref{fact : closed under intersection and union}, it is a simple exercise to show that nearest $(s,t)$-mincut of a vertex is always unique.

Suppose $u$ and $v$ are two vertices in $G$ and there exists a pair of edges $f_1=(A_1,A_1'),~f_2=(B_2,B_2')$ of the Skeleton such that $f_1,f_2\in \pi(\phi(u)) \cap \pi(\phi(v))$ and $f_1\ne f_2$. Since vertices $u,v$ are projected to both $f_1$ and $f_2$, by Lemma \ref{lem : projection mapping of terminal and stretched units}(2), $f_1$ and $f_2$ belong to a single proper path. Without loss of generality, assume that the proper path between $A_1$ and $B_2'$ contains the nodes $A_1'$ and $B_2$ (note that $A_1'$ can be same as $B_2$). Let $C_{f_1}$ (likewise $C_{f_2}$) be a minimal cut of ${\mathcal H}_S$ containing edge $f_1$ (likewise $f_2$) such that $A_1\in C_{f_1}$ (likewise $B_2\in C_{f_2}$). Let ${\mathcal B}_{f_1}$ (likewise ${\mathcal B}_{f_2}$) be the corresponding bunch of $C_{f_1}$ (likewise $C_{f_2}$). 
Suppose $s_1,s_2$ are a pair of Steiner vertices such that $\pi(\phi(s_1))$ is a node belonging to $C_{f_1}$ and $\pi(\phi(s_2))$ is a node belonging to $C_{f_1}$. Observe that $s_1,s_2$ always exists because there exist at least one $S$-mincut in each of the two bunches ${\mathcal B}_{f_1}$ and ${\mathcal B}_{f_2}$ corresponding to minimal cuts $C_{f_1}$ and $C_{f_2}$, respectively. The following property on nearest $(s_1,s_2)$-mincut of vertices follows immediately from the works of Dinitz and Vainshtein \cite{DBLP:conf/stoc/DinitzV94, DBLP:conf/soda/DinitzV95, DBLP:journals/siamcomp/DinitzV00} 
 (proof is provided in Appendix \ref{app : respecting nearest cuts} for the sake of completeness). 
\begin{lemma}[\cite{DBLP:conf/stoc/DinitzV94, DBLP:conf/soda/DinitzV95, DBLP:journals/siamcomp/DinitzV00}] \label{lem : respecting nearest cuts}
     Let $C_1$ and $C_2$ are nearest $(s_1,s_2)$-mincuts of $u$ belonging to bunch ${\mathcal B}_{f_1}$ and ${\mathcal B}_{f_2}$ respectively. Then, $v\in C_1$ if and only if $v\in C_2$.  
\end{lemma}
Without causing any ambiguity, we say that $u$ belongs to \textit{nearest $(A_1,B_2')$-mincut of $v$} if $u$ belongs to nearest $(s_1,s_2)$-mincut of $v$. Exploiting Lemma \ref{lem : respecting nearest cuts}, we construct the following matrix for every Stretched unit of $G$.\\

\noindent
\textbf{Construction of Matrix ${\mathcal M}$:} Let $Z_1$ and $Z_2$ be two sets of Stretched units. Matrix ${\mathcal M}$ has $|Z_1|$ rows and $|Z_2|$ columns. Suppose $\mu\in Z_1$ and $\nu\in Z_2$ are a pair of Stretched units such that there is at least one edge in the Skeleton that belongs to both $\pi(\mu)$ and $\pi(\nu)$. ${\mathcal M}[\mu][\nu]$ contains ordered pair $(N,M)$ if $\nu$ belongs to a nearest $(N,M)$-mincut of $\mu$, where $<N,M>$ is $\pi(\mu)\cap \pi(\nu)$. 
Otherwise, we store ${\mathcal M}[\mu][\nu]=0$. For graph $G$, matrix ${\mathcal M}$ is a $(n-|S|)\times (n-|S|)$ matrix, where both $Z_1$ and $Z_2$ are the set of Stretched units in $G$. \\


\noindent
Lemma \ref{lem : respecting nearest cuts} ensures that, even if $\pi(\mu)$ and $\pi(\nu)$ intersects with multiple edges, storing only one ordered pair in ${\mathcal M}[\mu][\nu]$ is sufficient.  Therefore, the space occupied by matrix ${\mathcal M}$ is ${\mathcal O}((n-|S|)^2)$. It leads to the following lemma.
\begin{lemma} \label{lem : matrix M}
    There is an ${\mathcal O}((n-|S|)^2)$ space matrix ${\mathcal M}$ such that, for any pair of stretched units $\mu,\nu$ of $G$, if there is at least one edge in $\pi(\mu) \cap \pi(\nu)=<N,M>$ and $\nu$ belongs to the nearest $(N,M)$-mincut of $\mu$, then ${\mathcal M}[\mu][\nu]$ stores the ordered pair of nodes $(N,M)$ of Skeleton; otherwise, ${\mathcal M}[\mu][\nu]=0$. 
\end{lemma}
We now use matrix ${\mathcal M}$ of Lemma \ref{lem : matrix M} to report an $S$-mincut in which both failed edges $e=(x,y)$ and $e'=(x',y')$ are contributing. We have cycle edges $e_1=(A,A')$ and $e_2=(B,B')$ from the same cycle of ${\mathcal H}_S$ such that $e_1\in \pi(e)$ and $e_2\in \pi(e')$. Let $\mu$ be a Stretched unit projected to a proper path $<P,Q>$ such that cycle-edge $e_1=(A,A')$ also belongs to the path $<P,Q>$. \\

\noindent
\textbf{Determine if $\mu$ belong to the nearest $(A,A')$-mincut of $\phi(x)$:} We can determine if $\mu$ belongs to the nearest $(A,A')$-mincut of $\phi(x)$ in ${\mathcal O}(1)$ time as follows. We look at the information stored at ${\mathcal M}[\phi(x)][\mu]$. If ${\mathcal M}[\phi(x)][\mu]=0$, then $\mu$ does not belong to the nearest $(A,A')$-mincut of $\phi(x)$. Otherwise, let ${\mathcal M}[\phi(x)][\mu]=(N,M)$. If path $<N,A>$ does not contain node $A'$, then $\mu$ belongs to the nearest $(A,A')$-mincut of $\phi(x)$, otherwise, it does not. It follows from Lemma \ref{lem : lca queries on skeleton} that we can verify in ${\mathcal O}(1)$ time whether path $<N,A>$ does not contain node $A'$. 
We now use this procedure to report an $S$-mincut in which both edges $e$ and $e'$ are contributing.\\

\noindent
\textbf{Reporting an $S$-mincut in which $e,e'$ are contributing:}
For edge $e_1=(A,A')$, we can determine in ${\mathcal O}(1)$ time whether $\phi(y)$ belongs to the nearest $(A,A')$-mincut of $\phi(x)$. If it is true, then assign $p\gets y$; otherwise, assign $p\gets x$. Similarly, for edge $e_2=(B,B')$, we assign either $q\gets x'$ or $q\gets y'$. We now construct a set $P$ that contains every vertex $u$ such that $e_1\in \pi(\phi(u))$ and $\phi(u)$ belongs to the nearest $(A,A')$-mincut of $\phi(p)$. Similarly, we construct a set $Q$ that contains every vertex $u$ such that $e_2\in \pi(\phi(u))$ and $\phi(u)$ belongs to the nearest $(B,B')$-mincut of $\phi(q)$. It is easy to observe that, using Lemma \ref{lem : projection mapping} and Lemma \ref{lem : lca queries on skeleton}, sets $P$ and $Q$ are constructed in ${\mathcal O}(n)$ time. 
It follows from Fact \ref{fact : closed under intersection and union} that $P\cup Q\cup C(S_{\mathcal B})$ is an $S$-mincut in which both edges $e,e'$ are contributing. This ensures that, using Skeleton, Projection mapping, and matrix ${\mathcal M}$, we can report an $S$-mincut $C$ in ${\mathcal O}(n)$ time such that both $e$ and $e'$ are contributing to $C$.\\

\noindent
Let us now consider the case when $\pi(e)$ intersects $\pi(e')$ at a cycle-edge or a tree-edge, say edge $f=(M,N)$ of Skeleton. Without loss of generality, let $C_1$ be the nearest $(M,N)$-mincut of $x$ and $C_2$  be the nearest $(M,N)$-mincut of $x'$.
We also consider the other three possible cases ($\{x,y\}$, $\{x,y'\}$, and $\{x',y'\}$) if $C_1$ or $C_2$ do not satisfy the conditions of the following lemma. This lemma can be established immediately using Fact \ref{fact : closed under intersection and union}.
\begin{lemma} \label{lem : single $S$-mincut}
    Both edges $e=(x,y)$ and $e'=(x',y')$ are contributing to a single $S$-mincut if and only if $y'\notin C_1$ and $y\notin C_2$.
\end{lemma}
Observe that, in a similar way as explained for the case when $e$ and $e'$ projected to different cycle-edges of the same cycle, using matrix ${\mathcal M}$ in Lemma \ref{lem : matrix M}, the conditions of Lemma \ref{lem : single $S$-mincut} can be verified in ${\mathcal O}(1)$ time. Moreover, an $S$-mincut in which both $e,e'$ are contributing can be reported in ${\mathcal O}(n)$ time, if $e,e'$ contributes to an $S$-mincut. The Skeleton ${\mathcal H}_S$, quotient mapping $\phi$, projection mapping $\pi$, and Matrix ${\mathcal M}$ occupy overall ${\mathcal O}((n-|S|)^2+n)$ space. 
So, this completes the proof of the following lemma.
\begin{lemma}  \label{lem : dual failure edges are projected to proper paths}
    There exists an ${\mathcal O}((n-|S|)^2+n)$ space data structure that, after the failure of any pair of edges $e,e'$ projected to proper paths of the Skeleton ${\mathcal H}_S$, can report in ${\mathcal O}(1)$ time the capacity of $S$-mincut. Moreover, an $S$-mincut can be reported in ${\mathcal O}(n)$ time for the resulting graph. 
\end{lemma}



\subsection{Both Failed Edges are Belonging to a Steiner Node} \label{sec : Steiner node}
Suppose both failed edges belong to a Steiner node and the corresponding $(\lambda_S+1)$ class be ${\mathcal W}$. We first consider the case when ${\mathcal W}$ is a Singleton $(\lambda_S+1)$ class. Later, in this section, we extend the result to $(\lambda_S+1)$ classes that contain more than one Steiner vertices.


\subsubsection{Handling Dual Edge Failures in a Singleton $(\lambda_S+1)$ class}
Suppose there exists exactly one Steiner vertex $s$ in ${\mathcal W}$. For any $(\lambda_S+1)$ cut $C$, we assume that $s\in C$; otherwise we can consider $\overline{C}$ instead of $C$. To determine whether both $e$ and $e'$ contribute to a single $(\lambda_S+1)$ cut, by Lemma \ref{lem : G(W)}, it is sufficient to work with graph $G({\mathcal W})$. Hence, we consider graph $G({\mathcal W})$ instead of $G$.
Note that there are four possibilities depending on which vertex of $\{x,y\}$ and which vertex of $\{x',y'\}$ belong to $C$, where $C$ is a $(\lambda_S+1)$ cut in which both edge $e,e'$ are contributing. Without loss of generality, here we give a procedure for verifying the existence of a $(\lambda_S+1)$ cut $C$ such that $x,x'\in C$ and $y,y'\in \overline{C}$. The other cases can also be verified using the same procedure.

Data structure ${\mathcal Q}_S({\mathcal W})$ from Theorem \ref{thm : data structure for singleton class} can determine in ${\mathcal O}(1)$ time whether there exists a $(\lambda_S+1)$ cut in which one of the two failed edges is contributing. However, it fails to determine whether both edges contribute to a single $(\lambda_S+1)$ cut.

Let $C_1$ be a nearest $(\lambda_S+1)$ cut from vertex $x$ to vertex $y$ (refer to Definition \ref{def : nearest minimum+1 cut}). Similarly, let $C_2$ be a nearest $(\lambda_S+1)$ cut from vertex $x'$ to vertex $y'$.  Observe that the following is the necessary condition for verifying whether both $e,e'$ contribute to a single $(\lambda_S+1)$ cut.
\begin{lemma} [Necessary Condition] \label{lem : new necessary condition} 
    If there is a single $(\lambda_S+1)$ cut $C$ such that $x,x'\in C$ and $y,y'\in \overline{C}$, then $y'\notin C_1$ and $y\notin C_2$. 
\end{lemma}
If $e'$ contributes to $C_1$ or $e$ contributes to $C_2$ then we are done. Henceforth we consider that $x\in C_1\setminus C_2$ and $x'\in C_2\setminus C_1$.

Suppose the capacity of global mincut is at least $4$.
The data structure ${\mathcal Q}_S({\mathcal W})$ in Theorem \ref{thm : data structure for singleton class}$(1)$ can verify this necessary condition (Lemma \ref{lem : new necessary condition}) in ${\mathcal O}(1)$ time using \textsc{belong} query.
However, for any vertex $u$, when the capacity of global mincut is at most $3$, recall that the data structure ${\mathcal Q}_S({\mathcal W})$ in Theorem \ref{thm : data structure for singleton class} may not store every cut $C\in N_{S({\mathcal W})}(u)$. So, there might exist a cut $C\in N_{S({\mathcal W})}(x)$ such that $y,y'\notin C$ and the data structure ${\mathcal N}_{\mathcal W}^{\le 3}(x)$ in Lemma \ref{lem : data structure for a single vertex at most 3} does not store $C$. Hence, data structure ${\mathcal Q}_S({\mathcal W})$ in Theorem \ref{thm : data structure for singleton class}$(2)$ might fail to verify the necessary condition (Lemma \ref{lem : new necessary condition}). We now establish the following lemma that helps in answering the necessary condition (Lemma \ref{lem : new necessary condition}) even when global mincut is at most $3$ and also plays a key role in answering dual edge failure query. 

\begin{lemma} [\textsc{Property ${\mathcal P}_3$}] \label{lem : property p3}
    Let $e_1$ be an edge that belongs to ${\mathcal W}$ and contributes to a pair of crossing $(\lambda_S+1)$ cuts $C,C'$ of $G({\mathcal W})$. Neither $C\setminus C'$ nor $C'\setminus C$ contains a vertex from ${\mathcal W}$ if and only if each of the two sets $C\setminus C'$ and $C'\setminus C$ contains a Steiner vertex from $S({\mathcal W})$.
\end{lemma} 
\begin{proof}
    Let edge $e_1=(p,q)$ belongs to ${\mathcal W}$ and $e_1$ contributes to $C,C'$ of $G({\mathcal W})$ (refer to Figure \ref{fig : proof}$(i)$). Suppose each of the two sets $C\setminus C'$ and $C'\setminus C$ contains no vertex from ${\mathcal W}$. Since they are crossing, there must exist at least one vertex from $S({\mathcal W})$ in each of the two sets $C\setminus C'$ and $C'\setminus C$.

    We now prove the converse part. Suppose each of the two sets $C\setminus C'$ and $C'\setminus C$ contains a Steiner vertex from $S({\mathcal W})$. Let $t_1\in C'\setminus C$ and $t_2\in C\setminus C'$. So, both $C\setminus C'$ and $C'\setminus C$ are Steiner cuts. Assume to the contrary that there is a vertex $u\in {\mathcal W}$ such that $u\in C\setminus C'$. Let us now consider graph $G({\mathcal W})\setminus e_1$ (refer to Section \ref{sec : Preliminaries} for the definition). Since $e_1$ contributes to both $C$ and $C'$, the capacity of cut $C$ and $C'$ in $G({\mathcal W})\setminus e_1$ is $\lambda_S$. By sub-modularity of cuts (Lemma \ref{submodularity of cuts}(2)), in $G({\mathcal W})\setminus e_1$, $c(C\setminus C')+c(C'\setminus C)\le 2\lambda_S$. Observe that, in graph $G({\mathcal W})$, edge $e_1$ contributes neither to cut $C\setminus C'$ nor to cut $C'\setminus C$. Hence, the following equation holds for graph $G({\mathcal W})$ as well. 
    \begin{equation} \label{eq : sum of diagonal cut capacities}
        c(C\setminus C')+c(C'\setminus C)\le 2\lambda_S
    \end{equation}
    It is also given that $u\in (C\setminus C')\cap {\mathcal W}$. So, $C\setminus C'$ is a Steiner cut that subdivides ${\mathcal W}$ because $s\notin C\setminus C'$. Therefore, in $G({\mathcal W})$, the capacity of $C\setminus C'$ is at least $\lambda_S+1$. It follows from Equation \ref{eq : sum of diagonal cut capacities} that $c(C'\setminus C)\le \lambda_S-1$. Since $t_1\in C'\setminus C$ and $s\notin C'\setminus C$, therefore, we have a Steiner cut $C'\setminus C$ of capacity at most $\lambda_S-1$ in $G({\mathcal W})$, a contradiction.       
\end{proof}
Exploiting Lemma \ref{lem : property p3}, we establish the following result (Lemma \ref{lem : necessary condition}) for nearest $(\lambda_S+1)$ cuts of an edge belonging to ${\mathcal W}$. This ensures that ${\mathcal Q}_S({\mathcal W})$ is indeed sufficient for verifying the necessary condition (Lemma \ref{lem : new necessary condition}) even when the capacity of global mincut is at most $3$.
\begin{figure}
 \centering
    \includegraphics[width=\textwidth]{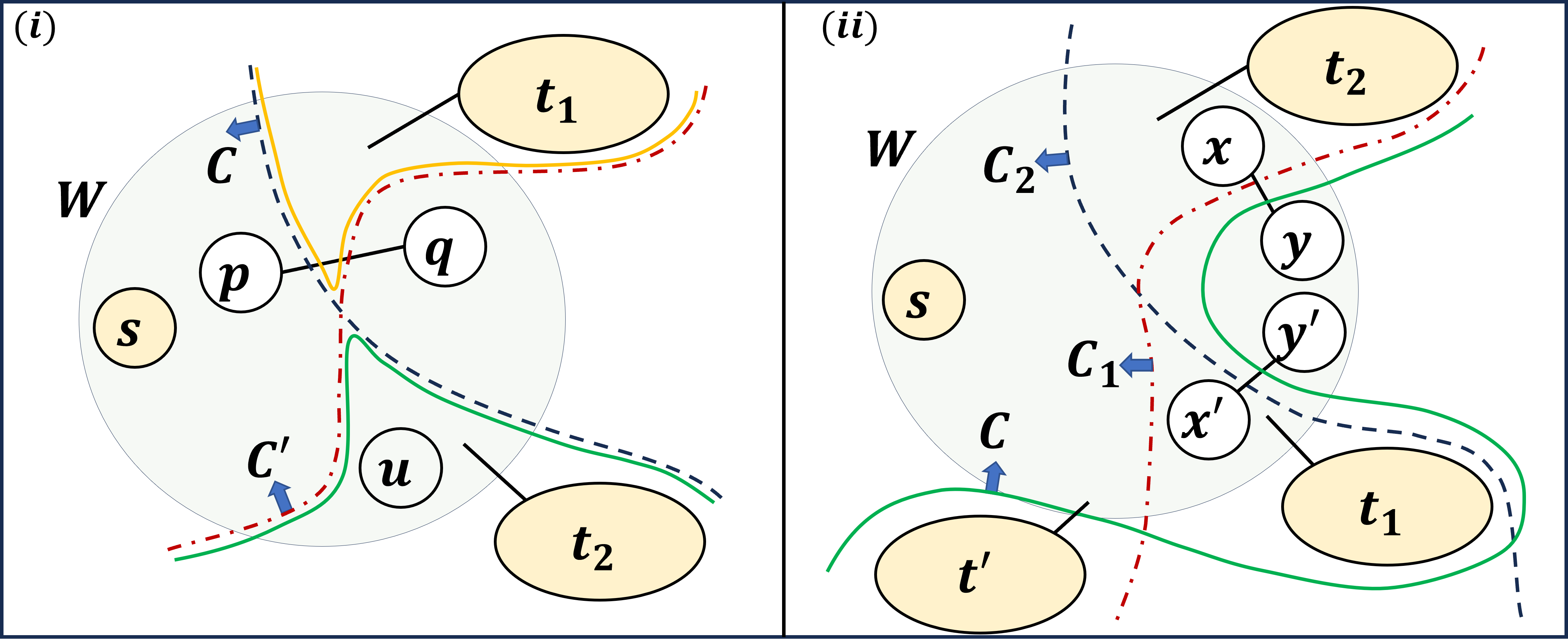} 
   \caption{(i) Illustration of the proof of Lemma \ref{lem : property p3} and Lemma \ref{lem : necessary condition}. Green and Yellow curves represent the Steiner cuts $C\setminus C'$ and $C'\setminus C$, respectively. $(ii)$ Illustration of the proof of Lemma \ref{lem : sufficient condition}.}
  \label{fig : proof}. 
\end{figure}
\begin{lemma} \label{lem : necessary condition}
    Let $(p,q)$ be any edge and $u$ be a vertex such that $p,q,u\in {\mathcal W}$ and $q\ne u\ne s$. Let $C$ be any cut from $N_{S({\mathcal W})}(p)$ with $q\notin C$. Then, $u\in C$ if and only if there is no $(\lambda_S+1)$ cut $C'$ such that $s,p\in C'$ and $q,u\in \overline{C'}$.  
\end{lemma}
\begin{proof}
    Suppose $u\in C$ for a cut $C\in N_{S({\mathcal W})}(p)$ with $q\notin C$. Assume to the contrary that there is a $(\lambda_S+1)$ cut $C'$ such that $s,p\in C'$ and $q,u\in \overline{C'}$ (refer to Figure \ref{fig : proof}$(i)$). Since $C$ and $C'$ are Steiner cuts, there must exist Steiner vertices $t_1,t_2$ such that $t_1\in \overline{C}$ and $t_2\in \overline{C'}$. It is given that $q\in \overline{C}\cap \overline{C'}\cap {\mathcal W}$ and $C,C'\in N_{S({\mathcal W})}(p)$. So, by Lemma \ref{lem : pair of bunches}, for each $i\in \{1,2\}$, $t_i$ cannot belong to $\overline{C}\cap \overline{C'}$. This implies that $t_1\in C'\setminus C$ and $t_2\in C\setminus C'$. So, we have a pair of crossing $(\lambda_S+1)$ cuts $C,C'$ such that edge $(p,q)$ contributes to both of them. Since $u\in (C\setminus C')\cap {\mathcal W}$, by Lemma \ref{lem : property p3}, $C\setminus C'$ or $C'\setminus C$ cannot contain any Steiner vertex, a contradiction.

    Suppose there is no $(\lambda_S+1)$ cut $C'$ such that $s,p\in C'$ and $q,u\in \overline{C'}$. Assume that there is a cut $C\in N_{S({\mathcal W})}(p)$ with $q\notin C$ such that $u\notin C$. Then $C$ is a $(\lambda_S+1)$ cut such that $s,p\in C$ and $q,u\notin C$, a contradiction. 
\end{proof}
Lemma \ref{lem : necessary condition} immediately leads to the following corollary.
\begin{corollary} \label{cor : necessary condition}
    Let $\lambda$ be the capacity of global mincut. For any $\lambda>0$, $y\in C_2$ (likewise $y'\in C_1$) if and only if there is no cut $C'\in N_{S({\mathcal W})}(x')$ with $y'\in \overline{C'}$ (likewise $C'\in N_{S({\mathcal W})}(x)$ with $y\in \overline{C'}$) that has $y\in \overline{C'}$ (likewise $y'\in \overline{C'}$). 
\end{corollary}
It follows from Corollary \ref{cor : necessary condition} that, for any capacity of global mincut, the necessary condition (Lemma \ref{lem : new necessary condition}) is verified in ${\mathcal O}(1)$ time by executing \textsc{belong} queries using data structure ${\mathcal Q}_S({\mathcal W})$ from Theorem \ref{thm : data structure for singleton class}.  

Suppose the necessary condition (Lemma \ref{lem : new necessary condition}) holds for edges $e$ and $e'$. Moreover, we have $s\in C_1\cap C_2$. This implies that $c(C_1\cap C_2)$ is at least $\lambda_S+1$. Unfortunately, the union, that is $C_1\cup C_2$, may not be a $(\lambda_S+1)$ cut. This is because $(\lambda_S+1)$ cuts are not closed under union.
To arrive at a sufficient condition, we first establish the following observation.
\begin{lemma} \label{lem : existence of a steiner vertex}
    If there exists a Steiner vertex $t$ such that $t\in \overline{C_1\cup C_2}$, then there is a single $(\lambda_S+1)$ cut in which both edges $e,e'$ are contributing. 
\end{lemma}
\begin{proof}
    Suppose Steiner vertex $t$ belongs to $\overline{C_1\cup C_2}$. Since $s\in (C\cap C)\cap {\mathcal W}$, $C_1\cup C_2$ is a Steiner cut. Observe that $y,y'\notin C\cup C'$ and $y,y'$ belongs to ${\mathcal W}$. So, $C\cup C'$ subdivides ${\mathcal W}$. Hence, $C_1\cup C_2$ has capacity at least $\lambda_S+1$. For cuts $C_1$ and $C_2$, by sub-modularity of cuts (Lemma \ref{submodularity of cuts}(1)), $c(C_1\cup C_2)+c(C_1\cap C_2)\le \lambda_S+2$. It follows that $C_1\cup C_2$ must be a $(\lambda_S+1)$ cut. Therefore, we have an $(\lambda_S+1)$ cut $C_1\cup C_2$ in which both edges $e,e'$ are contributing. 
\end{proof}
It follows from Lemma \ref{lem : existence of a steiner vertex} that the existence of a Steiner vertex $t$ with $t\in \overline{C_1 \cup C_2}$ guarantees that $C_1\cup C_2$ is a Steiner cut. Unfortunately, it is not always the case that $C_1\cup C_2$ contains a Steiner vertex. In addition, it seems quite possible that even if there does not exist any Steiner vertex in set $\overline{C_1\cup C_2}$, there is a $(\lambda_S+1)$ cut to which both edges $e,e'$ are contributing (e.g., refer to $(\lambda_S+1)$ cut $C$ in Figure \ref{fig : proof}$(ii)$). Interestingly, exploiting Lemma \ref{lem : property p3}, in the following lemma, we establish that the existence of a Steiner vertex in set $\overline{C_1\cup C_2}$ is also a sufficient condition.  

\begin{lemma} \label{lem : sufficient condition}
    There exists no Steiner vertex $t$ such that $t\in \overline{C_1\cup C_2}$ if and only if there is no $(\lambda_S+1)$ cut $C$ such that both edges $e,e'$ contribute to $C$.
\end{lemma}
\begin{proof}
Suppose there exists no Steiner vertex $t$ such that $t\in \overline{C_1 \cup C_2}$. Assume to the contrary that there is a $(\lambda_S+1)$ cut $C$ such that both edges $e,e'$ contribute to $C$ (refer to Figure \ref{fig : proof}$(ii)$). Since $C$ is a Steiner cut, there is a Steiner vertex $t'$ such that $t'\in \overline{C}$. Since $t'$ cannot belong to $\overline{C_1}\cap \overline{C_2}$, it follows that $t'$ belongs to at least one of $C_1$ and $C_2$. Without loss of generality, assume that $t'\in C_1$. So $t'\in C_1\setminus C$, and hence, $C_1\setminus C$ defines a Steiner cut. Since $C_1$ and $C_2$ are Steiner cuts, there exist Steiner vertices $t_1$ and $t_2$ such that $t_1\notin C_1$ and $t_2\notin C_2$. We show that $t_1\in C\setminus C_1$. Assume to the contrary that $t_1$ belongs to $\overline{C}$. 
For cuts $C,C_1$, there are vertices $s,x\in C\cap C_1$ and vertices $y,t_1\in \overline{C\cup C_1}$ such that $s,x,y\in {\mathcal W}$. Hence, both $C\cap C_1$ and $C\cup C_1$ subdivides ${\mathcal W}$ and they are Steiner cuts. By Definition \ref{def : lambda+1 S-class}, cuts $C\cap C_1$ and $C\cup C_1$ is at least $\lambda_S+1$. By sub-modularity of cuts (Lemma \ref{submodularity of cuts}(1)), for cuts $C$ and $C_1$, we have $c(C\cap C_1)+c(C\cup C_1)\le 2\lambda_S+2$. It follows that $C\cap C_1$ and $C\cup C_1$ both are $(\lambda_S+1)$ cuts.  
Moreover, $C\cap C_1$ is a proper subset of $C_1$ as $C\cap C_1$ does not contain $t'$. In addition, since $x\in C\cap C_1$, therefore, $C_1$ cannot belong to $N_{S({\mathcal W})}(x)$, a contradiction due to Definition \ref{def : nearest minimum+1 cut}. Therefore, we have $t_1\in C\setminus C_1$. Observe that edge $(x,y)$ contributes to both $C_1$ and $C$. Since both edges $e,e'$ contribute to $C$, so we have $x'\in C$. Moreover, it is known that $x'\notin C_1$. Therefore, $x'\in C\setminus C_1$. Finally we have an edge $(x,y)$ that contributes to a pair of crossing $(\lambda_S+1)$ cuts $C_1,C$ since $s\in C\cap C_1$, $t_1\in C\setminus C_1$, $t'\in C_1\setminus C$, and $y,y'\in \overline{C\cup C_1}$. We have also shown that both $C_1\setminus C$ and $C\setminus C_1$ contain a Steiner vertex from $S({\mathcal W})$. However, there is a vertex $x'\in {\mathcal W}$ that belongs to $C\setminus C_1$, a contradiction due to Lemma \ref{lem : property p3}.

The converse part holds because of Lemma \ref{lem : existence of a steiner vertex}.
\end{proof}
Based on Lemma \ref{lem : sufficient condition}, we now state the following lemma that helps in efficiently verifying whether the failure of $e,e'$ from a Singleton $(\lambda_S+1)$ class reduces the capacity of $S$-mincut. 
\begin{lemma} \label{lem : condition of dual failure}
    Both failed edges $e=(x,y)$ and $e'=(x',y')$ contribute to a single $(\lambda_S+1)$ cut $C$ in $G({\mathcal W})$ such that $x,x'\in C$ and $y,y'\in \overline{C}$ if and only if there exists a $(\lambda_S+1)$ cut $C_1\in N_{S({\mathcal W})}(x)$, a $(\lambda_S+1)$ cut $C_2\in N_{S({\mathcal W})}(x')$, and a Steiner vertex $t\in S({\mathcal W})$ such that $y,y'\notin C_1 \cup C_2$ and $t\in \overline{C_1 \cup C_2}$.
\end{lemma}
\begin{proof}
    Suppose $C$ is a $(\lambda_S+1)$ cut in which both failed edges $e$ and $e'$ contribute. So, we have $x,x'\in C$ and $y,y'\in \overline{C}$. By Lemma \ref{lem : new necessary condition}, this implies that there exist a nearest $(\lambda_S+1)$ cut $C_i$ from $x$ to $y$ and a nearest $(\lambda_S+1)$ cut $C_j$ from $x'$ to $y'$ such that $C_i\subseteq C$, $C_j\subseteq C$, and for $C_i$ and $C_j$, we have $y'\notin C_i$ and $y\notin C_j$. So, $y,y'\notin C_i\cup C_j$. Moreover, by Lemma \ref{lem : sufficient condition}, there must exist a Steiner vertex $t$ such that $t\in \overline{C_i \cup C_j}$. 

    The proof of the converse part holds immediately from 
    Lemma \ref{lem : sufficient condition}.
\end{proof}

\noindent
\textbf{Answering Query for Dual Edge Failure:}
Upon failure of edge $e=(x,y)$ and $e'=(x',y')$ from  Singleton $(\lambda_S+1)$ class ${\mathcal W}$, the data structure ${\mathcal Q}_S({\mathcal W})$ from Theorem \ref{thm : data structure for singleton class} verify in ${\mathcal O}(1)$ time whether the first condition of Lemma \ref{lem : condition of dual failure}, that is, whether $C_1$ and $C_2$ exists such that $y,y'\notin C_1\cup C_2$. For the second condition, observe that, for cut $C_1$ and $C_2$, we have to find out whether there is a Steiner vertex $t$ such that $t\in \overline{C_1 \cup C_2}$. Trivially, it takes ${\mathcal O}(|S|)$ time to verify this condition. We now design a data structure that can determine this in ${\mathcal O}(1)$ time. 

\subsubsection*{An $O((n-|S|)^{2}+n)$ Space Data Structure with O(1) Query Time}
To determine whether there exists a Steiner vertex $t$ in $\overline{C_1\cup C_2}$, it requires ${\mathcal O}(|S|)$ time because \textsc{Property ${\mathcal P}_2$} (Lemma \ref{lem : a bunch belong to at most two cuts}) fails to hold for a pair of vertices in $G({\mathcal W})$. In other words, for a pair of vertices $u_1$ and $u_2$ in ${\mathcal W}$, there can be $\Omega(|S({\mathcal W})|)$ number of Steiner vertices that can belong to $\overline{C\cup C'}$, where $C\in N_{S({\mathcal W})}(u_1)$ and $C'\in N_{S({\mathcal W})}(u_2)$ (refer to Figure \ref{fig : 2nd figure in overview}).

  \begin{figure}
 \centering
    \includegraphics[width=200pt]{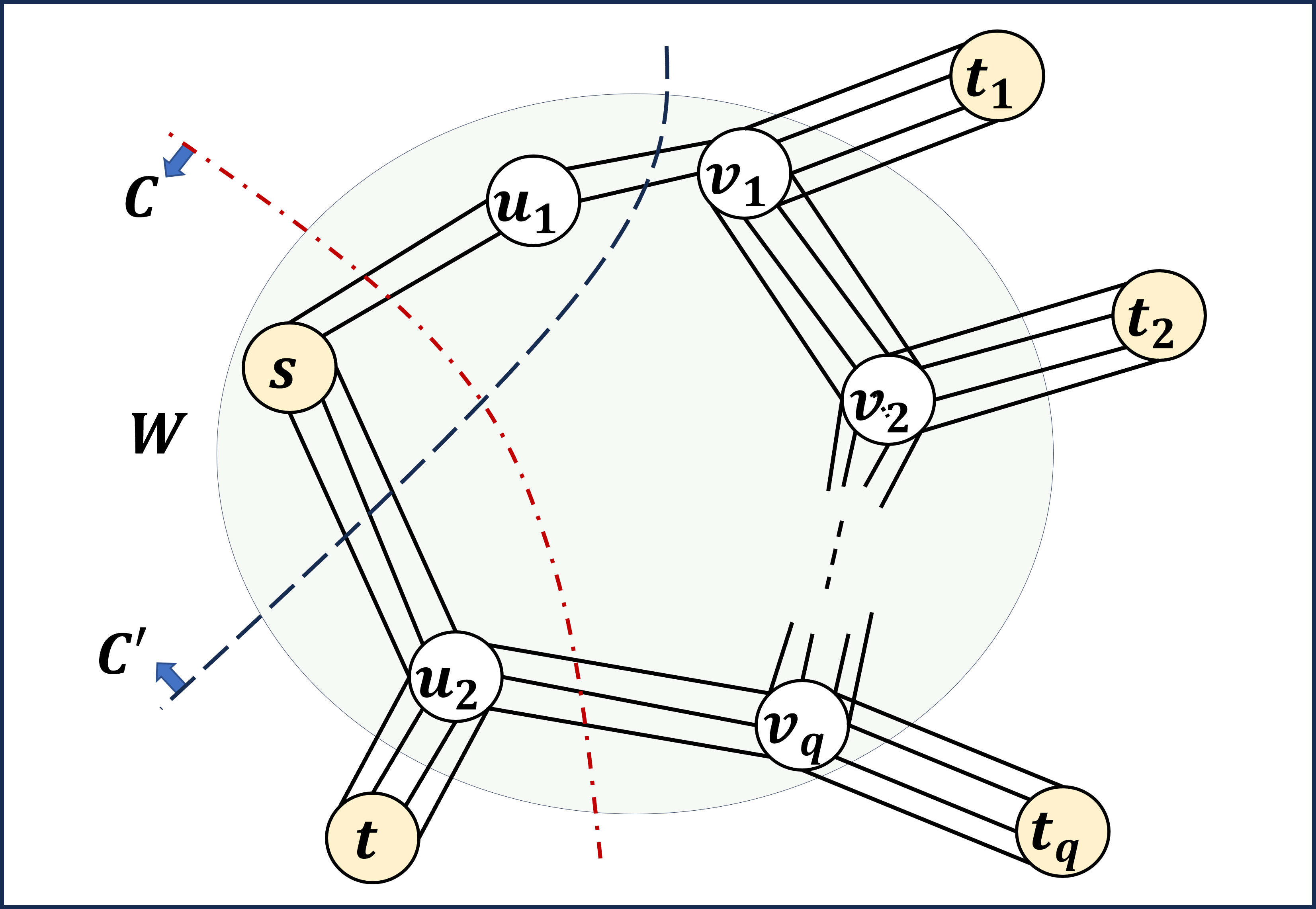} 
   \caption{Yellow vertices are Steiner vertices of $G({\mathcal W})$ and white vertices belong to ${\mathcal W}$. Let $q={\Omega}(|S({\mathcal W})|)$. In this graph, the capacity of S-mincut is $4$, $C'\in N_{S({\mathcal W})}(u_1)$ and $C\in N_{S({\mathcal W})}(u_2)$. There are $q$ Steiner vertices that belong to $\overline{C_1\cup C_2}$. }
  \label{fig : 2nd figure in overview}. 
\end{figure}

Let us consider the set ${\mathcal C}_{\lambda_S+1}$ consisting of all $(\lambda_S+1)$ cuts that subdivides the Singleton $(\lambda_S+1)$ class ${\mathcal W}$ and set ${\mathcal C}_s$ contains every cut from $N_{S({\mathcal W})}(u)$ for each $u\in {\mathcal W}$. Using ${\mathcal C}_{\lambda_S+1}$, we present a procedure (pseudocode is given in Algorithm \ref{alg : constructing S*}) for constructing a set $S^*\subseteq S({\mathcal W})$ such that every cut $C$ in ${\mathcal C}_S$ is marked with exactly one vertex from $S^*$. \\

\noindent
\textbf{Construction of set $S^*$:} Set $S^*$ is initially empty. Let ${\mathcal D}\gets {\mathcal C}_{\lambda_S+1}$ and $P\gets S({\mathcal W})$. Consider any Steiner vertex $s^*\in P$. For every cut $C\in {\mathcal D}$, if $s^*\in \overline{C}$, then we do the four steps -- (1) remove $C$ from ${\mathcal D}$, (2) mark $C$ with $s^*$, (3) $S^*\gets S^*\cup \{s^*\}$, and (4) for every Steiner vertex $t$ in $\overline{C}$, remove $t$ from $P$. We repeat this procedure till $P$ or ${\mathcal D}$ becomes an empty set.\\

\begin{algorithm}[h]
\caption{Construction of a Steiner Set $S^*$ for ${\mathcal W}$ given $S({\mathcal W})$}
\label{alg : constructing S*}
\begin{algorithmic}[1]
\Procedure{\textsc{ConstructionOfSteinerSet}$(S({\mathcal W}))$}{}
    \State Let $P\gets S({\mathcal W})\setminus \{s\}$, $Q\gets \emptyset$, and ${\mathcal D}\gets {\mathcal C}_{\lambda_S+1}$;
    \While{$P\ne \emptyset$ and ${\mathcal D}\ne \emptyset$}
    \State select a vertex $s^*$ from $P$;
    \State $Q\gets Q\cup \{s^*\}$;
    \For{every cut $C\in {\mathcal D}$}
        \If {$s^*\in \overline{C}$}
            \State Mark $C$ with $s^*$;
            \State ${\mathcal D}\gets {\mathcal D}\setminus \{C\}$;
            \For{every Steiner vertex $t\in \overline{C}\cap P$}
                \State Assign $P\gets P\setminus \{t\}$;
            \EndFor
            \Else \State do nothing;
        \EndIf
    \EndFor
    \EndWhile
    \State \Return $Q$;
\EndProcedure
\end{algorithmic}
\end{algorithm}

The following lemma holds from the construction of $S^*$.
\begin{lemma} \label{lem : a cut is marked with exactly one steiner vertex}
    Every cut of ${\mathcal C}_{\lambda_S+1}$ is marked with exactly one vertex of $S^*$. 
\end{lemma}
Since ${\mathcal C}_s$ is a subset of ${\mathcal C}_{\lambda_S+1}$, therefore, by Lemma \ref{lem : a cut is marked with exactly one steiner vertex}, every cut of ${\mathcal C}_s$ is also marked with exactly one vertex from $S^*$. It follows from the proof of Lemma \ref{lem : data structure for a single vertex at least 4} and Lemma \ref{lem : data structure for a single vertex at most 3} that there are ${\mathcal O}(|\overline{S({\mathcal W})}|^2)$ cuts in ${\mathcal C}_s$, where $\overline{S({\mathcal W})}$ is the set of all nonSteiner vertices belonging to ${\mathcal W}$. Therefore, it occupies ${\mathcal O}(|\overline{S({\mathcal W})}|^2)$ space to mark all cuts from ${\mathcal C}_s$. We now establish the following lemma that helps in answering dual edge failure query in ${\mathcal O}(1)$ time.
\begin{lemma} \label{lem : verifying third condition in constant}
    Suppose the nearest $(\lambda_S+1)$ cut $C_1$ from $x$ to $y$ and the nearest $(\lambda_S+1)$ cut $C_2$ from $x'$ to $y'$ satisfies the necessary condition, that is, $y'\notin  C_1$ and $y\notin C_2$. Then, there is a Steiner vertex $t$ in $\overline{C_1\cup C_2}$ if and only if $C_1$ and $C_2$ are marked with the same Steiner vertex from $S^*$.
\end{lemma}
\begin{proof}
    Suppose there is a Steiner vertex $t$ in $\overline{C_1\cup C_2}$ but $C_1$ and $C_2$ are marked with different Steiner vertices $s_1,s_2$ from $S^*$. So, $s_1,s_2$ also belongs to $\overline{C_1\cap C_2}$. Observe that $s_1$ can belong to $C_2\setminus C_1$ and $s_2$ can belong to $C_1\setminus C_2$ or vice versa. We know that $s\in C_1\cup C_2$. Since $C_1$ and $C_2$ satisfies the necessary condition, so, $y,y'\in (\overline{C_1\cup C_2})\cap {\mathcal W}$. Therefore, $C_1\cup C_2$ subdivides ${\mathcal W}$ and has a capacity at least $\lambda_S+1$. We also have $s$ in both $C_1\cap C_2$ and $C_1\cup C_2$, similarly, $t$ is in both  $\overline{C_1\cap C_2}$ and $\overline{C_1\cup C_2}$. Moreover, for cuts $C_1,C_2$, by sub-modularity of cuts (Lemma \ref{submodularity of cuts}(1)), $c(C_1\cap C_2)+c(C_1\cup C_2)\le 2\lambda_2+2$. It follows that $C_1\cap C_2$ is a $(\lambda_S+1)$ cut. Then $C_1\cap C_2$ is marked with both Steiner vertices $s_1,s_2$ from $S^*$, a contradiction due to Lemma \ref{lem : a cut is marked with exactly one steiner vertex}.
    

    If $C_1$ and $C_2$ are marked with the same Steiner vertex $s^*$ from $S^*$ then $s^*\in \overline{C_1\cup C_2}$.
\end{proof}
It follows from Lemma \ref{lem : verifying third condition in constant} that the second condition of Lemma \ref{lem : condition of dual failure}, that is, whether there exists a Steiner vertex $t\in S({\mathcal W})$ such that $t\in \overline{C_1\cup C_2}$, can be verified in ${\mathcal O}(1)$ time by determining whether $C_1$ and $C_2$ are marked with the same Steiner vertex.

Recall that the data structure ${\mathcal Q}_S({\mathcal W})$ in Theorem \ref{thm : data structure for singleton class} occupies ${\mathcal O}(|\overline{S({\mathcal W})}|^2)$ space if, for every cut $C$ in $N_S(u),~\forall u\in {\mathcal W}$, it does not store Steiner vertices of $G({\mathcal W})$ belonging to $\overline{C}$. Moreover, as discussed above, marking of every cut in ${\mathcal C}_s$ with exactly one vertex of $S^*$ occupies ${\mathcal O}(|\overline{S({\mathcal W})}|^2)$ space. Therefore, we have an ${\mathcal O}(|\overline{S({\mathcal W})}|^2)$ space data structure that can verify the conditions of Lemma \ref{lem : condition of dual failure} in ${\mathcal O}(1)$ time. Therefore, similar to the analysis of the space occupied by data structure ${\mathcal D}_s$ in Lemma \ref{lem : data structure for singleton class}, we can establish that there is an ${\mathcal O}((n-|S|)^2+n)$ space data structure that can report the capacity of $S$-mincut in ${\mathcal O}(1)$ time after the failure of any pair of edges belonging to any Singleton $(\lambda_S+1)$ class.
Moreover, to report an $S$-mincut after failure of edges $e,e'$, we store data structure ${\mathcal D}_s$ from Lemma \ref{lem : data structure for singleton class} for all Singleton $(\lambda_S+1)$ classes of $G$. Data structure ${\mathcal D}_s$ can report $\overline{C_1}$ and $\overline{C_2}$ in ${\mathcal O}(n)$ time. Therefore, we can report $\overline{C_1}\cap \overline{C_2}$ in ${\mathcal O}(n)$ time as well. This completes the proof of the following lemma.
\begin{lemma} \label{lem : singleton class final lemma}
    There is an ${\mathcal O}((n-|S|)^2+n)$ space data structure that, after the failure of any pair of edges belonging to a Singleton $(\lambda_S+1)$ class, can report the capacity of $S$-mincut in ${\mathcal O}(1)$ time. Moreover, there is an ${\mathcal O}(n(n-|S|+1))$ space data structure that can report an $S$-mincut for the resulting graph in ${\mathcal O}(n)$ time.
\end{lemma}


\subsubsection{Handling Dual Edge Failures for Generic Steiner node}
Suppose ${\mathcal W}$ contains more than one Steiner vertices. Similar to the construction of data structure in Lemma \ref{lem : data structure for generic steiner node} for all Steiner nodes, we construct the pair of graphs $G({\mathcal W})_1$ and $G({\mathcal W})_2$ from $G({\mathcal W})$ that satisfies the property of Lemma \ref{lem : all cuts in gw1} and Lemma \ref{lem : all cuts GW2} respectively. 

By construction, the $(\lambda_S+1)$ class ${\mathcal W}_1$ of graph $G({\mathcal W})_1$ is a Singleton $(\lambda_S+1)$ class and all vertices of ${\mathcal W}$ are mapped to ${\mathcal W}_1$. Let ${\mathbb W}$ be the set of all $(\lambda_S+1)$ classes of $G$ that contain at least two Steiner vertices. We construct graph $G({\mathcal W}')_1$ for every ${\mathcal W}'\in {\mathbb W}$. 
We store the data structures of Lemma \ref{lem : singleton class final lemma} for handling dual edge failures for $G({\mathcal W}')_1$ for every $(\lambda_S+1)$ class ${\mathcal W}'$ of ${\mathbb W}$. Thus, using Lemma \ref{lem : singleton class final lemma}, it leads to the following result.
\begin{lemma} \label{lem : dual failure and does not subdivide steiner set}
     There exists an ${\mathcal O}((n-|S|)^2+n)$ space data structure that, after the failure of any pair of edges $e,e'$ from any $(\lambda_S+1)$ class ${\mathcal W}$ corresponding to a Steiner node, can determine in ${\mathcal O}(1)$ time whether both $e,e'$ contribute to a single $(\lambda_S+1)$ cut $C$ that does not subdivide the Steiner set belonging to ${\mathcal W}$. Moreover, there is an ${\mathcal O}(n(n-|S|+1))$ space data structure that can report $S$-mincut $C$ for the resulting graph in ${\mathcal O}(n)$ time. 
\end{lemma}
We now consider graph $G({\mathcal W})_2$. By Lemma \ref{lem : all cuts GW2}, every Steiner mincut of $G({\mathcal W})_2$ is a $(\lambda_S+1)$ cut of $G$ that subdivides the Steiner set of ${\mathcal W}$. Therefore, to determine whether both failed edges $e,e'$ contribute to a single $(\lambda_S+1)$ cut of $G({\mathcal W})$ that subdivides Steiner set of ${\mathcal W}$, it is sufficient to verify whether edges $e,e'$ in $G({\mathcal W})_2$ contribute to a single Steiner mincut of $G({\mathcal W})_2$. 
For every $(\lambda_S+1)$ classes ${\mathcal W}_i$ of ${\mathbb W}$, we store the data structure of Lemma \ref{lem : dual failure edges are projected to proper paths} with the following modification. We store an instance, denoted by ${\mathcal M}(G({\mathcal W}_i)_2)$ of matrix ${\mathcal M}$ in Lemma \ref{lem : matrix M}. For matrix ${\mathcal M}({G({\mathcal W}_i)_2})$ both rows and columns are only the set of nonSteiner vertices of $G({\mathcal W}_i)_2$ that belong to ${\mathcal W}_i$; which is because both failed edges $e,e'$ belong to ${\mathcal W}_i$. By Lemma \ref{lem : matrix M}, this data structure helps in determining in ${\mathcal O}(1)$ time whether both edges are contributing to a Steiner mincut of $G({\mathcal W})_2$.

Let ${\mathcal W}_i\in {\mathbb W}$. The number of Steiner vertices of $G({\mathcal W}_i)_2$, say $|S_i|$, is the number of Steiner vertices in $G$ that belongs to ${\mathcal W}_i$. Let $N_i$ be the node of the Skeleton ${\mathcal H}_S$ that represents ${\mathcal W}_i$. The number of nonSteiner vertices of $G({\mathcal W}_i)_2$ is the sum of the number of nonSteiner vertices, say $n_i$, of $G$ that belong to ${\mathcal W}_i$ and the number of bunches adjacent to node $N_i$, which is $\text{deg}(N_i)$. So, in a similar way of analyzing the space occupied by the data structure in Lemma \ref{lem : dual failure edges are projected to proper paths}, we now analyze the space occupied by this data structure for graph $G({\mathcal W}_i)_2$. The Skeleton for $G({\mathcal W}_i)_2$ occupies ${\mathcal O}(|S_i|)$. The projection mapping and Quotient mapping of every vertex occupies ${\mathcal O}(n_i+|S_i|+\deg (N_i))$ space. Finally, matrix ${\mathcal M}({G({\mathcal W}_i)_2})$ for graph $G({\mathcal W}_i)_2$ occupies ${\mathcal O}(n_i^2)$ space. Hence, the space occupied by this data structure is ${\mathcal O}(\Sigma_{{\mathcal W}\in {\mathbb W}}(n_i^2+|S_i|+n_i+\text{deg}(N_i)))$ for ${\mathcal W}_i$. Therefore, for all $(\lambda_S+1)$ classes of ${\mathbb W}$, the space occupied by the data structure is ${\mathcal O}((n-|S|)^2+n)$ since $\Sigma_{{\mathcal W}\in {\mathbb W}}n_i=(n-|S|)$, $\Sigma_{{\mathcal W}\in {\mathbb W}}n_i^2={\mathcal O}((n-|S|)^2)$, and $\Sigma_{{\mathcal W}\in {\mathbb W}}|S_i|={\mathcal O}(|S|)$. This completes the proof of the following lemma.
\begin{lemma} \label{lem : dual failure and subdivide steiner set}
     There exists an ${\mathcal O}((n-|S|)^2+n)$ space data structure that, after the failure of any pair of edges $e,e'$ from any $(\lambda_S+1)$ class ${\mathcal W}$ corresponding to a Steiner node, can determine in ${\mathcal O}(1)$ time whether both $e,e'$ contributes to a single $(\lambda_S+1)$ cut $C$ that subdivides the Steiner set belonging to ${\mathcal W}$.
\end{lemma}
For a ${\mathcal W}\in {\mathbb W}$, suppose there is a Steiner mincut $C$ of $G({\mathcal W})_2$ in which both edges $e,e'$ are contributing. Let $S'$ be the set of Steiner vertices of $G({\mathcal W})$ that do not belong to ${\mathcal W}$. Observe that the data structure in Lemma \ref{lem : dual failure and subdivide steiner set} can report $C\cap {\mathcal W}$ in ${\mathcal O}(n)$ time. However, it cannot report vertices of $S'$ that belong to $C$. 
This is because it does not store matrix ${\mathcal M}$ for vertices in $S'$. Therefore, to report a Steiner mincut of $G({\mathcal W})_2$ in which both edges are contributing, we also store another instance, denoted as ${\mathcal M}'$, of matrix ${\mathcal M}$ from Lemma \ref{lem : matrix M}. Matrix ${\mathcal M}'$ is constructed for set $S'$ and the set of nonSteiner vertices belonging to ${\mathcal W}$ as follows. The rows of matrix ${\mathcal M}'$ are the nonSteiner vertices belonging to ${\mathcal W}$ and columns are the vertices of $S'$. It follows from Lemma \ref{lem : matrix M} that matrix ${\mathcal M}'$ helps to determine in ${\mathcal O}(1)$ time whether a vertex $s'\in S'$ belongs to the nearest $(A,A')$-mincut of a nonSteiner vertex $u$ of ${\mathcal W}$, where $(A,A')$ is an edge of the Skeleton to which both $s'$ and $u$ are projected. Therefore, matrix ${\mathcal M}'$ can report $C\cap S'$ in ${\mathcal O}(n)$ time in the worst case. Similar to the data structure in Lemma \ref{lem : data structure for G2}, we can store the corresponding tree-edge or the pair of cycle-edges of ${\mathcal H}_S$ for each vertex in $S'$. So, it can be used in reporting all the vertices of $G$ that are mapped to $C\cap S'$ in ${\mathcal O}(n)$ time. Finally, the obtained data structure can report the cut $C$ in ${\mathcal O}(n)$ time. For each ${\mathcal W}_i\in {\mathbb W}$, observe that matrix ${\mathcal M}_i'$ occupies ${\mathcal O}(\text{deg}(N_i)n_i)$ space. Therefore, ${\mathcal O}(|S|(n-|S|))$ space is occupied by all $(\lambda_S+1)$ classes of ${\mathbb W}$. So, using Lemma \ref{lem : dual failure and subdivide steiner set}, to report $C$ in ${\mathcal O}(n)$ time, 
the obtained data structure occupies overall ${\mathcal O}((n-|S|)^2+|S|(n-|S|)+n)$ space, which is ${\mathcal O}(n(n-|S|+1))$.  Thus,  it leads to the following lemma.
\begin{lemma} \label{lem : dual failure and subdivide steiner set reporting cut}
     There exists an ${\mathcal O}((n-|S|)^2+n))$ space data structure that, after the failure of any pair of edges $e,e'$ from any $(\lambda_S+1)$ class ${\mathcal W}$ corresponding to a Steiner node, can determine in ${\mathcal O}(1)$ time whether both $e,e'$ contributes to a single $(\lambda_S+1)$ cut $C$ that subdivides the Steiner set belonging to ${\mathcal W}$. Moreover, there is an ${\mathcal O}(n(n-|S|+1))$ space data structure that can report cut $C$ in ${\mathcal O}(n)$ time.
\end{lemma}
Lemma \ref{lem : singleton class final lemma}, Lemma \ref{lem : dual failure and does not subdivide steiner set},  and Lemma \ref{lem : dual failure and subdivide steiner set reporting cut} establish the following lemma for handling dual edge failures for all Steiner nodes of $G$.
\begin{lemma} \label{lem : generic Steiner node dual edge failure}
    There exists an ${\mathcal O}((n-|S|)^2+n)$ space data structure that, after the failure of any pair of edges belonging to a $(\lambda_S+1)$ class corresponding to a Steiner node, can report the capacity of $S$-mincut in ${\mathcal O}(1)$ time. Moreover, there is an ${\mathcal O}(n(n-|S|+1))$ space data structure that can report an $S$-mincut for the resulting graph in ${\mathcal O}(n)$ time.
\end{lemma}





\subsection{Both Failed Edges are Belonging to a NonSteiner Node} \label{sec : nonSteiner node}
 Suppose $\mu$ be a nonSteiner node of the Flesh graph of $G$ and the corresponding $(\lambda_S+1)$ class be ${\mathcal W}_{\mu}$. So, $\mu$ can either be a terminal unit or a Stretched unit. \\

\noindent
\textbf{Both Failed Edges are in a Stretched Unit:} Suppose $\mu$ is a Stretched unit. Similar to the case for designing a data structure for Stretched units in Appendix \ref{sec : stretched unit data structure}, we construct the pair of graphs $G({\mathcal W}_{\mu})^I$ and $G({\mathcal W}_{\mu})^U$ using covering technique of Baswana, Bhanja, and Pandey \cite{DBLP:journals/talg/BaswanaBP23}. By construction, there is a Singleton $(\lambda_S+1)$ class ${\mathcal W}_{\mu}\cup \{s\}$ in $G({\mathcal W}_{\mu})^I$ and there is a Singleton $(\lambda_S+1)$ class ${\mathcal W}_{\mu}\cup \{t\}$ in $G({\mathcal W}_{\mu})^U$. Therefore, we store the data structure of Lemma \ref{lem : singleton class final lemma} for $G({\mathcal W}_{\mu})^I$, as well as for $G({\mathcal W}_{\mu})^U$, for each $(\lambda_S+1)$ class ${\mathcal W}_{\mu}$ of $G$ for which the corresponding node is a Stretched unit. This leads to the following lemma.
\begin{lemma} \label{lem : stretched unit}
    There is an ${\mathcal O}((n-|S|)^2)$ space data structure that, after the failure of any pair of edges belonging to a $(\lambda_S+1)$ class corresponding to a Stretched unit, can report the capacity of $S$-mincut in ${\mathcal O}(1)$ time. Moreover, there is an ${\mathcal O}(n(n-|S|+1))$ space data structure that can report an $S$-mincut for the resulting graph in ${\mathcal O}(n)$ time.
\end{lemma}
 
\noindent
\textbf{Both Failed Edges are in a Terminal Unit:} Suppose $\mu$ be a terminal unit. We construct the graph $G({\mathcal W}_{\mu})^I$ and $G({\mathcal W}_{\mu})^U$
as discussed in Appendix \ref{sec : terminal unit data structure}. The remaining construction is the same as of the Stretched unit in Lemma \ref{lem : stretched unit}, which leads to the following lemma. 
\begin{lemma} \label{lem : terminal units with no steiner vertex}
     Let ${\mathbb W}_T$ be the set of all $(\lambda_S+1)$ classes of $G$ such that for any $(\lambda_S+1)$ class ${\mathcal W}\in {\mathbb W}_T$, ${\mathcal W}\cap S=\emptyset$ and ${\mathcal W}$ corresponds to a terminal unit of $G$. There is an ${\mathcal O}((n-|S|)^2)$ space data structure that, after the failure of any pair of edges belonging to a $(\lambda_S+1)$ class in ${\mathbb W}$, can report the capacity of $S$-mincut in ${\mathcal O}(1)$ time. Moreover, there is an ${\mathcal O}(n(n-|S|+1))$ space data structure that can report an $S$-mincut for the resulting graph in ${\mathcal O}(n)$ time.
\end{lemma}
Finally, Lemma \ref{lem : dual failure edges are projected to proper paths}, Lemma \ref{lem : generic Steiner node dual edge failure}, Lemma \ref{lem : stretched unit}, and Lemma \ref{lem : terminal units with no steiner vertex} completes the proof of Theorem \ref{thm : dual edge failure} for dual edge failures.

\section{Sensitivity Oracle: Dual Edge Insertions} \label{app : dual edge insertion}
In this section, we design a compact data structure ${\mathcal I}_S$ that, upon insertion of any pair of edges $e=(x,y)$ and $e'=(x',y')$ in $G$ with $x,y,x',y'\in V$, can efficiently report an $S$-mincut for the resulting graph. To arrive at data structure ${\mathcal I}_S$, we crucially exploit the Connectivity Carcass of Dinitz and Vainshtain \cite{DBLP:conf/stoc/DinitzV94, DBLP:conf/soda/DinitzV95, DBLP:journals/siamcomp/DinitzV00}, the Covering technique of Baswana, Bhanja, and Pandey \cite{DBLP:journals/talg/BaswanaBP23}, and the data structure for Singleton $(\lambda_S+1)$ classes from Theorem \ref{thm : data structure for singleton class}. 

For $S=\{s,t\}$, Picard and Queyranne \cite{DBLP:journals/mp/PicardQ80} showed that, upon insertion of any edge $(u,v)$, the capacity of $(s,t)$-mincut increases by $1$ if and only if $u$ belongs to the $(\lambda_S+1)$ class containing $s$ and $v$ belongs to the $(\lambda_S+1)$ class containing $t$ or vice versa. This property can be extended to $S$-mincut as follows.  
\begin{lemma} \label{lem : $S$-mincuts of a bunch increases}
    Upon insertion of any edge $(u,v)$, for any bunch ${\mathcal B}$, the capacity of every $S$-mincut in  ${\mathcal B}$ increases by $1$ if and only if $x\in C(S_{\mathcal B})$ and $y\in C(S\setminus S_{\mathcal B})$ or vice versa.
\end{lemma}
Baswana and Pandey \cite{DBLP:conf/soda/BaswanaP22}, exploiting Lemma \ref{lem : $S$-mincuts of a bunch increases}, designed an ${\mathcal O}(n)$ space data structure that, after insertion of any edge in $G$, can report  the capacity of $S$-mincut and an $S$-mincut in ${\mathcal O}(1)$ time and ${\mathcal O}(n)$ time, respectively (the proof is provided in Appendix \ref{app : single edge insertion} for the sake of completeness). Therefore, for each edge $e_1$ in $\{e,e'\}$, this data structure is sufficient to determine whether the capacity of $S$-mincut increases by $1$ by separately adding edge $e$ and $e'$. We now consider the insertion of both edges $e,e'$ to determine whether $S$-mincut increases.


Skeleton ${\mathcal H}_S$ is a cactus graph.  By construction of ${\mathcal H}_S$ \cite{DBLP:conf/stoc/DinitzV94, DBLP:conf/soda/DinitzV95, DBLP:journals/siamcomp/DinitzV00}, every cycle in ${\mathcal H}_S$ contains at least four nodes. This is because a cycle is created in the Skeleton only if there is a pair of crossing $S$-mincuts such that every corner set contains at least one Steiner vertex. Moreover, every node that appears in a cycle is an empty node and has degree exactly $3$ -- one adjacent edge is a tree-edge and the other two adjacent edges are the cycle-edges from the cycle (refer to Section 4.1 in \cite{DBLP:conf/soda/DinitzV95}). 
 Let us denote a cycle $C$ of ${\mathcal H}_S$ by $k$-cycle if cycle $C$ has exactly $k$ nodes.
For a tree-edge $e_1$ (likewise a pair of cycle-edges $e_1,e_2$) adjacent to node $N$ in ${\mathcal H}_S$, we say that a \textit{node $N$ of ${\mathcal H}_S$ belongs to a tight cut} $C(A,e_1)$ (likewise $C(A,e_1,e_2)$) (refer to Definition \ref{def : tight cut}) if upon removal of edge(s) $e_1$ (likewise $e_1,e_2$) from ${\mathcal H}_s$, node $N$ and node $A$ remains connected in one of the two resulting subgraphs of ${\mathcal H}_S$. 
By construction of Connectivity Carcass, the following lemma holds for a cycle or a junction in ${\mathcal H}_S$ (proof is given in Appendix \ref{app : proof of cycle property} for completeness). 
\begin{lemma} [\cite{DBLP:conf/stoc/DinitzV94, DBLP:conf/soda/DinitzV95, DBLP:journals/siamcomp/DinitzV00}] \label{lem : cycle property}
    Let $C_k$ be a $k$-cycle, $k\ge 4$, (likewise $k$-junction, $k\ge 3$) in ${\mathcal H}_S$. Let $A$ and $B$ be a pair of nodes that belong to $C_k$ (likewise, for $k$-junction $A$ and $B$ are not the core node) and $A\ne B$. Suppose $A$ is adjacent to edges $e_A$, $e'_A$ (likewise edge $e_A$) and  $B$ is adjacent to edges $e_B$, $e_B'$ (likewise edge $e_B$) in $C_k$. Then, there is no vertex $u$ in $G$ such that $u\in C(A,e_A,e_A')\cap C(B,e_B,e_B')$ (likewise, $u\in C(A,e_A)\cap C(B,e_B)$).
\end{lemma}
Exploiting Lemma \ref{lem : cycle property}, we establish the following lemma that states the conditions in which the $S$-mincut capacity never increases upon insertion of any pair of edges.
\begin{lemma} \label{lem : condition of insertion}
    Upon insertion of any pair of edges in $G$, the capacity of $S$-mincut remains the same if Skeleton ${\mathcal H}_S$ has $(1)$ at least one $k$-cycle/$k$-junction where $k>4$, $(2)$ at least one $4$-cycle/$4$-junction and at least one $3$-junction, or $(3)$ at least three $3$-junctions.
\end{lemma}
\begin{proof}
    We prove each case by contradiction. So, assume to the contrary that the capacity of $S$-mincut increases upon insertion of edges $e,e'$.\\
    (1) Suppose there is a $k$-cycle $C_{k}$, $k>4$, in ${\mathcal H}_S$.  Therefore, the capacity of every $S$-mincut belonging to the bunch represented by each pair of cycle-edges of $C_k$ increases. 
    Using pigeon hole principle, it follows from Lemma \ref{lem : $S$-mincuts of a bunch increases} that there exists at least one endpoint of edges $e,e'$ that belongs to at least two tight cuts, say $C(M_1,(M_1,N_1),(M_1,N_1'))$ and $C(M_2,(M_2,N_2),(M_2,N_2'))$ where node $M_1$ of $C_k$ is adjacent to nodes $N_1,N_1'$ of $C_k$, node $M_2$ of $C_k$ is adjacent to nodes $N_2,N_2'$ of $C_k$, and $M_1\ne M_2$, a contradiction due to Lemma \ref{lem : cycle property}. The proof is similar for the case of $k$-junction, $k>4$. 

    \noindent
    (2) Suppose Skeleton ${\mathcal H}_S$ has a $4$-cycle $C_4=(A_1,A_2,A_3,A_4)$ and a $3$-junction $C_3=(B_1,B_2,B_3,B)$, where $B$ is the core node of $C_3$. Without loss of generality, assume that minimal cut represented by edges $(A_2,A_3)$ and $(A_3,A_4)$ separates nodes $A_1,A_2,A_4$ of $C_4$ from the node $A_3$ and the nodes of $C_3$. Similarly, without loss of generality, assume that minimal cut represented by edge $(B_1,B)$ separates nodes $B_2,B_3$ of $C_3$ from the node $B_1$ and the nodes of $C_4$.     
    Therefore, Using pigeon hole principle, it follows from Lemma \ref{lem : $S$-mincuts of a bunch increases} that there exist $5$ tight cuts ($C(A_1, (A_1,A_2), (A_1,A_4))$, $C(A_2, (A_2,A_1), (A_2,A_3))$, $C(A_4,(A_4,A_1),(A_4,A_3))$, $C(B_2,(B_2,B))$, and $C(B_3,(B_3, B))$) such that at least two tight cuts of them must contain at least one endpoint of the edges $e,e'$, a contradiction due to Lemma \ref{lem : cycle property}. The proof is along a similar line if $C_k$ is a $k$-junction, $k>4$.

   \noindent
   $(3)$ Suppose there are three $3$-junctions $C_3^A=(A_1,A_2,A_3,A)$,  $C_3^B=(B_1,B_2,B_3,B)$, and $C_3^D=(D_1,D_2,D_3,D)$ where $A,B$ and $D$ are the core nodes of $C_1,C_2,$ and $C_3$, respectively. 
   Without loss of generality, assume that minimal cut represented by edge $(A,A_3)$ separates nodes $A_1,A_2$ of $C_3^A$ from the node $A_3$ and the nodes of $C_3^B$, $C_3^D$. Since Skeleton is a cactus graph, without loss of generality, assume that minimal cut represented by edge $(D_2,D)$ separates nodes $D_1,D_3$ of $C_3^D$ from the node $D_2$ and the nodes of $C_3^A$, $C_3^B$. Since every edge of a $3$-junction is a tree-edge, therefore, there must exist at least one edge, say $(B_1,B)$, of $C_3^B$ such that the minimal cut defined by $(B_1,B)$ separates $B_1$ from $B_2,B_3$ and the nodes of $C_3^A$, $C_3^B$. Using pigeon hole principle, it follows from Lemma \ref{lem : $S$-mincuts of a bunch increases}, there exists $5$ tight cuts ($C(A_1,(A_1,A))$, $C(A_2,(A_2,A))$, $C(B_2,(B_2,B))$, $C(B_3, (B_3,B))$, $C(D_1,(D_1,D))$, $C(D_3,(D_3,D))$) such that at least two of them contains at least one endpoint of the edges $e,e'$, a contradiction due to Lemma \ref{lem : cycle property}.
   %
\end{proof}
Suppose ${\mathcal H}_S$ contains at least one $k$-cycle, $k\ge 5$. It follows from Lemma \ref{lem : condition of insertion} that the $S$-mincut remains the same after the insertion of any pair of edges in $G$. Let $C_k$ be a $k$-cycle in ${\mathcal H}_S$, where $k>4$. Let $\{A_1,A_2,\ldots,A_5\}$ be any five nodes of $C_k$ such that $A_i\ne A_j$, $i,j\in [5]$. Let us store five tight cuts $C(A_i,e_1^i,e_2^i)$, for $i\in [5]$, where $e_1^i$ and $e_2^i$ are the adjacent cycle-edges of node $A_i$ in $C_k$. It follows from the proof of Lemma \ref{lem : condition of insertion} that, upon insertion of any pair of edges, the capacity of at least one of these five cuts remains the same. We can report the tight cut for which the capacity remains $\lambda_S$ after the insertion of a pair of edges. 
The space occupied by this structure is ${\mathcal O}(n)$, and by Lemma \ref{lem : tight cut reporting}, ${\mathcal O}(n)$ time is required to report the tight cut. We do the same for $k$-junction, $k\ge 5$ and for the other cases of Lemma \ref{lem : condition of insertion} ($(2)$ and $(3)$). This leads to the following lemma.
\begin{lemma} \label{lem : does not increase conditions}
    Suppose Skeleton ${\mathcal H}_S$ has $(1)$ at least one $k$-cycle/$k$-junction, $k>4$, $(2)$ at least one $4$-cycle/$4$-junction and at least one $3$-junction, or $(3)$ at least three $3$-junctions. There is an ${\mathcal O}(n)$ space data structure that, upon insertion of any pair of edges, can report the capacity of $S$-mincut and an $S$-mincut for the resulting graph in ${\mathcal O}(1)$ and ${\mathcal O}(n)$ time, respectively, for each of the three cases $(1)$, $(2)$, and $(3)$.   
\end{lemma}

We now consider the cases for which the capacity of $S$-mincut may increase upon insertion of a pair of edges $e,e'$. It follows from Lemma \ref{lem : condition of insertion} that the $S$-mincut can increase for the following four cases. 
\begin{itemize}
    \item \textbf{Case 1:} ${\mathcal H}_S$ has exactly one $4$-cycle or $4$-junction, and has no $3$-junction, $k$-cycle, or $k$-junction, where $k\ge 5$. 
    \item \textbf{Case 2:} ${\mathcal H}_S$ has exactly two $3$-junctions and has no $k$-cycles or $k$-junctions, where $k\ge 4$.
    \item \textbf{Case 3:} ${\mathcal H}_S$ has exactly one $3$-junction and has no $k$-cycle or $k$-junction, where $k\ge 4$. 
    \item \textbf{Case 4:} ${\mathcal H}_S$ is a simple path.     
\end{itemize}
We now handle each case separately as follows. \\
\begin{figure}
 \centering
    \includegraphics[width=0.6\textwidth]{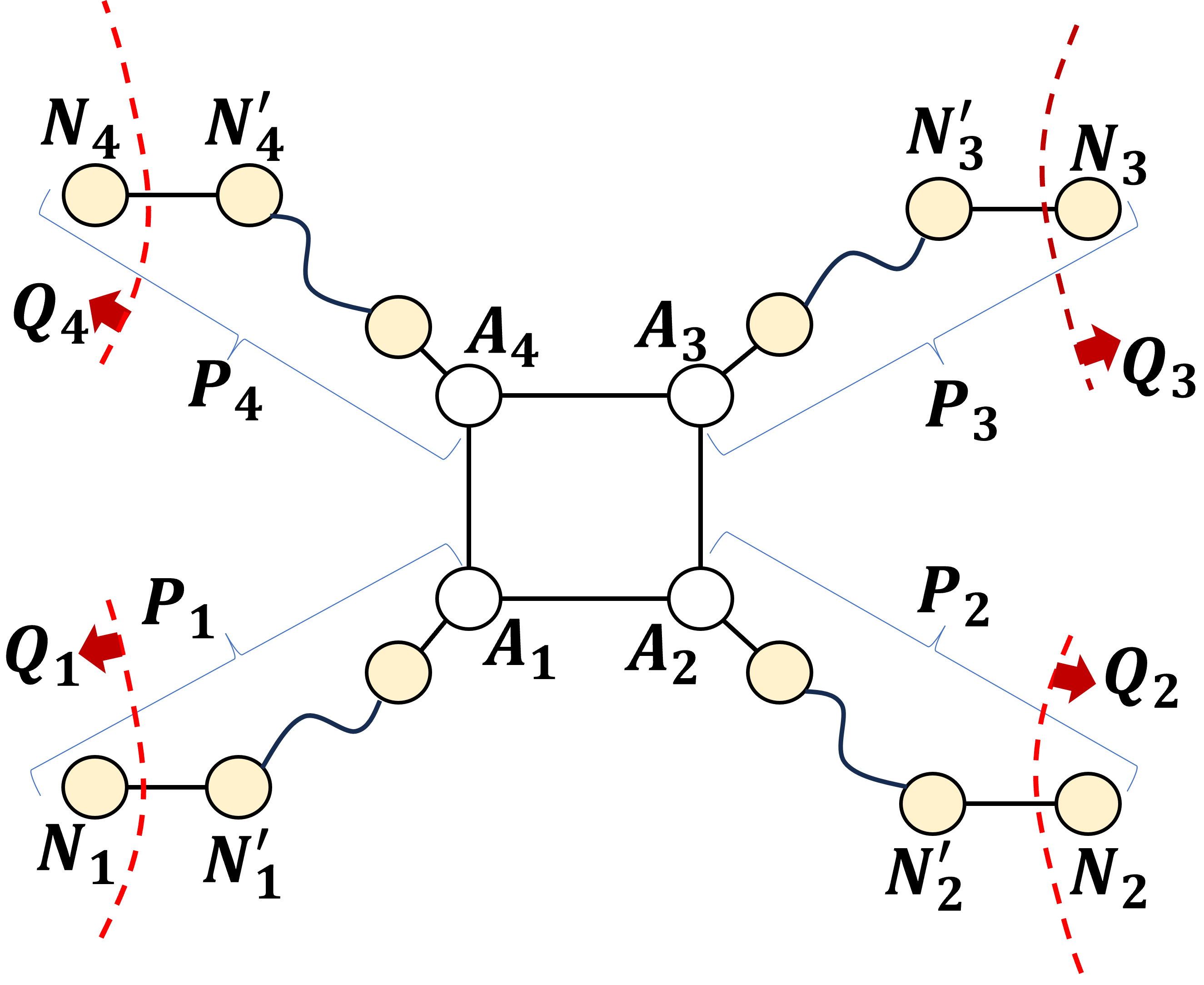} 
   \caption{A Skeleton containing exactly one 4-cycle and there are no $3$-junction, $k$-junction, and $k$-cycles, where $k>3$. Yellow vertices are Steiner vertices and white vertices are empty nodes. }
  \label{fig : dual insertion, four cycle}. 
\end{figure}

\noindent
\textbf{Case 1:} Let $C_4$ be a $4$-cycle in ${\mathcal H}_S$ such that the nodes of $C_4$ appears in the order $A_1,A_2,A_3,A_4,A_1$ (refer to Figure \ref{fig : dual insertion, four cycle}). Let us fix any $i\in [4]$. By construction of ${\mathcal H}_S$, since $A_i$ belongs to a cycle, node $A_i$ is an empty node. For node $A_i$, and its adjacent edges $e_1^i,e_2^i$ in $C_4$, there is an $S$-mincut $C(A_i,e_1^i,e_2^i)$. Since there is no other $3$-junction, $k$-junction, and $k$-cycle, where $k\ge 4$, observe that every edge not belonging to ${\mathcal H}_S$ is a tree-edge and also does not belong to any $k$-junction, $k>2$. So, there exists a path $P_i$ adjacent to $A_i$ containing at least one edge (the tree-edge adjacent to $A_i$). One endpoint of path $P_i$ is $A_i$. Let $N_i$ be the other endpoint of $P_i$ and adjacent to edge $(N_i,N_i')$ in $P_i$. Let $Q_i=C(N_i,(N_i,N_i'))$. Since $N_i$ is a node of degree one and $Q_i$ is an $S$-mincut, then there is a Steiner node of Flesh that is mapped to $N_i$. 
\begin{lemma} \label{lem : 4 cycle insertions}
    Upon insertion of edges $e,e'$, the capacity of $S$-mincut increases by $1$ if and only if one edge of $e,e'$ is between $Q_1=C(N_1,(N_1,N_1'))$ and $Q_3=C(N_3,(N_3,N_3'))$ and the other edge of $e,e'$ is between $Q_2=C(N_2,(N_2,N_2'))$ and $Q_4=C(N_4,(N_4,N_4'))$. 
\end{lemma}
\begin{proof}
    Suppose the capacity of $S$-mincut increases by $1$. Assume to the contrary that at least one tight cut $Q$ of the four tight cuts $Q_1,Q_2,Q_3,$ and $Q_4$ does not have any point of the two inserted edges. Then, by Lemma \ref{lem : $S$-mincuts of a bunch increases}, it is easy to observe that the capacity of tight cut $Q$ remains unchanged after the insertion of edges $e,e'$. Therefore, the capacity of $S$-mincut remains the same, a contradiction.  
    
    Let us consider any bunch represented by any tree-edge of path $P_i$. Observe that $C(S_{\mathcal B})$ (likewise $C(S\setminus S_{\mathcal B})$) contains one endpoint of one of the two edges $e,e'$. It follows from Lemma \ref{lem : $S$-mincuts of a bunch increases} that every $S$-mincut of bunch ${\mathcal B}$ increases by $1$. In a similar way, we can show that the capacity of every $S$-mincut belonging to the bunch ${\mathcal B}$ represented by a pair of cycle-edges of $C_4$ increases by at least $1$. Moreover, the capacity of tight cuts $Q_1,Q_2,Q_3$, and $Q_4$ increases by exactly $1$. This completes the proof.  
\end{proof}
Let us now consider a $4$-junction $C_4$ containing nodes $A_1,A_2,A_3,A_4,A$, where $A$ is the core node. Let us fix any $i\in[4]$. In a similar way to the $4$-cycle case, we can define tight cut $Q_i$ for node $A_i$. The following lemma shows the condition when the capacity of $S$-mincut increases for this case.
\begin{lemma} \label{lem : 4 junction insertions}
    Upon insertion of edges $e,e'$, the capacity of $S$-mincut increases by $1$ if and only if every tight cut $Q_i$, $i\in [4]$, contains exactly one endpoint of edges $e,e'$.
\end{lemma}
\begin{proof}
    Suppose the capacity of $S$-mincut increases by $1$. Assume to the contrary there is a tight cut $Q$ of the four tight cuts $Q_1,Q_2,Q_3$, and $Q_4$ does not contain any endpoint of edges $e,e'$ or contains more than one endpoint of edges $e,e'$. If $Q$ contains more than one endpoint of $e,e'$, then, by Definition \ref{def : tight cut}, there must exist another tight cut $Q'$ of the four tight cuts $Q_1,Q_2,Q_3$, and $Q_4$ such that $Q'$ contains no endpoint of $e,e'$. So we consider the case when $Q$ does not contain any endpoint of $e,e'$. In this case, by Lemma \ref{lem : $S$-mincuts of a bunch increases}, the capacity of tight cut $Q$ remains the same, a contradiction.  

    The proof of the converse part is along a similar line to the proof of the converse part of Lemma \ref{lem : 4 cycle insertions}. 
\end{proof}
It follows from Lemma \ref{lem : 4 cycle insertions} and Lemma \ref{lem : 4 junction insertions} that we only need to store the four tight cuts $Q_1,Q_2,Q_3,$ and $Q_4$ to verify whether the capacity of $S$-mincut increases by $1$ or does not increase upon insertion of edges $e,e'$. An $S$-mincut can also be reported by reporting one of the four tight cuts $Q_i, i\in [4]$, which can be done in ${\mathcal O}(n)$ time using Lemma \ref{lem : tight cut reporting}. Therefore, it leads to the following lemma.
\begin{lemma} \label{lem : 4-junction or 4-cycles}
    Suppose Skeleton ${\mathcal H}_S$ has exactly one $4$-cycle or $4$-junction, and has no $3$-junction, $k$-cycle, or $k$-junction, where $k\ge 5$. Then, there is an ${\mathcal O}(n)$ space data structure that, upon insertion of any pair of edges, can report the capacity of $S$-mincut and an $S$-mincut for the resulting graph in ${\mathcal O}(1)$ and ${\mathcal O}(n)$ time respectively.
\end{lemma}
\noindent
\textbf{Case 2:} Suppose Skeleton ${\mathcal H}_S$ has exactly two $3$-junctions $C_3^A=(A_1,A_2,A_3,A)$, $C_3^B=(B_1,B_2,B_3,B)$, where $A,B$ are the core nodes of $C_3^A,C_3^B$ respectively. There are no $k$-cycles or $k$-junctions, where $k\ge 4$. Therefore, there must exist a path $P$ consisting of only tree-edges between the core node $A$ of $C_3^A$ and the core node $B$ of $C_3^B$. Without loss of generality, assume that node $A_3$ and node $B_3$ belong to $P$. For every $i\in [2]$, in a similar way to the $4$-cycle case in Case 1, we can define tight cut $Q_i^A$ for node $A_i$ and tight cut $Q_i^B$ for node $B_i$.

\begin{lemma} \label{lem : case 2 main lemma}
    Upon insertion of pair of edges $e,e'$ in $G$, the capacity of $S$-mincut increases by $1$ if and only if one edge is added between $Q_1^A,Q_2^B$ and the other edge is added between $Q_2^A,Q_1^B$ or one edge is added between $Q_1^A,Q_1^B$ and the other edge is added between $Q_2^A,Q_2^B$.  
\end{lemma}
\begin{proof}
    Suppose upon insertion of pair of edges $e,e'$ in $G$, the capacity of $S$-mincut increases by $1$. Assume to the contrary that one edge is added between $Q_1^A,Q_2^A$ and the other edge is added between $Q_1^B,Q_2^B$ or at least one tight cut $Q$ of four tight cuts $Q_1^A,Q_2^A,Q_1^A,Q_1^B$ does not contain any endpoint of the two edges $e,e'$. In the latter case, it is easy to establish that the capacity of $S$-mincut $Q$ remains unchanged in $G$, a contradiction. In the former case, observe that the capacities of both tight cuts $C(A_3,(A_3,A))$ and $C(B_3,(B_3,B))$ in $G$ do not change upon insertion of edges $e,e'$, a contradiction.

    Suppose one edge is added between $Q_1^A,Q_2^B$ and the other edge is added between $Q_2^A,Q_1^B$ or one edge is added between $Q_1^A,Q_1^B$ and the other edge is added between $Q_2^A,Q_2^B$. By Lemma \ref{lem : $S$-mincuts of a bunch increases}, we can show in a similar way as in the proof of the converse part of Lemma \ref{lem : 4 cycle insertions}, the capacity of every $S$-mincut belonging to each of the bunch defined by every minimal cut of ${\mathcal H}_S$ increases by $1$. 
\end{proof}
It follows from Lemma \ref{lem : case 2 main lemma} that it is sufficient to store the five tight cuts $Q_1^A,Q_2^A,Q_1^A,Q_1^B$, and $C(A_3,(A_3,A))$ to determine if the capacity of $S$-mincut increases in ${\mathcal O}(1)$ time. It can report an $S$-mincut after the insertion of pair of edges in ${\mathcal O}(n)$ time using Lemma \ref{lem : tight cut reporting}. 
Hence, it completes the proof of the following lemma.
\begin{lemma} \label{lem : case 2}
    Suppose Skeleton ${\mathcal H}_S$ has exactly two $3$-junctions and has no $k$-cycles or $k$-junctions, where $k\ge 4$. Then, there is an ${\mathcal O}(n)$ space data structure that, upon insertion of any pair of edges, can report the capacity of $S$-mincut and an $S$-mincut for the resulting graph in ${\mathcal O}(1)$ and ${\mathcal O}(n)$ time respectively.
\end{lemma}

\noindent
\textbf{Case 3:} Suppose the Skeleton has exactly one $3$-junction, and there is no $k$-cycle or $k$-junction for $k\ge 4$. Let $C_3=(A_1,A_3,A_3,A)$ be the $3$-junction.  For every $i\in [3]$, in a similar way to the $4$-cycle case in Case 1, we can define tight cut $Q_i$ for node $A_i$. 

Observe that the $S$-mincut increases by exactly $1$ if and only if the two edges are inserted in one of the three following ways -- (a) One edge is between $Q_1,Q_2$ and the other edge is between $Q_1,Q_3$, (b) one edge is between $Q_1,Q_3$ and the other edge is between $Q_2,Q_3$, and (c) one edge is between $Q_1,Q_2$ and the other edge is between $Q_2,Q_3$. So, similar to Case 1 and Case 2, it is sufficient to store $Q_1,Q_2,$ and $Q_3$ for determining whether the capacity of $S$-mincut increases by $1$ after insertion of edges $e,e'$. This leads to the following lemma.
\begin{lemma} \label{lem : case 3}
      Suppose the Skeleton has exactly one $3$-junction, and there is no $k$-cycle or $k$-junction for $k\ge 4$. Then, there is an ${\mathcal O}(n)$ space data structure that, upon insertion of any pair of edges, can report the capacity of $S$-mincut and an $S$-mincut for the resulting graph in ${\mathcal O}(1)$ and ${\mathcal O}(n)$ time respectively. 
\end{lemma}

\noindent
\textbf{Case 4:} Suppose Skeleton ${\mathcal H}_S$ is a path between two nodes $S$ and $T$. It is easy to observe that both $S$ and $T$ are Steiner nodes. Without causing any ambiguity, we say that $S$ and $T$ are the corresponding $(\lambda_S+1)$ classes. We first assume that $S$ and $T$ are Singleton $(\lambda_S+1)$ classes and the corresponding Steiner vertices are $s$ and $t$, respectively. Later, we show how to extend it to general $(\lambda_S+1)$ classes containing zero or multiple Steiner vertices.

Upon insertion of the pair of edges $e,e'$, by Lemma \ref{lem : $S$-mincuts of a bunch increases}, the $S$-mincut definitely increases by $1$ if at least one edge lies between $S$ and $T$. Suppose one endpoint of one edge (likewise the other edge) belongs to $S$ (likewise $T$) and the other endpoint does not belong to $T$ (likewise $S$). Without loss of generality, assume that  $x\in S$ and $y'\in T$ but $x'\notin S$ and $y\notin T$. The following property holds for this case. 
\begin{lemma} 
\label{lem : belong to nearest cut}
    The capacity of $S$-mincut increases by $1$ if and only if $x'$ belongs to the nearest $(s,t)$-mincut of $y$.
\end{lemma}
\begin{proof}
    Suppose $x'$ does not belong to the nearest $(s,t)$-mincut of $y$. In that case, the nearest $(s,t)$-mincut of $y$ remains unaffected.

    Suppose $x'$ belongs to the nearest $(s,t)$-mincut of $y$. Upon insertion of $(x,y)$, the capacity of every $S$-mincut that keeps $y$ on the side of $t$ increases by $1$. The $S$-mincut that are left are those that keeps $y$ on the side of $s$. Since $x'$ belongs to the nearest $(s,t)$-mincut of $y$, therefore, by Definition \ref{def : nearest s,t mincut of a vertex}, it follows that the capacity of every $S$-mincut increases by $1$.
\end{proof}
We store the matrix ${\mathcal M}$ (refer to Lemma \ref{lem : matrix M} in Appendix \ref{sec : projected to skeleton}) for all Stretched units of $G$. It occupies ${\mathcal O}((n-|S|)^2)$ space. By using quotient mapping $\phi$, matrix ${\mathcal M}$, and projection mapping $\pi$, it follows from Lemma \ref{lem : matrix M} that we can verify in ${\mathcal O}(1)$ time the condition of Lemma \ref{lem : belong to nearest cut}. The $S$-mincut for the resulting graph can be reported by reporting the nearest $(s,t)$-mincut of $y$ in ${\mathcal O}(n)$ time using ${\mathcal M}$, projection mapping $\pi$, and the quotient mapping $\phi$. This leads to the following lemma.
\begin{lemma} \label{lem : increase by 1 when in nearest cut of a vertex}
    Suppose one endpoint of one inserted edge (likewise the other inserted edge) belongs to $S$ (likewise $T$) and the other endpoint does not belong to $T$ (likewise $S$). Then, there is an ${\mathcal O}((n-|S|)^2+n)$ space data structure that can report the capacity of $S$-mincut and an $S$-mincut for the resulting graph in ${\mathcal O}(1)$ and ${\mathcal O}(n)$ time respectively.
\end{lemma}
We now verify whether the capacity of $S$-mincut increases by $2$. Suppose both edges are between $S$ and $T$. Without loss of generality, assume that $x,x'\in S$ and $y,y'\in T$. Let $C$ be a $(\lambda_S+1)$ cut that has both $s$ and $t$ on the same side. In this case, after the insertion of both edges $e,e'$, the capacity of $S$-mincut becomes $\lambda_S+1$. Storing and reporting cut $C$ requires ${\mathcal O}(n)$ space and ${\mathcal O}(n)$ time respectively. Now, we assume that such a cut $C$ does not exist. Then, the following lemma is easy to establish. 
\begin{lemma} \label{lem : increase by 2 condition}
    The capacity of $S$-mincut increases by $2$ if and only if 
    \begin{enumerate}
        \item there is no $(\lambda_S+1)$ cut that separates $x,x'$ from $s$ and 
        \item there is no $(\lambda_S+1)$ cut that separates $y,y'$ from $t$.
    \end{enumerate}
\end{lemma}
We can trivially store, for every pair $u,v$ of nonSteiner vertices belonging to $S$, if $s$ is separated from $u,v$ by a $(\lambda_S+1)$ cut. Similarly, we store for every pair of nonSteiner vertices belonging to $T$. The space occupied is ${\mathcal O}((n-|S|)^2)$. It requires ${\mathcal O}(1)$ time to verify the condition of Lemma \ref{lem : increase by 2 condition}. If condition of Lemma \ref{lem : increase by 2 condition} is satisfied, then $S$-mincut increases by $2$, otherwise by $1$.\\

\noindent
\textbf{Reporting an $S$-mincut for Resulting Graph:} Let us consider a $(\lambda_S+1)$ class ${\mathcal W}\in \{S,T\}$. Observe that there is only one Steiner vertex other than the Steiner vertex belonging to ${\mathcal W}$. We store the data structure ${\mathcal Q}_S({\mathcal W})$ of Theorem \ref{thm : data structure for singleton class} only for the Steiner vertex belonging to ${\mathcal W}$. Therefore, it follows from Theorem \ref{thm : data structure for singleton class} that the space occupied by ${\mathcal Q}_S({\mathcal W})$ is ${\mathcal O}(|\overline{S({\mathcal W})}|)$. Suppose ${\mathcal W}=S$. If for the pair of vertices $x,x'$, there is a $(\lambda_S+1)$ cut $C$ of $G$ such that $s\in C$ and $x,x'\in \overline{C}$, then we report $C\cap S$. To report the cut $C$ completely, we also report every vertex $x$ of $G$ such that $\phi(x)$ that does not belong to $S$. We do the same for $(\lambda_S+1)$ class $T$. Now suppose there is no pair of $(\lambda_S+1)$ cuts $C,C'$ such that $s\in C$, $x,x'\in \overline{C}$ and  $t\in C'$, $y,y'\in \overline{C'}$. This means that the condition of Lemma \ref{lem : increase by 2 condition} is satisfied. In this case, each $S$-mincut of $G$ becomes a Steiner cut of capacity $\lambda_S+2$ and there is no $(\lambda_S+1)$ cut that separates $s$ from $t$. Therefore, we can report set $S$ as a Steiner cut of capacity $\lambda_S+2$ in ${\mathcal O}(n)$ time. This leads to the following lemma.
\begin{lemma} \label{lem : data structure for dual edge insertion in singleton class}
    Suppose ${\mathcal H}_S$ is a path. Let the first and last node of the path be $S$ and $T$, respectively. If both $S$ and $T$ correspond to Singleton $(\lambda_S+1)$ class, then there is an ${\mathcal O}((n-|S|)^2+n)$ space data structure that, upon insertion of any pair of edges, can report the capacity of $S$-mincut and an $S$-mincut for the resulting graph in ${\mathcal O}(1)$ and ${\mathcal O}(n)$ time respectively.
\end{lemma}
\noindent
\textbf{Handling Dual Edge Insertion for Generic Steiner nodes:} Let ${\mathcal W}\in \{S,T\}$ be a $(\lambda_S+1)$ class containing more than one Steiner vertices. Similar to the construction of data structure in Lemma \ref{lem : data structure for generic steiner node} for all Steiner nodes, we construct the pair of graphs $G({\mathcal W})_1$ and $G({\mathcal W})_2$ from $G$ that satisfy the properties of Lemma \ref{lem : all cuts in gw1} and Lemma \ref{lem : all cuts GW2}.

The $(\lambda_S+1)$ class ${\mathcal W}_1$ of $G({\mathcal W})_1$ is a Singleton $(\lambda_S+1)$ class to which all vertices of ${\mathcal W}$ are mapped. We construct data structure of Lemma \ref{lem : data structure for dual edge insertion in singleton class} for ${\mathcal W}_1$ for graph $G({\mathcal W})_1$. 

We now consider graph $G({\mathcal W})_2$. By Lemma \ref{lem : all cuts GW2}, every Steiner mincut of $G({\mathcal W})_2$ is a $(\lambda_S+1)$ cut of $G$ that subdivides the Steiner set of ${\mathcal W}$. Therefore, to determine if there is no $(\lambda_S+1)$ cut that separates a Steiner vertex of ${\mathcal W}$ and $x,x'$ in $S$ (likewise $y,y'$ in $T$), it is sufficient to determine if the capacity of Steiner mincut of $G({\mathcal W})_2$ increases by at least $1$ upon insertion of pair of edges $e,e'$. So, we store the data structures of Lemma \ref{lem : does not increase conditions}, Lemma \ref{lem : 4-junction or 4-cycles}, Lemma \ref{lem : case 2}, Lemma \ref{lem : case 3}, and Lemma \ref{lem : increase by 1 when in nearest cut of a vertex} for graph $G({\mathcal W})_2$. Since ${\mathcal W}$ is only $S$ or $T$, this leads to the following lemma.

\begin{lemma} \label{lem : case 4}
    Suppose ${\mathcal H}_S$ is a path. Let the first and last node of the path be $S$ and $T$, respectively. Then there is an ${\mathcal O}((n-|S|)^2+n)$ space data structure that, upon insertion of any pair of edges, can report the capacity of $S$-mincut and an $S$-mincut for the resulting graph in ${\mathcal O}(1)$ and ${\mathcal O}(n)$ time respectively.
\end{lemma}
Finally, we have designed a data structure ${\mathcal I}_S$ for all four cases for handling dual edge insertion. Therefore, the following theorem follows from Lemma \ref{lem : does not increase conditions}, Lemma \ref{lem : 4-junction or 4-cycles}, Lemma \ref{lem : case 2}, Lemma \ref{lem : case 3}, and Lemma \ref{lem : case 4}. 
\begin{theorem} [Dual Edge Insertion] \label{thm : dual edge insertion}
     Let $G=(V,E)$ be an undirected multi-graph on $n=|V|$ vertices and $m=|E|$ edges. For any Steiner set $S\subseteq V$, there is an ${\mathcal O}((n-|S|)^2+n)$ space data structure that, after the insertion of any pair of edges in $G$, can report the capacity of $S$-mincut in ${\mathcal O}(1)$ time and an $S$-mincut in ${\mathcal O}(n)$ time. 
\end{theorem}
Theorem \ref{thm : dual edge insertion} leads to Theorem \ref{thm : dual edge failure} for insertion of a pair of edges.

\section{Lower Bound} \label{sec : lower bound oracle}


In this section, we provide a lower bound on space and query time for any dual edge Sensitivity Oracle for $S$-mincut. We first establish a lower bound in directed graphs and then extend it to undirected graphs. In directed multi-graphs, the edge-set of any cut $C\subset V$ is defined as the number of edges that leave $C$ and enter $\overline{C}$. Similarly, the capacity of $C$ is also defined as the number of edges in the edge-set of cut $C$.

We begin by defining a graph $H=(V_H,E_H)$ on $n$ vertices and $m$ edges as follows. The vertex set $V_H$ is divided into three sets $S,L,R$ where $S$ is a set of isolated vertices, $L$ and $R$ forms a $L\times R$ bipartite graph on $m$ edges. The number of vertices in $S$ is at least $2$ and at most $n$, the number of vertices in $L$ is $\lfloor \frac{n-|S|}{2} \rfloor$, and the number of vertices in $R$ is $\lfloor\frac{n-|S|+1}{2}\rfloor$. Each edge is unweighted and oriented from set $L$ to set $R$. Let ${\mathcal B}$ be the class of all graphs $H$. Given an instance $B$ of ${\mathcal B}$, we construct the following dag ${\mathcal D}(B)$ on ${\mathcal O}(n)$ vertices and ${\mathcal O}(m)$ edges. \\

\noindent
\textbf{Construction of ${\mathcal D}(B)$:} Let $t$ be any vertex from $S$. Contract the set of vertices in $S\setminus \{t\}$ into a single vertex $s$. Let $s$ and $t$ form the Steiner set of ${\mathcal D}(B)$. Consider vertex $s$ as a source vertex and $t$ as a sink vertex. For each vertex $u\in L$ with $p>0$ outgoing edges, add $p$ edges from $s$ to $u$. Similarly, for each vertex $u\in R$ with $p>0$ incoming edges, add $p$ edges from $u$ to $t$. For each vertex $u$ in $L\cup R$, add two edges -- one from $s$ to $u$ and the other from $u$ to $t$.\\    

 By construction of ${\mathcal D}(B)$, the following lemma holds.
\begin{lemma} \label{lem : path exists in u,v iff u,v edge exists}
    For any pair of vertices in $B$, $(u,v)$ edge is in $B$ if and only if there exists a path from $u$ to $v$ in ${\mathcal D}(B)$.
\end{lemma}

For any directed graph $G'$ with Steiner set $\{s',t'\}$, the following graph is constructed by Picard and Queyranne in \cite{DBLP:journals/mp/PicardQ80} for storing and characterizing all $(s',t')$-mincut of $G'$.\\

\noindent
Description of ${\mathcal D}_{PQ}(G')$: A directed acyclic graph with source $T'$ and sink $S'$ that stores all $(s',t')$-mincuts of $G'$ and characterize them as follows (Lemma \ref{lem : characterization of s,t mincut}). A $1$-transversal cut is an $(s',t')$-cut whose edge-set does not intersect any simple path more than once.
\begin{lemma}\label{lem : characterization of s,t mincut}
    An $(s',t')$-cut $C$ is an $(s',t')$-mincut if and only if $C$ is a $1$-transversal cut in ${\mathcal D}_{PQ}(G')$,
\end{lemma}
The dag ${\mathcal D}_{PQ}(G')$ can be obtained from $G'$ as follows. Let $G'_r$ be the residual graph of $G$ for maximum $(s',t')$-flow. By contracting every strongly connected component of the residual graph $G'_r$, we arrive at ${\mathcal D}_{PQ}(G')$. The source vertex $s'$ belongs to node $S'$ and sink vertex $t'$ belongs to node $T'$ of ${\mathcal D}_{PQ}(G')$.\\

Observe that each vertex, except source $s$ and sink $t$, in ${\mathcal D}(B)$ has the number of incoming edges the same as the number of outgoing edges. Exploiting this fact, Baswana, Bhanja, and Pandey in \cite{DBLP:journals/talg/BaswanaBP23} established an important property of ${\mathcal D}(B)$ as follows.
\begin{lemma}[Lemma 8.3 and Theorem B.3 in \cite{DBLP:journals/talg/BaswanaBP23}] \label{lem : D_PQ is same as D}
    Let ${\mathcal D}_{R}(B)$ be the graph obtained from ${\mathcal D}(B)$ after flipping the orientation of each edge in ${\mathcal D}(B)$. ${\mathcal D}_{PQ}({\mathcal D}(B))$ is the same as graph ${\mathcal D}_R(B)$.
\end{lemma}
Lemma \ref{lem : D_PQ is same as D} plays a key role in establishing the following close relations (Lemma \ref{lem : relation between edge and dual edge failure} and Lemma \ref{lem : relation between edge and dual edge insertion}) between the existence of an edge in graph $B$ and the increase/reduction of capacity in $(s,t)$-mincut in graph ${\mathcal D}(B)$ after the insertion/failure of a pair of edges in/from ${\mathcal D}(B)$. 
\begin{lemma} \label{lem : relation between edge and dual edge failure}
    Upon failure of a pair of edges $(s,u)$ and $(v,t)$ in ${\mathcal D}(B)$, the capacity of $(s,t)$-mincut in ${\mathcal D}(B)$ reduces by exactly $1$ if and only if edge $(u,v)$ is present in graph $B$. 
\end{lemma}
\begin{proof}
    Suppose the capacity of $(s,t)$-mincut in ${\mathcal D}(B)$ reduces by exactly $1$. Assume to the contrary that there is no edge $(u,v)$ in $B$. So, it follows from Lemma \ref{lem : path exists in u,v iff u,v edge exists} and Lemma \ref{lem : D_PQ is same as D} that $u$ is not reachable from $v$ in ${\mathcal D}_{PQ}({\mathcal D}(B))$. As a result, observe that $(s,u)$ edge is a contributing edge of $1$-transversal cut $C$ defined by the set of vertices reachable from $v$ in ${\mathcal D}_{PQ}({\mathcal D}(B))$. Since, by Lemma \ref{lem : characterization of s,t mincut}, each $1$-transversal cut is an $(s,t)$-mincut in ${\mathcal D}_{PQ}({\mathcal D}(B))$, $C$ is an $(s,t)$-mincut in ${\mathcal D}(B)$. Therefore, the capacity of $(s,t)$-mincut reduces by $2$ as both edges $(s,u)$ and $(v,t)$ are contributing to $C$, a contradiction.      
    
    We now prove the converse part. It follows from Lemma \ref{lem : D_PQ is same as D} that, in ${\mathcal D}(B)$, there is an $(s,t)$-mincut to which $(s,u)$ is contributing, and there is also an $(s,t)$-mincut to which $(v,t)$ is contributing. So, the capacity of $(s,t)$-mincut in ${\mathcal D}(B)$ definitely reduces by at least $1$ after the failure of both the edges $(s,u)$ and $(v,t)$. Now, suppose edge $(u,v)$ is present in $B$. 
    By Lemma \ref{lem : path exists in u,v iff u,v edge exists} and Lemma \ref{lem : D_PQ is same as D}, 
    $u$ is reachable from $v$ in ${\mathcal D}_{PQ}({\mathcal D}(B))$. So, for each 1-transversal cut in ${\mathcal D}_{PQ}({\mathcal D}(B))$ to which edge $(v,t)$ is contributing, vertex $u$ must lie on the side of $s$. Therefore, there is no $1$-transversal cut in ${\mathcal D}_{PQ}({\mathcal D}(B))$ in which both edges $(s,u)$ and $(v,t)$ are contributing. Since, by Lemma \ref{lem : characterization of s,t mincut}, only $1$-transversal cuts in ${\mathcal D}_{PQ}({\mathcal D}(B))$ are $(s,t)$-mincuts, therefore, the capacity of $(s,t)$-mincut in ${\mathcal D}(B)$ reduces by exactly $1$.
\end{proof}

\begin{lemma} \label{lem : relation between edge and dual edge insertion}
    Upon insertion of a pair of edges $(s,v)$ and $(u,t)$ in ${\mathcal D}(B)$, the capacity of $(s,t)$-mincut in ${\mathcal D}(B)$ increases by exactly $1$ if and only if edge $(u,v)$ is present in graph $B$. 
\end{lemma}
\begin{proof}
    Suppose the capacity of $(s,t)$-mincut in ${\mathcal D}(B)$ increases by exactly $1$. Assume to the contrary that $(u,v)$ edge is not present in $B$. In graph ${\mathcal D}_{PQ}({\mathcal D}(B))$, by Lemma \ref{lem : path exists in u,v iff u,v edge exists} and Lemma \ref{lem : D_PQ is same as D}, $u$ is not reachable from $v$. If follows from Lemma \ref{lem : characterization of s,t mincut} that the set of vertices $R_v$ that are reachable from $v$ defines a $(s,t)$-mincut since it is a $1$-transversal cut, and $u\notin R_v$. Therefore, observe that the capacity of $R_v$ is not affected by the insertion of the pair of edges $(s,v)$ and $(u,t)$. So, the capacity of $(s,t)$-mincut remains the same, a contradiction.

    Suppose $(u,v)$ edge is present in $B$. By Lemma \ref{lem : path exists in u,v iff u,v edge exists} and Lemma \ref{lem : D_PQ is same as D}, in graph ${\mathcal D}_{PQ}({\mathcal D}(B))$, $u$ is reachable from $v$. Upon insertion of edge $(s,v)$, every $(s,t)$-mincut $C$ survives that keeps $v$ on the side of $s$. The capacity of other $(s,t)$-mincuts are increased by $1$. We know that $u$ is reachable from $v$ in ${\mathcal D}_{PQ}({\mathcal D}(B))$. So, for every $1$-transversal cut $C$ of ${\mathcal D}_{PQ}({\mathcal D}(B))$ that keeps $v$ on the side of $s$, $C$ also keeps $u$ on the side of $s$. Therefore, upon insertion of edge $(u,t)$, the capacity of every $1$-transversal cut in ${\mathcal D}_{PQ}({\mathcal D}_B)$ with $v$ is on the side of $s$ increases by $1$. So, the capacity of every $(s,t)$-mincut increases by exactly $1$.  
\end{proof}
Let ${\mathcal F}(B)$ be a data structure that can report the capacity of $(s,t)$-mincut after the failure of any pair of edges in graph ${\mathcal D}(B)$. For any given pair of vertices $u\in L$ and $v\in R$, it follows from Lemma \ref{lem : relation between edge and dual edge failure} that ${\mathcal F}(B)$ can determine whether $(u,v)$ edge is present in $B$. For a pair of instances $B_1$ and $B_2$ from ${\mathcal B}$, there must exist at least one edge $(u,v)$ such that $(u,v)$ is present in $B_1$ but not in $B_2$ or vice versa.  
Therefore, it follows from Lemma \ref{lem : relation between edge and dual edge failure} that the encoding of ${\mathcal F}(B_1)$ must differ from the encoding of ${\mathcal F}(B_2)$. Observe that there can be $\Omega(2^{(n-|S|)^2})$ possible instances in ${\mathcal B}$. Hence, there exists an instance $B$ from ${\mathcal B}$ such that ${\mathcal F}(B)$ requires $\Omega((n-|S|)^2)$ bits of space. 
We can argue in the similar way using Lemma \ref{lem : relation between edge and dual edge insertion} for insertion of a pair of edges. This completes the proof of the following lemma.
\begin{lemma} \label{lem : dual edge failure directed}
    Let $G=(V,E)$ be a directed multi-graph on $n$ vertices and $S\subseteq V$ be a Steiner set. Any data structure that, after the insertion/failure of any pair of edges from $E$, can report the capacity of $S$-mincut must occupy $\Omega((n-|S|)^2)$ bits of space in the worst case, irrespective of query time.  
\end{lemma}
Baswana, Bhanja, and Pandey \cite{DBLP:journals/talg/BaswanaBP23} also established the following relation between directed and undirected graphs. A dag with a single source vertex and a single sink vertex is said to be a \textit{balanced dag} if each node, except the source and sink, has the same number of incoming and outgoing edges. 
\begin{lemma}[Theorem 8.5 and Theorem B.3 in \cite{DBLP:journals/talg/BaswanaBP23}] \label{lem : undirected graph exists}
    Let $D$ be a balanced dag with a single source vertex and a single sink vertex. Let $D_R$ be the graph obtained after reversing the direction of each edge of $D$. Let $H$ be an undirected graph obtained by replacing every directed edge of $D$ with an undirected edge. Then, ${\mathcal D}_{PQ}(H)$ is the same as graph $D_R$. 
\end{lemma}
It follows from the construction that graph ${\mathcal D}(B)$ is a balanced dag. Therefore, by Lemma \ref{lem : undirected graph exists}, there exists an undirected graph $H$ of ${\mathcal D}(B)$ such that ${\mathcal D}_{PQ}(H)$ is the same as the graph ${\mathcal D}_R(B)$, where ${\mathcal D}_R(B)$ is obtained from ${\mathcal D}(B)$ after reversing the direction of each edge. 
Hence, along a similar line of Lemma \ref{lem : relation between edge and dual edge failure} and Lemma \ref{lem : relation between edge and dual edge insertion}, and using ${\mathcal D}_R(B)$, the following lemma can be established.
\begin{lemma} \label{lem : undirected lower bound}
    Upon failure (likewise insertion) of any pair of edges $(s,u)$ and $(v,t)$ (likewise $(s,v)$ and $(u,t)$) in undirected graph $H$, the capacity of $S$-mincut in $H$ decreases (likewise increases) by $1$ if and only if $(u,v)$ edge is present in $B$.
\end{lemma}
Exploiting Lemma \ref{lem : undirected lower bound} and, along similar lines of the proof of Lemma \ref{lem : dual edge failure directed}, we can establish Theorem \ref{thm : lower bound on dual edge failure}.



\section{Conclusion} \label{sec : conclusion}
The problem of designing a compact data structure for minimum+1 Steiner cuts and designing dual edge Sensitivity Oracle for Steiner mincut is addressed for the first time in this manuscript. Our results provide a complete generalization of the two important existing results while matching their bounds -- for global cut (Steiner set is the vertex set $V$) by Dinitz and Nutov \cite{DBLP:conf/stoc/DinitzN95} and for $(s,t)$-cut (Steiner set is only a pair of vertices $\{s,t\}$) by Baswana, Bhanja, and Pandey \cite{DBLP:journals/talg/BaswanaBP23}. The space occupied by the dual edge Sensitivity Oracle for Steiner min-cut, presented in this manuscript, is optimal almost for the whole range of $S$. In addition, the query time for reporting Steiner mincut and its capacity is also worst case optimal.

An immediate future work would be to design a Sensitivity Oracle for Steiner mincut that handles the insertions or failures of any $f>2$ edges. For every Steiner set, our result improves the space by a factor of $\Omega(\frac{m}{n})$ over the trivial result 
for handling any $f$ edge insertions or failures for Steiner mincut.
We have crucially exploited the fact that a $(\lambda_S+1)$ cut cannot subdivide more than one $(\lambda_S+1)$ classes for designing dual edge Sensitivity Oracle. However, this no longer holds for cuts of capacity greater than $\lambda_S+1$, which are essential for handling any $f>2$ edge insertions or failures. In conclusion, while we believe that our results on $(\lambda_S+1)$ cuts must be quite useful for addressing the generalized version of the problem, some nontrivial insights/techniques are also required.

\bibliography{main}
\newpage
\appendix

\section{Proof of Theorem \ref{thm : gen 3 star}(3)} \label{app : gen 3 star extended}

  \begin{figure}
 \centering
    \includegraphics[width=\textwidth]{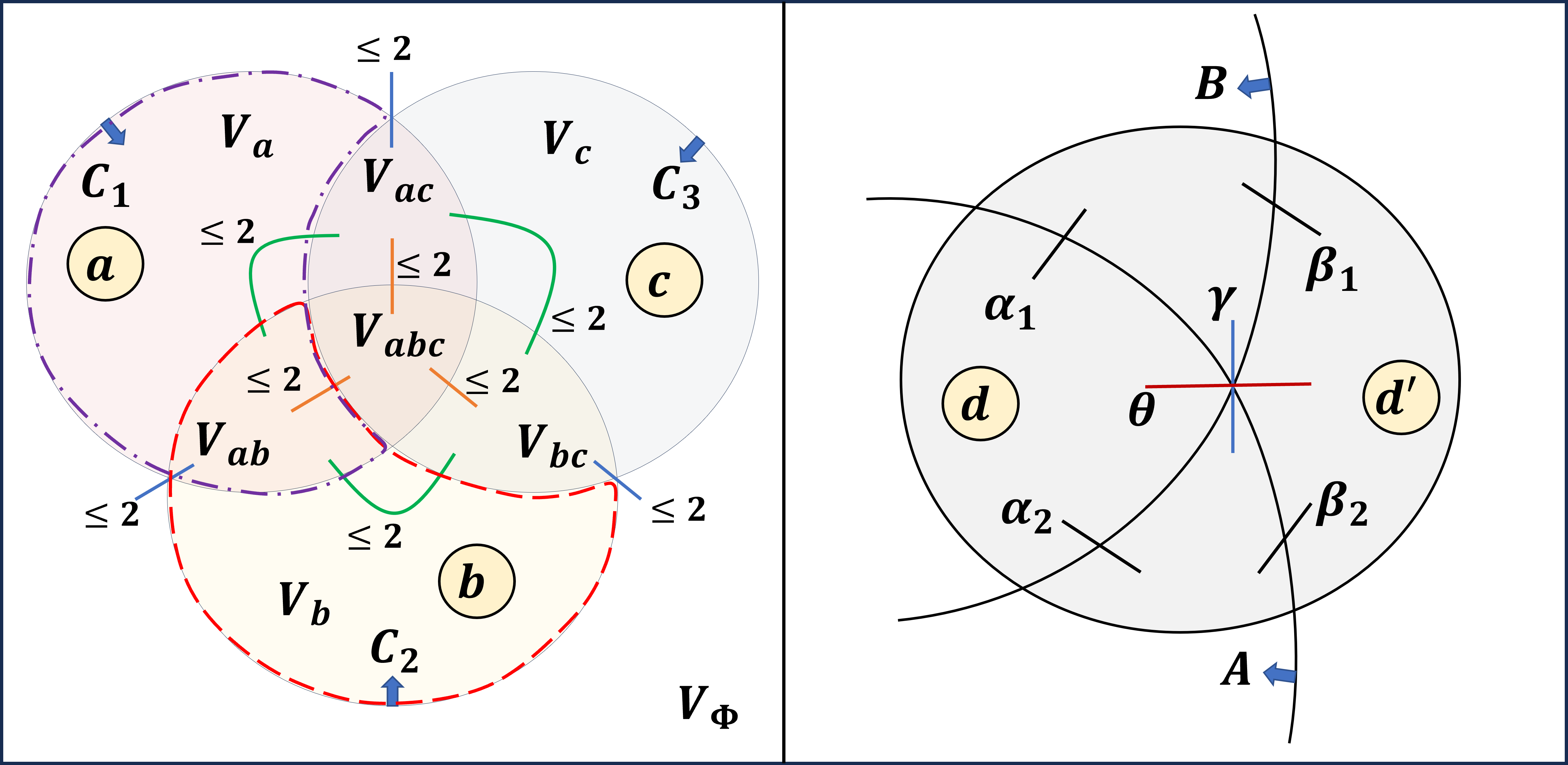} 
   \caption{$(i)$ Eight disjoint vertex sets formed using three $(\lambda_S+1)$ cuts. The red dashed curve represents $D_1$ and the blue dashed-dot curve represents $D_2$. $(ii)$ A pair of Steiner cuts forming four subsets of $V$, and their adjacency information.} 
  \label{fig : extended gen 3 star}. 
\end{figure}

We refer to Figure \ref{fig : extended gen 3 star} for better understanding. Let $A$ and $B$ be a pair of Steiner cuts in $G$. Let there be a pair of Steiner vertices $d,d'$ such that $d\in A\cap B$ and $d'\in \overline{A\cup B}$. For any four integers $p,q,i,j\ge 0$,  $c(A)=\lambda_S+p$, $c(B)=\lambda_S+q$, $c(A\cap B)=\lambda_S+i$, and $c(A\cup B)=\lambda_S+j$. We now establish the following lemma for the number of edges that lie between $A\setminus B$ and $B\setminus A$, which can be seen as a generalization of Lemma 5.4 in \cite{DBLP:journals/talg/BaswanaBP23} in undirected graphs.
\begin{lemma} \label{lem : gamma value}
 Let $\gamma$ denote the number of edges between $A\setminus B$ and $B\setminus A$. Then, $\gamma=\frac{p+q-i-j}{2}$.
\end{lemma}
\begin{proof}
    The following equations follow directly from Figure \ref{fig : extended gen 3 star}$(ii)$.
    \begin{align}
    & c(A\cap B)= \alpha_1+\alpha_2+\theta=\lambda_S+i  \label{eq : A cap B} \\
    & c(A\cup B)= \beta_1+\beta_2+\theta=\lambda_S+j\label{eq : A cup B} \\
    & c(B)=\beta_1+\alpha_2+\gamma+\theta=\lambda_S+q \label{eq : B} \\
    & c(A)= \beta_2+\alpha_1+\gamma+\theta=\lambda_S+p \label{eq : A} 
\end{align}
Adding equations \ref{eq : A} and \ref{eq : A cap B} and adding equations \ref{eq : A cap B} and \ref{eq : A cup B}, we arrive at the following two equations. 
    \begin{align}
    & 2\alpha_1+\beta_2+\alpha_2+\gamma+2\theta=2\lambda_S+i+p \label{eq : A plus A cap B} \\
    & \alpha_1+\alpha_2+\beta_1+\beta_2+2\theta=2\lambda_S+i+j \label{eq : A cap B plus A cup B}
\end{align}
Subtracting Equation \ref{eq : A cap B plus A cup B} from Equation \ref{eq : A plus A cap B}, we get the following. 
   \begin{align}
    & \alpha_1-\beta_1+\gamma=p-j
    \label{eq : first one for gamma}
\end{align}
Similarly, the addition of equations \ref{eq : B} and \ref{eq : A cup B} gives us $2\beta_1+\beta_2+\alpha_2+\gamma+2\theta=2\lambda_S+q+j$. Subtracting equation \ref{eq : A cap B plus A cup B} from this, we arrive at the following. 
   \begin{align}
    & \beta_1-\alpha_1+\gamma=q-i \label{eq : second one for gamma}
\end{align}
Finally, the addition of equation \ref{eq : first one for gamma} and equation \ref{eq : second one for gamma} provides the following.
   \begin{align}
    & \gamma=\frac{p+q-i-j}{2} \label{eq : gamma}
\end{align}
\end{proof}
Exploiting Lemma \ref{lem : gamma value}, we now prove property (3) of Theorem \ref{thm : gen 3 star}. Let us consider the following two Steiner cuts: $D_1=C_2\cap \overline{C_3}=V_b\cup V_{ab}$ and $D_2=\overline{C_1\cap \overline{C_3}}=V_b\cup V_{bc} \cup V_{c}\cup V_{ac}\cup V_{abc} \cup V_{\Phi}$. Observe that $b\in C_2\cap \overline{C_3}$ and $c\notin C_2\cap \overline{C_3}$. By sub-modularity of cuts (Lemma \ref{submodularity of cuts}(1)), $c(C_2\cap \overline{C_3})+c(C_2\cup \overline{C_3})\le 2\lambda_S+2$. Since $c(C_2\cup \overline{C_3})\ge \lambda_S$, therefore, $c(C_2\cap \overline{C_3})\le \lambda_S+2$. So, we get $c(D_1)\le \lambda_S+2$. In a similar way, we can show that $c(D_2)\le \lambda_S+2$. 

For cuts $D_1$ and $D_2$, $D_1\cap D_2$ is $V_b$ and $D_1\cup D_2$ is $V_b \cup V_c \cup V_{ab} \cup V_{bc} \cup V_{ac}\cup V_{abc} \cup V_{\Phi}$. Since, $b\in D_1\cap D_2$ and $a\in \overline{D_1\cup D_2}$, therefore, $D_1\cap D_2$ and $D_1\cup D_2$ are Steiner cuts. We also have $D_1\setminus D_2=V_{ab}$ and $D_2\setminus D_1=V_c\cup V_{bc}\cup V_{ac}\cup V_{abc} \cup V_{\Phi}$. It follows from Lemma \ref{lem : gamma value} that the number of edges between $D_1\setminus D_2$ and $D_2\setminus D_1$ is maximized if, in equation \ref{eq : gamma}, $p,q$ are maximized and $i,j$ are minimized. Therefore, to maximize $p$ and $q$, we consider $D_1$ and $D_2$ to be Steiner cuts of capacity $\lambda_S+2$. Also, to minimize $i$ and $j$, we consider that $D_1\cap D_2$ and $D_1\cup D_2$ are S-mincuts. So, we have $p=q=2$ and $i=j=0$. Therefore, by Lemma \ref{lem : gamma value}, we get $\gamma= \frac{2+2-0-0}{2}=2$. This ensures that the number of edges between $V_{ab}$ and $V_c\cup V_{bc}\cup V_{ac}\cup V_{abc} \cup V_{\Phi}$ is at most $2$. 
Along a similar line, we can establish that there are at most two edges between $V_{bc}$ and $V_a\cup V_{ac}\cup V_{ab}\cup V_{abc} \cup V_{\Phi}$, and at most two edges between $V_{ac}$ and $V_b\cup V_{ab}\cup V_{bc}\cup V_{abc} \cup V_{\Phi}$. This completes proof of property (3). 

\section{Sensitivity Oracle for Single Edge Insertion} \label{app : single edge insertion}
 Baswana and Pandey \cite{DBLP:conf/soda/BaswanaP22} designed an ${\mathcal O}(n)$ space data structure that, after insertion of any edge in $G$ and given a pair of Steiner vertices $s,t\in S$, can report an $S$-mincut and its capacity in ${\mathcal O}(n)$ and ${\mathcal O}(1)$ time. However, to determine if the $S$-mincut has increased, we have to query if the $S$-mincut has increased between every pair of vertices, which requires ${\mathcal O}(|S|^2)$ time.

In this section, we show that Skeleton ${\mathcal H}_S$, projection mapping $\pi$, and Quotient mapping $\phi$ are sufficient to design an ${\mathcal O}(n)$ space data structure for reporting an $S$-mincut after the insertion of any edge $e=(x,y)$ in $G$. 

Let $P$ be a path in the Skeleton ${\mathcal H}_S$ such that there is no path $P'$ with $P$ as a proper subpath of $P'$. Suppose the two endpoints of path $P$ are node $N$ and node $N'$ of ${\mathcal H}_S$. Let $e_1$ and $e_2$ be the pair of edges of ${\mathcal H}_S$ such that $e_1$ is adjacent to $N$ and $e_2$ is adjacent to $N'$ in path $P$. Now, there are two possible cases -- (a) every node of the Skeleton belongs to path $P$ and (b) there is at least one node $M$ that does not belong to path $P$. \\

\noindent
\textbf{Case (a):} By Lemma \ref{lem : $S$-mincuts of a bunch increases}, the capacity of each $S$-mincut belonging to the bunch represented by every minimal cut of path $P$ if and only if $x\in C(N,e_1)$ and $y\in C(N',e_2)$ or vice versa. So, by using quotient mapping $\phi$ and projection mapping $\pi$, we can verify this condition in ${\mathcal O}(1)$ time. To report an $S$-mincut, using Lemma \ref{lem : tight cut reporting}, we can report $C(N,e_1)$ in ${\mathcal O}(n)$ time.  \\

\noindent
\textbf{Case (b):} Let us consider a minimal cut $\tilde C$ of Skeleton ${\mathcal H}_S$ such that $N,N'\in \tilde C$ and $M\notin \tilde C$. Let ${\mathcal B}$ be the corresponding bunch and $C$ be an $S$-mincut in ${\mathcal B}$. Let us now establish the following lemma.
\begin{lemma} \label{lem : insertion of an edge}
    Upon insertion of edge $(x,y)$, the capacity of $S$-mincut remains the same in Case (b). 
\end{lemma}
\begin{proof}
    Assume to the contrary that the capacity of $S$-mincut increases by $1$. Therefore, edge $(x,y)$ contributes to each $S$-mincut belonging to the bunch corresponding to every minimal cut of path $P$. By Lemma \ref{lem : $S$-mincuts of a bunch increases}, $x\in C(N,e_1)$ and $y\in C(N',e_2)$ or vice versa. Moreover, every $S$-mincut belonging to bunch ${\mathcal B}$ also increases. So, either $x\in C(S_{\mathcal B})$ and $y\in C(S\setminus S_{\mathcal B})$ or $y\in C(S_{\mathcal B})$ and $x\in C(S\setminus S_{\mathcal B})$. It follows that the projection of at least one endpoint of edge $e$ belongs to $\tilde C$ and also does not belong to $\tilde C$, a contradiction.   
\end{proof}
It follows from Lemma \ref{lem : insertion of an edge} that the capacity can be reported in ${\mathcal O}(1)$ time. Moreover, storing the following three tight cuts is sufficient to report the resulting $S$-mincut. These are $C(N,e_1)$, $C(N',e_2)$, and $C(M,e_t)$ (or $C(M,e_c,e_{c'})$), where $e_t$ is the tree-edge (likewise $e_c,e_{c'}$ are cycle-edges from the same cycle) adjacent to node $M$ whose removal separates $N,N'$ from $M$. Therefore, it is now easy to see that an $S$-mincut can be reported in ${\mathcal O}(n)$ time using an ${\mathcal O}(n)$ space data structure.

This completes the proof of Theorem \ref{thm : single edge Sensitivity Oracle} for handling insertion of a single edge.

\section{Proofs Pertaining to \textit{Connectivity Carcass}}
We have used a set of results that follow from the existing results \cite{DBLP:conf/soda/BaswanaP22, DBLP:conf/stoc/DinitzV94, DBLP:conf/soda/DinitzV95, DBLP:journals/siamcomp/DinitzV00}. Unfortunately, these results are not mentioned explicitly in those papers. In this section, for the sake of completeness and to eliminate any doubt about the correctness of our results, we provide proof of every such result. 

\subsection{Proof of Lemma \ref{lem : respecting nearest cuts}} \label{app : respecting nearest cuts}
\begin{proof}
    Suppose $v\in C_1$ but $v\notin C_2$. Observe that $C_1\cap C_2$ contains Steiner vertex $s_1$, vertex $u$, and does not contain vertex $v$ and Steiner vertex $s_2$. Similarly, $\overline{C_1\cup C_2}$ contains Steiner vertex $s_2$ and does not contain Steiner vertex $s_1$. This implies that both $C_1\cap C_2$ and $C_1\cup C_2$ are Steiner cuts. Also, since $C_1$ and $C_2$ are $S$-mincuts, by sub-modularity of cuts (\ref{submodularity of cuts}(1)), $c(C_1\cap C_2)+c(C_1\cup C_2)\le 2\lambda_S$. Therefore, $C_1\cap C_2$ and $C_1\cup C_2$ are $S$-mincuts. It is easy to observe that $C_1\cap C_2$ belongs to bunch ${\mathcal B}_{f_1}$ because $C_1$ belongs to ${\mathcal B}_{f_1}$ and $C_1\cap S\subseteq C_2\cap S$. We have a $S$-mincut $C_1\cap C_2$ in bunch ${\mathcal B}_{f_1}$ such that $C_1\cap C_2\subset C_1$ because $v\notin C_1\cap C_2$. So, $C_1$ is not the nearest $(s_1,s_2)$-mincut of $u$ belonging to ${\mathcal B}_{f_1}$, a contradiction. 

    Suppose $v\in C_2$, but $v\notin C_1$. We first show that if $v\in C_2$, then $u$ belongs to the nearest $(s_2,s_1)$-mincut $C_v$ of $v$ belonging to ${\mathcal B}_{f_2}$. Assume to the contrary that $C_v$ does not contain $u$. Let $C_3=\overline{C_v}$. Observe that $C_3$ cannot be a proper subset of $C_2$. Therefore, $C_2$ and $C_3$ are crossing cuts. Since $C_2$ and $C_3$ are $S$-mincuts, by sub-modularity of cuts (Lemma \ref{submodularity of cuts}(1)), we can show, in exactly the same way as $C_1\cap C_2$ in the forward direction of the proof, that $C_2\cap C_3$ is an $S$-mincut. However, since $v\notin C_2\cap C_3$, $C_2$ is not the nearest $(s_1,s_2)$-mincut of $u$ belonging to ${\mathcal B}_{f_2}$, a contradiction.\\
    Now, we consider the nearest $(s_2,s_1)$-mincut $C_v$ of $v$ that contains $u$. Assume to the contrary that $v\notin C_1$. By considering $S$-mincut $C_v$ and $S$-mincut $C_1$, in a similar way to the proof of the forward direction, we can show that there is a $S$-mincut $\overline{C_1\cup C_v}$ belonging to ${\mathcal B}_{f_2}$ such that $\overline{C_1\cup C_v}\subset C_v$, a contradiction.  
\end{proof}

\subsection{Proof of Lemma \ref{lem : cycle property}} \label{app : proof of cycle property}
\begin{proof}
    Let us prove the statement considering $C_k$ to be a $k$-cycle. The proof is along similar lines if $C_k$ is a $k$-junction. For pair of nodes $A$ and $B$, since $A\ne B$, there is a node $A_1$ belonging to $C_k$ and $A_1$ is lying in a path between $A$ and $B$. Observe that $C(A,e_A,e_A')\cap S$, $C(B,e_B,e_B')\cap S$, and $C(A_1,e_{A_1},e_{A_1}')\cap S$ are disjoint nonempty sets, where $e_{A_1}$ and $e_{A_1}'$ are the adjacent edges of $A'$ belonging to $C_k$. Therefore, it follows from the construction of Skeleton that $C(A,e_A,e_A')\cap S \subset \overline{C(B,e_B,e_B')} \cap S$ and also $C(B,e_B,e_B')\cap S  \subset \overline{C(A,e_A,e_A')}\cap S$.

    We assume to the contrary that there is a vertex $u$ of $G$ such that $u\in C(A,e_A,e_A')\cap C(B,e_B,e_B')$. Let $P=\overline{C(B,e_B,e_B')}$. Since $C(A,e_A,e_A')\cap S\subset P\cap S$, $P$ does not subdivide the Steiner set $C(A,e_A,e_A')\cap S$. By sub-modularity of cuts (Lemma \ref{submodularity of cuts}(1)), $c(C(A,e_A,e_A')\cap P)+c(C(A,e_A,e_A')\cup P)\le 2\lambda_S$. Also, since $C(A,e_A,e_A')\cap S\subset P\cap S$ and $P$ is a Steiner cut, it is easy to observe that both $C(A,e_A,e_A')\cap P$ and $C(A,e_A,e_A')\cup P$ are Steiner cuts. It follows that $C(A,e_A,e_A')\cap P$ is an $S$-mincut and belongs to the same bunch defined by the pair of edges $e_A,e_A'$. So, we have a $S$-mincut $C(A,e_A,e_A')\cap P$ which is a proper subset of $C(A,e_A,e_A')$ from the same bunch, a contradiction due to the definition of tight cuts.
\end{proof}

\subsection{Proof of Theorem \ref{thm : reporting $S$-mincut using skeleton and projection mapping} (An S-mincut Separating a Pair of Vertices)} \label{app : proof of theorem 58 for reporting S mincut separating vertices}
Let $u$ and $v$ be any pair of vertices in $G$. Determining whether there is an $S$-mincut separates $u$ and $v$ can be done in ${\mathcal O}(1)$ time by storing only the ${\mathcal O}(n)$ space quotient mapping $\phi$ (refer to Fact \ref{fact : quotient mapping} in Appendix \ref{sec : extended preliminaries}). In this section, we show that there is an ${\mathcal O}(n)$ space data structure that can report an $S$-mincut separating $u$ and $v$ in ${\mathcal O}(n)$ time if exists.
The data structure was designed by Baswana and Pandey \cite{DBLP:conf/soda/BaswanaP22} for answering single edge Sensitivity Oracle for Steiner mincuts. To report an $S$-mincut separating $u$ and $v$, we use Skeleton ${\mathcal H}(G)$, projection mapping $\pi(G)$, and an ordering of nonSteiner vertices $\tau(G)$. They occupy overall ${\mathcal O}(n)$ space. We begin by defining the following concept.
\begin{definition} [Minimal cut of Skeleton separating projections of a pair of units]
    A minimal cut $\tilde C$ of the Skeleton is said to separate projections $<M_1,N_1>$ and $<M_2,N_2>$ of a pair of units $\mu_1$ and $\mu_2$ if $M_1,N_1\in \tilde C$ and $M_2,N_2\notin \tilde C$.    
\end{definition}
 Let $\tilde C$ be a minimal cut that separates $\pi(\mu)$ and $\pi(\nu)$ for a pair of units $\mu,\nu$ of the Flesh graph. Let ${\mathcal H}_S^1$ and ${\mathcal H}_S^2$ be the two subgraphs of Skeleton ${\mathcal H}_S$ is formed after the removal of the edge-set of $\tilde C$. It follows from Lemma \ref{lem : projection mapping} that projection of any unit of Flesh is a proper path in ${\mathcal H}_S$. Therefore, $\pi(\mu)$ is a subgraph of ${\mathcal H}_S^1$ and $\pi(\nu)$ is a subgraph of ${\mathcal H}_S^2$ or vice versa.
 
 Depending on the projection mappings of both units $\mu=\phi(u)$ and $\nu=\phi(v)$, there are two possible cases that arise -- (a) $\pi(\mu)$ is disjoint from $\pi(\nu)$ and (b) $\pi(\mu)$ has at least one node in common with $\pi(\nu)$. \\

\noindent 
\textbf{Case (a):} Let $e=(A,B)$ be an edge of the Skeleton such that a minimal cut $\tilde C$ containing edge $e$ separates $\pi(\mu)$ from $\pi(\nu)$. It follows from Lemma \ref{lem : lca queries on skeleton} that such an edge $e$ of the Skeleton can be reported in ${\mathcal O}(1)$ time. Therefore, to report an $S$-mincut separating a pair of vertices, it is sufficient to report one of the two tight cuts $C(A,e)$ or $C(B,e)$. It follows from Lemma \ref{lem : tight cut reporting} that this procedure takes ${\mathcal O}(n)$ time. \\

\noindent
\textbf{Case (b):} Let us consider two possible scenarios of this case -- (1) $\pi(\mu)$ is not same as $\pi(\nu)$ and (2) $\pi(\mu)$ is the same as $\pi(\nu)$.\\

\noindent
\textbf{Case (b.1):} Suppose $\pi(\mu)=<M_1,N_1>$ is not the same as $\pi(\nu)=<M_2,N_2>$. Let us consider the proper path $\pi(\mu)$. Without loss of generality, assume that $\pi(\mu)$ is not a proper subset of $\pi(\nu)$; otherwise, consider $\pi(\nu)$. In this case, there must exist at least one node of $\pi(\nu)$, say $N$, that appears in the proper path $\pi(\mu)$. Since $\pi(\mu)$ is not the same as $\pi(\nu)$, there is at least one edge $e=(N,N')$ of $\pi(\mu)$ such that $e\in \pi(\mu)$ but $e\notin \pi(\nu)$. By Lemma \ref{lem : lca queries on skeleton}, required edge $e$ can be obtained in ${\mathcal O}(1)$ time. Observe that either $M_2,N_2\in C(N,e)$ or $M_2,N_2\in C(N',e)$. It can be determined quite easily in ${\mathcal O}(1)$ time by using Lemma \ref{lem : projection mapping}.
So, if $M_2,N_2\in C(N,e)$, then report $C(N,e)$ in ${\mathcal O}(n)$ time using Lemma \ref{lem : tight cut reporting}; otherwise report $C(N',e)$.\\

\noindent
\textbf{Case (b.2):} The main challenge is in handling this case when $\pi(\mu)$ is the same as $\pi(\nu)$. We need the following concept of extendibility of two proper paths.
\begin{definition} [Definition 2.5 in \cite{DBLP:conf/soda/BaswanaP22})]
    Let $P_1$ and $P_2$ be a pair of proper paths in ${\mathcal H}_S$ such that there is an edge $e=(N,N')$ that belongs to both $P_1$ and $P_2$. Path $P_1$ is said to be extendable to $P_2$ in the direction from $N$ to $N'$ if there is a proper path $P$ with $P_1$ as its prefix and $P_2$ as its suffix. 
\end{definition}
Let $e=(N,N')$ be any tree-edge in ${\mathcal H}(G)$ (analysis is similar for a pair of cycle-edges). Let ${\mathcal B}$ be the bunch corresponding to a minimal cut of ${\mathcal H}(G)$ containing edge $e$ (consider any one minimal cut containing edge $e$ if $e$ is a cycle-edge). Let us construct a graph $G(e)$ from $G$ as follows. Contract all vertices belonging to $C(S_{\mathcal B})$ into a single vertex $s$. Similarly, contract all vertices belong to $C(S\setminus S_{\mathcal B})$
into a single vertex $t$. Finally, remove all the self-loops. The following lemma holds immediately from the construction of graph $G(e)$.
\begin{lemma} \label{lem : G(e) lemma}
    A cut $C$ in $G$ is an $S$-mincut in $G$ belonging to ${\mathcal B}$ if and only if $C$ is an $(s,t)$-mincut of $G(e)$.
\end{lemma}
We now construct the DAG ${\mathcal D}_{PQ}(G(e))$ of Picard and Queyranne \cite{DBLP:journals/mp/PicardQ80} for graph $G(e)$, as discussed in Appendix \ref{sec : lower bound oracle}. Exploiting Lemma \ref{lem : characterization of s,t mincut}, Baswana and Pandey showed that an ${\mathcal O}(n)$ space topological ordering of ${\mathcal D}_{PQ}(G(e))$ is sufficient to report an $(s,t)$-mincut separating any pair of vertices in $G(e)$ (refer to Lemma 3.1 in \cite{DBLP:conf/soda/BaswanaP22} and a detailed proof is also provided in Section 4.1 in \cite{DBLP:journals/talg/BaswanaBP23} for directed graphs).  Unfortunately, storing a topological ordering for every tree-edge or every pair of cycle-edges from the same cycle of Skeleton would occupy ${\mathcal O}(|S|^2n)$ space in the worst case. We now explain the data structure of Baswana and Pandey \cite{DBLP:conf/soda/BaswanaP22} that occupies only ${\mathcal O}(n)$ space \cite{DBLP:conf/soda/BaswanaP22}.\\

\noindent
\textbf{Description of the data structure:} The set of all Stretched units of Flesh graph is partitioned into groups as follows. A pair of Stretched units $\mu_1$ and $\mu_2$ belong to the same group if $\pi(\mu_1)=\pi(\mu_2)$. Let ${\mathcal G}$ be the set of all obtained groups. For any stretched unit $\mu_1$ belonging to any group $\alpha\in {\mathcal G}$, the proper path to which $\mu$ is projected is denoted by $P_{\alpha}$. Let $e=(M,M')$ be any tree-edge (the analysis is similar if $e$ is a cycle-edge) belonging to $P_{\alpha}$ for any group $\alpha\in {\mathcal G}$. Let $\tau_e$ be a topological ordering of the DAG ${\mathcal D}_{PQ}(G(e))$. The Stretched units belonging to group $\alpha$ are stored in the same order as they appear in $\tau_e$. Also, the direction of order is stored, that is, whether they are stored in the direction from $M$ to $M'$ or from $M'$ to $M$. So, ordering $\tau(G)$ is the set of orderings of all Stretched units from every group in ${\mathcal G}$. Since the set of vertices belonging to a group in ${\mathcal G}$ is disjoint from any another group in ${\mathcal G}$, therefore, the space occupied by the data structure is ${\mathcal O}(n)$. \\

\noindent
\textbf{Reporting an $S$-mincut separating $u,v$ in case (b.2):} We know that $\pi(\mu)$ is the same as $\pi(\nu)$. Therefore, they appear in the same group $\alpha\in {\mathcal G}$. Suppose the ordering of Stretched units belonging to $\alpha$ is done using a tree-edge $e=(M,M')$ (the analysis is the same if $e$ is a cycle-edge) belonging to $P_{\alpha}$. We construct two sets of vertices $\Gamma_1$ and $\Gamma_2$ as follows. Without loss of generality, assume that Stretched units in $\alpha$ are stored in the direction from $M$ to $M'$ and $\mu$ appears before $\nu$ in the ordering $\tau_e$. Set $\Gamma_1$ contains every vertex $x$ such that $\phi(x)\in \alpha$ and $\phi(x)$ appears before $\mu$ in the ordering $\tau_e$. Set $\Gamma_2$ contains every vertex $x$ such that $\pi(\mu)$ is extendable to $\pi(\phi(x))$ in the direction from $M'$ to $M$. The following lemma is established in \cite{DBLP:conf/soda/BaswanaP22}.
\begin{lemma} [Lemma 3.3 in \cite{DBLP:conf/soda/BaswanaP22}] \label{lem : lemma to report s mincut separating pair of vertices}
    The set of vertices belonging to $C(M,e)\cup \Gamma_1 \cup \Gamma_2$ defines an $S$-mincut that contains vertex $u$ and does not contain vertex $v$.  
\end{lemma}
Exploiting Lemma \ref{lem : projection mapping} and Lemma \ref{lem : tight cut reporting}, it is easy to observe that the set $C(M,e)\cup \Gamma_1 \cup \Gamma_2$ can be reported in ${\mathcal O}(n)$ time. So, it follows from Lemma \ref{lem : lemma to report s mincut separating pair of vertices} that there is an ${\mathcal O}(n)$ space data structure that can report an $S$-mincut separating any given pair of vertices in ${\mathcal O}(n)$ time. This completes the proof of Theorem \ref{thm : reporting $S$-mincut using skeleton and projection mapping}. 

Suppose there is an edge between $u$ and $v$. By Theorem \ref{thm : reporting $S$-mincut using skeleton and projection mapping}, we can report an $S$-mincut in which edge $(u,v)$ is contributing in ${\mathcal O}(n)$ time. So, for handling the failure of an edge, Theorem \ref{thm : single edge Sensitivity Oracle} is a simple corollary of this result.

\end{document}